\documentclass{article}
\usepackage{neurodata} 
\usepackage{amsmath,amsthm,amsopn,amsfonts,pdfpages,dsfont,amssymb}
\usepackage{url}
\usepackage{bbm,subcaption}
\usepackage[colorlinks=true,pagebackref,linkcolor=magenta, citecolor=gray]{hyperref}
\usepackage{stackengine}
\usepackage[shortlabels]{enumitem}
\usepackage{algorithmicx}
\usepackage{algorithm}
\usepackage{natbib}

\newcommand{\jesus}[1]{}

\newcommand{\e}{\mathbb{E}}
\newcommand{\p}{\mathbb{P}}
\newcommand{\real}{\mathbb{R}}

\newtheorem{theorem}{Theorem}
\newtheorem{assumption}{Assumption}
\newtheorem{proposition}[theorem]{Proposition}
\newtheorem{lemma}[theorem]{Lemma}

\newtheorem{corollary}[theorem]{Corollary}
\theoremstyle{definition}
\newtheorem{definition}{Definition}
\theoremstyle{remark}

\newcommand{\bA}{\mathbf{A}}
\newcommand{\bB}{\mathbf{B}}
\newcommand{\bC}{\mathbf{C}}
\newcommand{\bD}{\mathbf{D}}
\newcommand{\bE}{\mathbf{E}}
\newcommand{\bF}{\mathbf{F}}
\newcommand{\bG}{\mathbf{G}}
\newcommand{\bH}{\mathbf{H}}
\newcommand{\bI}{\mathbf{I}}
\newcommand{\bM}{\mathbf{M}}
\newcommand{\bN}{\mathbf{N}}
\newcommand{\bV}{\mathbf{V}}
\newcommand{\bU}{\mathbf{U}}
\newcommand{\bX}{\mathbf{X}}
\newcommand{\bP}{\mathbf{P}}
\newcommand{\bQ}{\mathbf{Q}}
\newcommand{\bR}{\mathbf{R}}
\newcommand{\bS}{\mathbf{S}}
\newcommand{\bs}{\mathbf{s}}
\newcommand{\bW}{\mathbf{W}}
\newcommand{\by}{\mathbf{y}}
\newcommand{\bZ}{\mathbf{Z}}
\newcommand{\bPi}{\mathbf{\Pi}}
\newcommand{\bLambda}{\mathbf{\Lambda}}
\newcommand{\bXi}{\mathbf{\Xi}}
\newcommand{\bSigma}{\mathbf{\Sigma}}

\algnewcommand\algorithmicinput{\textbf{Input:}}
\algnewcommand\Input{\item[\algorithmicinput]}
\algnewcommand\algorithmiciterate{\textbf{Iterate:}}
\algnewcommand\Iterate{\item[\algorithmiciterate]}
\algnewcommand\algorithmicinitialize{\textbf{Initialize:}}
\algnewcommand\Initialize{\item[\algorithmicinitialize]}
\algnewcommand\algorithmicoutput{\textbf{Output:}}
\algnewcommand\Output{\item[\algorithmicoutput]}
\algnewcommand\RETURN{\State \algorithmicreturn}%

\DeclareMathOperator*{\argmin}{argmin}
\newcommand{\rdpg}[1]{\operatorname{RDPG}\left(#1\right)}
\newcommand{\sbm}[1]{\operatorname{SBM}\left(#1\right)}
\newcommand{\cosie}[1]{\operatorname{COSIE}\left(#1\right)}

\title{
Inference for Multiple Heterogeneous Networks with a Common Invariant Subspace
}

\usepackage{authblk}

\author[1]{Jes\'us Arroyo\thanks{\href{mailto:jesus.arroyo@jhu.edu}{jesus.arroyo@jhu.edu}}}
\author[2]{Avanti Athreya}
\author[3]{Joshua Cape}
\author[2]{Guodong Chen}
\author[2]{Carey E.~Priebe}
\author[4]{Joshua T.~Vogelstein\thanks{\href{mailto:jovo@jhu.edu}{jovo@jhu.edu}}}

\affil[1]{Center for Imaging Science, Johns Hopkins University}
\affil[2]{Department of Applied Mathematics and Statistics, Johns Hopkins University}
\affil[3]{Department of Statistics, University of Pittsburgh}
\affil[4]{Department of Biomedical Engineering, Kavli Neuroscience Discovery Institute, Johns Hopkins University}

\begin{document}
	
	\maketitle

\begin{abstract}
{The development of models and methodology for the analysis of data from multiple heterogeneous networks is of importance both in statistical network theory and across a wide spectrum of application domains. Although single-graph inference is well-studied, multiple graph inference is largely unexplored, in part because of the challenges inherent in appropriately modeling graph differences and yet retaining sufficient model simplicity to render estimation feasible. This paper addresses exactly this gap, by introducing a new model, the common subspace independent-edge (COSIE) multiple random graph model, which describes a heterogeneous collection of networks with a shared latent structure on the vertices but potentially different connectivity patterns for each graph. The COSIE model encompasses many popular network representations, including the stochastic blockmodel. The COSIE model is both flexible enough to meaningfully account for important graph differences and tractable enough to allow for accurate inference in multiple networks. In particular, a joint spectral embedding of adjacency matrices---the multiple adjacency spectral embedding (MASE)---leads, in a COSIE model, to simultaneous consistent estimation of underlying parameters for each graph. Under mild additional assumptions, MASE estimates satisfy asymptotic normality and yield improvements for graph eigenvalue estimation. In both simulated and real data, the COSIE model and the MASE embedding can be deployed for a number of subsequent network inference tasks, including dimensionality reduction, classification, hypothesis testing and community detection. Specifically, when MASE is applied to a dataset of connectomes constructed through diffusion magnetic resonance imaging, the result is an accurate classification of brain scans by human subject and a meaningful determination of heterogeneity across scans of different subjects.}
\end{abstract}

	\section{Introduction}	
	
	Random graph inference has witnessed extraordinary growth over the last several decades \citep{Goldenberg2009,Kolaczyk2017,Abbe2017}. To date, the majority of work has focused primarily on inference for a single random graph, leaving largely unaddressed the problem of modeling the structure of multiple-graph data, including multi-layered and time-varying graphs \citep{Kivela2014,Boccaletti2014,Holme2012}. Such multiple graph data arises naturally in a wide swath of application domains, including neuroscience \citep{Bullmore2009}, biology \citep{Han2004,Li2011}, and the social sciences \citep{Szell2010,Mucha2010}. 
	
	Several existing models for multiple-graph data require strong assumptions that limit their flexibility \citep{Holland1983,Wang2017,Nielsen2018} or scalability to the size of real-world networks \citep{Durante2017,Durante2018}. Principled approaches to multiple-graph inference \citep{Tang2014,Tang2017,Levin2017,Ginestet2017,Arroyo-Relion2017,Kim2019}  reflect this challenge, in part precisely because of the challenges inherent in constructing a multiple-graph model that adequately captures real-world graph heterogeneity while remaining analytically tractable. The aim, here, is to address exactly this gap.
	
	In this paper, we resolve the following questions: first, can we construct a simple, flexible multiple random network model that can be used to approximate real-world data? Second, for such a model, can we leverage common network structure for accurate, scalable estimation of model parameters, while also being able to describe distributional properties of our estimates? Third, how can we use such estimators in subsequent inference tasks---for instance, multiple-network testing, community detection, or graph classification?
Fourth, how well do such modeling, estimation, and testing procedures perform on simulated and real data, in comparison to current state-of-the-art techniques?

    Towards this end, we present a semiparametric model for multiple graphs with aligned vertices that is based on a common subspace structure between their expected adjacency matrices, but with allowance for heterogeneity within and across the graphs. The common subspace structure has a meaningful interpretation and generalizes several existing models for multiple networks; moreover, the estimation of such a common subspace in a set of networks is an inherent part of well-known graph inference problems like community detection, graph classification, or eigenvalue estimation.

The paper is organized as follows. We start by briefly encapsulating the principal contributions of this paper, and give an overview of the related work. In Section~\ref{sec:model} and Section~\ref{sec:fit}, we give a formal treatment of the COSIE model and MASE procedure. In Section~\ref{sec:theory}, we study the theoretical statistical performance of MASE and provide a bound on the error of estimating the common subspace, and we establish asymptotic normality of the estimates of the individual graph parameters. In Section~\ref{sec:simulations}, we investigate the empirical performance of MASE for estimation and testing when compared to other benchmark procedures. In Section~\ref{sec:data}, we use MASE to analyze connectomic networks of human brain scans. Specifically, we show the ability of our model and estimation procedure to characterize and discern the differences in brain connectivity across subjects. In Section~\ref{sec:concl}, we conclude with a discussion of open problems for future work.

\subsection{Summary of contributions}

The model describes random graphs $G^{(1)}, \cdots, G^{(m)}$ with $n$ labeled vertices whose Bernoulli adjacency matrices $\bA^{(i)}\in\{0,1\}^{n\times n}, 1 \leq i \leq m$, have expectation of the form $\bV \bR^{(i)} \bV^\top$, where $\bR^{(i)}\in\real^{d\times d}$ is a low-dimensional matrix and $\bV\in\real^{n\times d}$ is a matrix with orthonormal columns. Because the invariant subspaces defined by $\bV$ are common to all the graphs, we call this model the common-subspace independent edge (COSIE) random graph model. The $\bR^{(i)}$ matrices are dimension $d \times d$, need not be diagonal, and can vary with each graph. In particular, despite the shared expectation matrix factor $\bV$, each graph $G^{(i)}$ can have a different distribution. 

	Now, to estimate $\bV$ and $\bR^{(i)}$, we rely on the low-rank structure of the expectation matrices in COSIE. Specifically, we first spectrally decompose each of the adjacency matrices $\bA^{(i)}$ to obtain the \emph{adjacency spectral embedding} \citep{Sussman2012} $\widehat{\bX}^{(i)}$, defined as $\widehat{\bX}^{(i)}= \widehat{\bV}^{(i)} |\widehat{\bD}^{(i)}|^{1/2}$, where $\widehat{\bD}^{(i)}$ is the $d \times d$ diagonal matrix of the top $d$ eigenvectors of $\bA^{(i)}$, sorted by magnitude, and $\widehat{\bV}^{(i)}$ the $n \times d$ matrix of associated orthonormal eigenvectors. In the COSIE model, we will leverage the common subspace structure and use these eigenvector estimates $\widehat{\bV}^{(i)}$ to obtain an improved estimate of the true common subspace $\bV$. In fact, we use $\widehat{\bV}^{(i)}$ to build the \emph{multiple adjacency spectral embedding} (MASE), as follows. We define the $n \times m d$ matrix $\widehat{\bU}$ by $\widehat{\bU}=(\widehat{\bV}^{(1)},\cdots,\widehat{\bV}^{(m)})$ 
	and we let $\widehat{\bV}$ be the matrix of the top $d$ leading left singular vectors of $\widehat{\bU}$. Figure~\ref{fig:alg1} provides a graphical illustration of the MASE procedure for a collection of stochastic blockmodels. We set $\widehat{\bR}^{(i)} = \widehat{\bV}^\top  \bA ^{(i)} \widehat{\bV}$.
The MASE algorithm, then, outputs $\widehat{\bV}$ and $\widehat{\bR}^{(i)},$ for $ 1 \leq i \leq m$.

	\begin{figure}
        \centering
        \includegraphics[width=\textwidth]{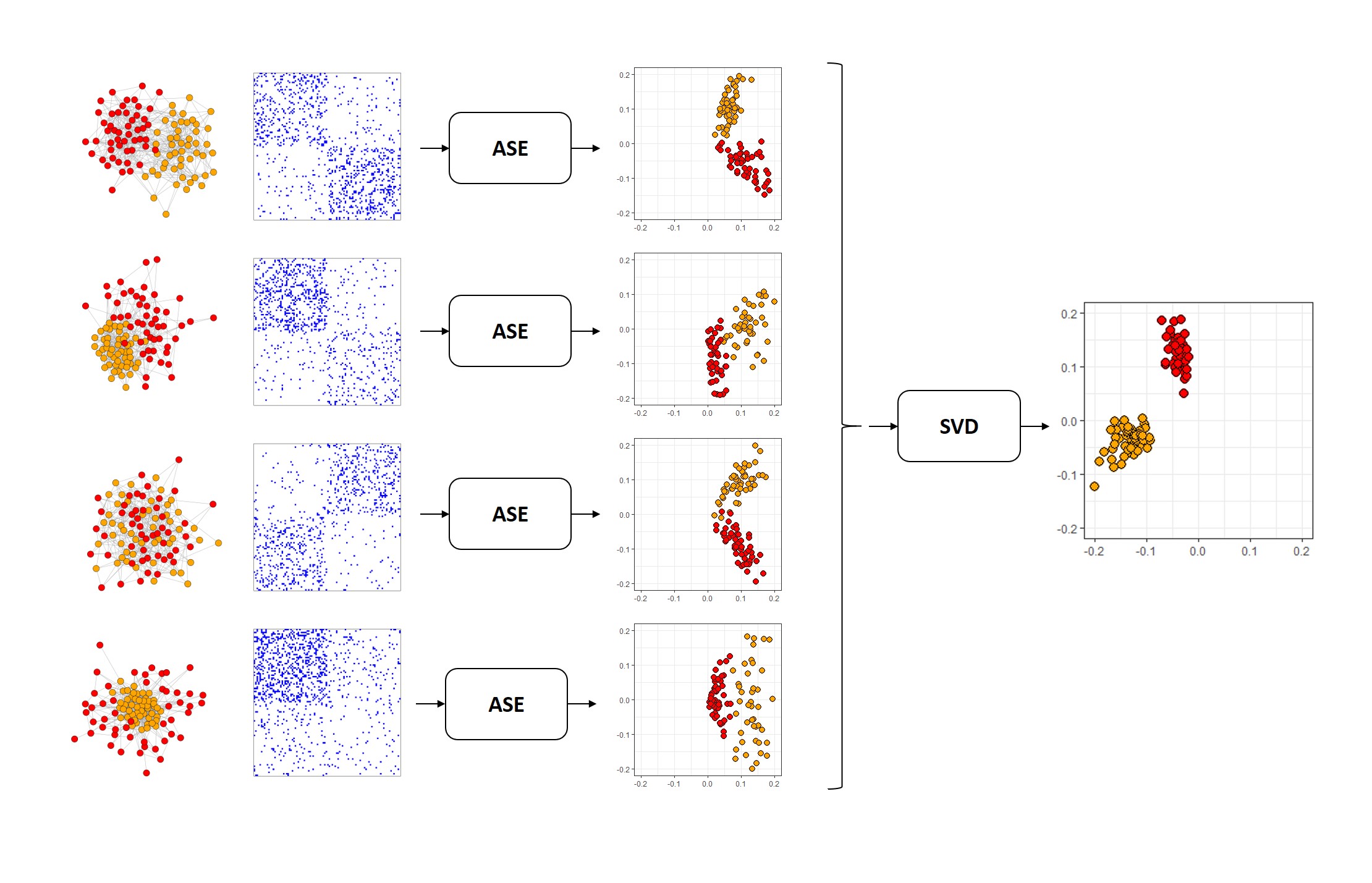}
        \caption{Graphical representation of Algorithm \ref{alg:mase} for estimating the common invariant subspace $\bV$ in a multilayer stochastic blockmodel (see Definition~\ref{definition:multilayerSBM}).}
        \label{fig:alg1}
    \end{figure}


The main contributions of the paper are listed as follows.

\begin{enumerate}
    \item We introduce the \emph{common subspace independent edge} (COSIE) model for multiple graphs. COSIE is a flexible, tractable model for random graphs that retains enough homogeneity---via the common subspace---for ease of estimation, and permits sufficient heterogeneity---via the potentially distinct score matrices $\bR^{(i)}$---to model important collections of different graphs. The COSIE model is appropriate for modeling a variety of real network phenomena. Indeed, the ubiquitous SBM is COSIE, and COSIE encompasses the important case of a collection of vertex-aligned SBMs \citep{Holland1983} whose block assignments are the same but whose block probability matrices may differ (Proposition~\ref{prop:sbm-is-cosie}). COSIE also provides a generalization to the multiple-graph setting of latent position models such as the \emph{random dot product graph} \citep{Young2007}.

    \item  The estimates obtained by MASE provide consistent estimates for the common subspace $\bV$ and the score matrices $\bR^{(i)}$. In particular, Theorem~\ref{thm:V-Vhat} shows that  under mild assumptions on the sparsity of the graphs, there exists a constant $C$ such that for all $n \geq n_0$, we have 
	\begin{equation*}
		\mathbb{E}\left[\min_{\bW\in\mathcal{O}_d}\|\widehat{\bV}-\bV\bW\|_F\right]  \leq C\left(\sqrt{\frac{1}{nm}} + \frac{1}{n}\right), 
	\end{equation*}
 	where $\mathcal{O}_d$ represents the set of $d \times d$ orthogonal matrices. This result employs bounds introduced by \cite{Fan2017} for subspace estimation in distributed covariance matrices.	Furthermore, under delocalization assumptions on $\bV$, we can show that the MASE estimates $\widehat{\bR}^{(i)}$ are asymptotically normal as the size of the graph $n$ increases, with a bias term that decays as $m$ increases. Namely, there exists a sequence of orthogonal matrices $\bW$ such that
	\begin{equation*}
	    \frac{1}{\sigma_{i,k,l}}(\widehat{\bR}^{(i)} - \bW^\top\bR^{(i)}\bW  + \bH)_{kl} \overset{d}{\longrightarrow} \mathcal{N}(0, 1), 
	\end{equation*}
    as $n \rightarrow \infty$, where $\bH$ is a random matrix that satisfies $\e[\|\bH\|_F]=O(\frac{d}{\sqrt{m}})$,  $\sigma^2_{i,k,l}=O(1)$, and $\mathcal{N}(0, 1)$ is the standard normal distribution. The norm of the score matrices typically grows with the size of the graphs, and in particular, our method requires that $\|\bR^{(i)}\|_F=\omega(\log n)$. 
    The asymptotic distribution result then
    shows that the estimates of the score matrices obtained by MASE are accurate, allowing their use in subsequent inference tasks. These results and simulation evidence (see Figure~\ref{fig:eigs-m-hist} in Section~\ref{sec:simulations}) also suggest that the multiplicity of graphs in a COSIE can  reduce the bias in eigenvalue estimation as observed in \cite{Tang2018b}.
   
    \item  The  embedding method MASE can be successfully deployed for a number of subsequent network inference tasks, including dimensionality reduction, classification, hypothesis testing and community detection, and it performs competitively with respect to (and in some cases significantly better than) a number of existing procedures. In Section~\ref{sec:community-detection} we show the ability of the method to perform community detection in multilayer stochastic blockmodels.
    Section \ref{sec:simulations}  examines the empirical performance of MASE in a number  of  graph inference tasks. In particular, hypothesis testing on populations of graphs is a relatively nascent area in which methodologically sound procedures for graph comparison are scarce. Here, we demonstrate a preview (with more details in Section \ref{sec:simulations}) of the transformative impact MASE can have on graph hypothesis testing.  COSIE also provides  foundations for the critical (and wide-open) question of multi-sample graph hypothesis testing, whereby we can use these estimates for $\bV$ and $\bR^{(i)}$ to build principled tests to compare different populations of networks.

     \item COSIE and MASE are robust to the challenges of real data, and provide domain-relevant insights for real data analysis. In particular, COSIE and MASE can be deployed effectively in a pressing real-data problem: that of identifying differences in brain connectivity in medical imaging on human subjects. In Section~\ref{sec:data}, we use COSIE to model a collection of brain networks, specifically data from the HNU1 study \citep{Zuo2014} consisting of diffusion magnetic resonance imaging (dMRI), and show that the MASE procedure elucidates differences across heterogeneous networks but also recognizes important similarities, and provides a rigorous statistical framework within which biologically relevant differences can be assessed.

\end{enumerate}

\subsection{Related work}

Latent space approaches for the multiple graph data setting have been presented before, but they either tend to limit the heterogeneity in the distributions across the graphs, or include a large number of parameters, which complicate the scalability and interpretation. \citet{Levin2017} presents a method to estimate the latent positions for a set of graphs. Although the method can in principle obtain different latent positions for each graph that do not require any further Procrustes alignment, the method is only studied under a joint RDPG model which assumes that the latent positions of all the graphs are the same. \cite{Wang2017} introduced a semiparametric model for graphs that is able to effectively handle heterogeneous distributions. However, their model usually requires a larger number of parameters to represent the same distributions than COSIE, which complicates the  interpretation, and the non-convexity of the problem makes estimation more difficult. In fact, to represent a multilayer SBM with $K$ communities this model requires to embed the graphs in a latent space of dimension $O(K^2)$, compared to $O(K)$ for COSIE (see Proposition~\ref{prop:sbm-is-cosie}). Tensor factorizations  \citep{zhang2018tensor} provide a way to model the structure of a collection of matrices, but current  approaches  that are based on tensor decompositions or other matrix factorizations \citep{Zhang2019,wang2019common,Wang2019}
present challenges on estimation due to the non-convexity of the problem, and interpretation because of the large number of parameters. In particular,  the model introduced in \citep{wang2019common} has a similar factorization structure to COSIE, but the non-convexity of the likelihood complicates the tractability of the model. More recently, \cite{Nielsen2018} studied a multiple RDPG model that is a special case of the model of \cite{Wang2017} by imposing further identifiability constraints.  These constraints limit the ability of their model to represent heterogeneous distributions, and, in fact, it is not possible to obtain equivalent statements to Proposition~\ref{prop:sbm-is-cosie} under this model.  
Bayesian formulations of this problem have also been introduced \citep{Durante2017, Durante2018}, but computational methods to fit these models limit their applicability to much smaller graphs, with vertices numbering only in the dozens. The COSIE model, on the other hand, is both flexible enough to account for important  differences in a collection of heterogeneous graphs, while being tractable enough to allow for accurate and principled inference.

In the single graph setting, \cite{Sussman2012, Lyzinski2014, Athreya2016}  showed that the ASE of a graph, $\widehat{\bX}$, is a consistent, asymptotically normal estimate of underlying graph parameters for the RDPG model.  In \cite{Levin2017}, it is shown that such a spectral embedding can be profitably deployed for estimation and testing in multiple independent graphs, when all the graphs are sampled from the same connection probability matrix $\bP$. Here, the COSIE model provides the ground for studying a multiple heterogeneous graph setting, in which each graph $\bA^{(i)}$ is sampled from a different connection probability $\bP^{(i)}$.
	
   The method for estimating the common invariant subspace is related to other similar methods in the literature.  \cite{Crainiceanu2011} proposed a population singular value decomposition method for representing a sample of rectangular matrices with the same dimensions, and use it to study arrays of images. In a different work, \cite{Fan2017} introduced a method for estimating the principal components of distributed data. Their method computes the leading eigenvectors of the covariance matrix on each server, and then obtains the leading eigenvectors of the average of the subspace projections, which is shown to converge to the solution of PCA with the same rate as if the data were not distributed. Both methods correspond to the unscaled ASE version of MASE. In our case, we show that this particular way of estimating the parameters of the COSIE model results in a consistent estimator of the common invariant subspace, and asymptotically normal estimators of the individual parameters.
   
   The MASE algorithm also provides a simple extension of spectral clustering to the multiple-graph setting when all graphs share the same community structure but possibly different interconnection matrices. This community detection method is scalable and can effectively handle the heterogeneity of the graphs. Other spectral approaches for this problem rely in averaging the adjacency matrices
	\citep{Han,Bhattacharyya2018,paul2020spectral} or a transformation of them \citep{Bhattacharyya2017}, which requires stronger assumptions on the connection probabilities matrices or can increase the computational cost. Other approaches for this problem, including likelihood and modularity maximization \citep{Han,Peixoto,Mucha2010,paul2020spectral} require to solve large combinatorial problems, which usually limits their scalability.

\section{The model: Common subspace independent-edge random graphs (COSIE)\label{sec:model}}
	
	We consider a sample of $m$ observed graphs $G^{(1)},\ldots, G^{(m)}$, with $G^{(i)} = (\mathcal{V},\mathcal{E}^{(i)})$, where $\mathcal{V}=\{1,\ldots,n\}$ denotes a set of $n$ labeled vertices, and $\mathcal{E}^{(i)}\subset \mathcal{V}\times \mathcal{V}$ is the set of edges corresponding to graph $i$. Assume the vertex sets are the same (or at least aligned). Assume the graphs are undirected, with no self-loops (it is worth emphasizing however, that these results can be easily extended to allow for loops, directed or weighted graphs).
	For each graph $G^{(i)}$, denote by $\bA^{(i)}$ to the $n\times n$ adjacency matrix that represents the edges; the matrix $\bA^{(i)}$ is binary, symmetric and hollow, and $\bA^{(i)}_{uv}=1$ if $(u,v)\in \mathcal{E}^{(i)}$. 
		Assume that the above-diagonal entries $\bA^{(i)}_{uv}$, $v>u$, of the adjacency matrix for graph $i$ are independent Bernoulli random variables with $\bP_{uv}$ the probability of an edge between vertex $u$ and vertex $v$. Consolidate these probabilities in the matrix $\bP$. 
	
	To model the graphs we consider an independent-edge random graph framework, in which the edges of a graph $\bA$ are conditionally independent given a probability matrix $\bP\in[0,1]^{n\times n}$, so that each entry of $\bP_{uv}$ denotes the probability of a Bernoulli random variable representing an edge between the corresponding vertices $u$ and $v$, and hence, the probability of observing a graph is given by
	\begin{equation*}
	    \p(\bA|\bP) = \prod_{u>v}\bP_{uv}^{\bA_{uv}}(1-\bP_{uv})^{1-\bA_{uv}}.
	\end{equation*}
	We also write $\bP=\e[\bA|\bP]$. This framework 
	encompasses many popular statistical models for graphs \citep{Holland1983,Hoff2002,Bickel2009} which differ in the way the structure of the matrix $\bP$ is defined. In the single graph setting, imposing some structure in the matrix $\bP$ is necessary to make the estimation problem feasible, as there is only one graph observation. Under this framework, we can characterize the distribution of a sample of graphs $\bA^{(1)},\ldots, \bA^{(k)}$ by modeling their expectations $\bP^{(1)}, \ldots, \bP^{(m)}$. As our goal is to model a heterogenous population of graphs with possible different distributions, we do not assume that the expected matrices are all equal, so as in the single graph setting, further assumptions on the structure of these matrices are necessary.
	
	To introduce our model, we start by reviewing some models for a single graph that motivate our approach. One of such models is the \emph{random dot product graph} (RDPG) defined below.
	
	\begin{definition}[Random dot product graph model \citep{Young2007}] 
	Let $\bX=(X_1, \ldots, X_n)^\top\in\real^{n\times d}$ be a matrix such that the inner product of any two rows satisfy $0\leq X_u^\top X_v\leq 1$. We say that a random adjacency matrix $\bA$ is distributed as a random dot product graph with latent positions $\bX$, and write $\bA \sim \rdpg{\bX}$, if the conditional distribution of $\bA$ given $\bX$ is
	\[\p(\bA|\bX) = \prod_{u>v} (X_u^\top X_v)^{\bA_{uv}}(1 - X_u^\top X_v)^{1-\bA_{uv}}.\]
	\end{definition}
	
	The RDPG is a type of latent space model \citep{Hoff2002} in which the rows of $\bX$ correspond to latent positions of the vertices, and the probability of an edge between a pair of vertices is proportional to the angle between their corresponding latent vectors, which gives an appealing interpretation of the matrix of latent positions $\bX$. Under the RDPG model, the matrix of edge probabilities satisfies $\bP=\bX\bX^\top $, so the RDPG framework contains the class of independent-edge graph distributions for which $\bP$ is a positive semidefinite matrix with rank at most $d$. More recently, 
	 \cite{Rubin-Delanchy2017} introduces the generalized RDPG (GRDPG) model, which extends this framework to the whole class of matrices $\bP$ with rank at most $d$, by introducing a diagonal matrix $\bI(p,q)$ to the model, such that $\bI(p,q)$ is of size $d\times d$ and there are $p$ diagonal entries equal to $1$ and $q$ entries equal to $-1$. Given $\bX$ and $\bI(p,q)$, the edge probabilities  are modeled as $\bP=\bX\bI(p,q)\bX^\top $, which keeps the interpretation of the latent space while extending the class of graphs that can be represented within this formulation.

	The stochastic blockmodel (SBM) \citep{Holland1983} is a particular example of a RDPG in which there are only $K<n$  different rows in $\bX$, aiming to model community structure in a graph. In a SBM, vertices are partitioned into $K$ different groups, so each vertex $u\in\mathcal{V}$ has a corresponding label $z_u\in\{1,\ldots,K\}$. The probability of an edge appearing between two vertices only
	depends on their corresponding labels, and these probabilities are encoded in a matrix $\bB\in\real^{K\times K}$ such that
	\[\p(\bA_{uv}=1|\bB, z_u,z_v) =\bP_{uv}= \bB_{z_uz_v}.\]
	To write the SBM as a RDPG, let $\bZ\in\{0,1\}^{n\times K}$ be a binary matrix that denotes the community memberships of the vertices,  with $\bZ_{uk} = 1$ if $z_u=k$, and 0 otherwise. Then, the edge probability matrix of the SBM is
	\begin{equation}
	    \bP = \bZ\bB\bZ^\top  \label{eq:sbm1},
	\end{equation}
	and denote $\bA\sim\sbm{\bZ, \bB}$. Let $\bB=\bW\bD\bW^\top $ be the eigendecomposition of $\bB$. From this representation, it is easy to see that if we write $\bX = \bZ \bW|\bD|^{1/2}$, then the distribution of a SBM corresponds to a GRPDG graph with latent positions $\bX$. In particular, if $\bB$ is a positive semidefinite matrix, this is also a RDPG. Multiple extensions to the SBM have been proposed in order to produce a more realistic and flexible model, 	including degree heterogeneity \citep{Karrer2011a}, multiple community membership \citep{Airoldi2007}, or hierarchical partitions \citep{Lyzinski2017}. These extensions usually fall within the same framework of Equation~\eqref{eq:sbm1} by placing different constraints on the matrix $\bZ$, and hence they can also be studied within the RDPG model.

    In defining a model for multiple graphs, we adopt a low-rank assumption on the expected adjacency matrices, as the RDPG model and all the other special cases do. To leverage the information of multiple graphs, a common structure among them is necessary. We thus assume that all the expected adjacency matrices of the independent edge graphs share a common invariant subspace, but allow each individual matrix to be different within that subspace. 
    
	\begin{definition}[Common Subspace Independent Edge graphs]
	 Let 
	 $$\bV=(V_1,\ldots,V_n)^\top \in\real^{n\times d}$$ 
	 be a matrix with orthonormal columns, and $\bR^{(1)},\ldots,\bR^{(m)}\in\real^{d\times d}$ be symmetric matrices such that $0\leq V_u^\top \bR^{(i)}V_v\leq 1$ for all $u,v\in[n]$, $i\in[m]$. We say that the random adjacency matrices $\bA^{(1)},\ldots, \bA^{(m)}$ are jointly distributed according to the \emph{common subspace independent-edge} graph model with bounded rank $d$ and parameters $\bV$ and $\bR^{(1)},\ldots,\bR^{(m)}$ if for each $i=1,\ldots,m$, 
	 given  $\bV$ and $\bR^{(i)}$ the entries of each $\bA^{(i)}$ are independent and distributed according to
	 \[\p(\bA^{(i)}|\bV, \bR^{(i)}) = \prod_{u<v} (V_u^\top \bR^{(i)}V_v)^{\bA^{(i)}_{uv}}(1-V_u^\top \bR^{(i)}V_v)^{1-\bA^{(i)}_{uv}}.\]
	 We denote by $(\bA^{(1)},\ldots,\bA^{(m)})\sim \cosie{\bV; \bR^{(1)},\ldots, \bR^{(m)}}$ to the joint model for the adjacency matrices.\label{definition:model}
	\end{definition}
	
	Under the COSIE graph model, the expected adjacency matrices of the graphs share a common invariant subspace $\bV$ that can be interpreted, upon scaling by the score matrices, as the common joint latent positions of the $m$ graphs. This invariant subspace can only be identified up to an orthogonal transformation, so the interpretation of the latent positions is preserved. Note that each graph is marginally distributed as a GRDPG. The score matrices $\bR^{(i)}$ can be expressed as $\bR^{(i)}= \bW^{(i)}\bD^{(i)}(\bW^{(i)})^\top $ as the eigendecomposition of $\bR^{(i)}$, such that $\bW^{(i)}$ is a orthogonal matrix of size $d\times d$, and $\bD^{(i)}$ is a diagonal matrix containing the eigenvalues. Then, the corresponding latent positions of graph $i$ in the GRDPG model are given by $\bX^{(i)}=\bV\bW^{(i)}|\bD^{(i)}|^{1/2}$, with $(|\bD|)_{uv} = |\bD_{uv}|$.

The model also introduces individual score matrices $\bR^{(1)}, \ldots, \bR^{(m)}$ that control the connection probabilities between the edges of each graph. When an $\bR^{(i)}$ is diagonal, its entries contain the eigenvalues of the corresponding probability matrix $\bP^{(i)}$, but in general $\bR^{(i)}$ does not have to be diagonal. Formal identifiability conditions of the model are discussed in Section~\ref{sec:identifiability}.

The COSIE model can capture any distribution of multiple independent-edge graphs with aligned vertices when the model is equipped with a distribution for the score matrices and $d$ is sufficiently large. That is, for any set of probability matrices $\bP^{(1)}, \ldots, \bP^{(m)},$ there exist an embedding dimension $d\leq n$ such that those matrices can be represented with the COSIE model. Indeed, it is trivial to note that if $d=n$, we can set $\bV=\bI$ and $\bR^{(i)}=\bP^{(i)}$. However, for many classes of interest, the embedding dimension $d$ necessary to exactly or approximately represent the graphs is usually much smaller than $n$. In those cases, the COSIE model can effectively reduce the dimensionality of the problem from $O(mn^2)$ different parameters to only $O(nd + md^2)$.

To motivate our model, we consider the \emph{multilayer stochastic blockmodel} for multiple graphs, originally introduced by \cite{Holland1983}. In this model, the community labels of the vertices remain fixed across all the graphs in the population, but the connection probability between and within communities can be different on each graph. The parsimony and simplicity of the model, while maintining heterogeneity across the graphs, has allowed its use in different statistical tasks, including community detection \citep{Han},
multiple graph inference \citep{Pavlovic2019,Kim2019} and modelling time-varying networks \citep{Matias2016}. In brain networks, communities are usually in agreement with  functional brain systems, which are common across individuals \citep{power2011functional}. Formally, the model is defined as follows.

	\begin{definition}[Multilayer stochastic blockmodel \citep{Holland1983}] Let $\bZ\in\{0,1\}^{n\times K}$ be a matrix such that $\sum_{k=1}^K\bZ_{uk}=1$ for each $u\in[n]$, and
		$\bB^{(1)},\ldots, \bB^{(m)}\in[0,1]^{K\times K}$ be symmetric matrices. The random adjacency matrices $\bA^{(1)},\ldots,\bA^{(m)}$ are jointly distributed as a multilayer SBM, denoted by $(\bA^{(1)}, \ldots, \bA^{(m)}) \sim\sbm{\bZ; \bB^{(1)}, \ldots, \bB^{(i)}}$, if each $\bA^{(i)}$ is independently distributed as $\sbm{\bZ, \bB^{(i)}}$.
		\label{definition:multilayerSBM}
	\end{definition}

    		The next proposition formalizes our intuition statement about the semiparametric aspect of the model. Namely, the COSIE graph model can represent a multilayer SBM with $K$ communities using a dimension $d$ that is at most $K$. The proof can be found in the appendix.

	\begin{proposition}
	Suppose that $\bZ\in\{0,1\}^{n\times K}$ and $\bB^{(1)},\ldots, \bB^{(m)}\in[0,1]^{K\times K}$ are the parameters of the multilayer SBM. Then, for some $d\leq K$, there exists a matrix with orthonormal columns $\bV\in\real^{n\times d}$ and symmetric matrices $\bR^{(1)},\ldots, \bR^{(m)}\in\real^{d\times d}$ such that 
	\begin{equation*}
	    \bZ\bB^{(i)}\bZ^\top  = \bV\bR^{(i)}\bV^\top .
	\end{equation*} \label{prop:sbm-is-cosie}
	\end{proposition}
	The previous result shows that the multilayer SBM is a special case of the COSIE model.  Conversely, if $\bV\in\real^{n\times d}$ is the invariant subspace of the COSIE model, and $\bV$ has only $K$ different rows, then the COSIE model is equivalent to the multilayer SBM.
	Furthermore, several extensions of the single-graph SBM can be translated directly into the multilayer setting, and are also contained within the COSIE model. By allowing the matrix $\bZ$ to have more than one non-zero value on each row, overlapping memberships can be incorporated \citep{Latouche2011}, and if the rows of $\bZ$ are nonnegative real numbers such that $\sum_{k=1}^K\bZ_{uk}=1$ then we obtain an extension of the mixed membership model \citep{Airoldi2007}, which can further incorporate degree heterogeneity by multiplying the rows of $\bZ$ by a constant \citep{Zhang2014}. More broadly, if the rows of $\bV$ in a COSIE graph are characterized by a hierarchical structure, then the graphs correspond to a multilayer extension of the hierarchical SBM
	\citep{Lyzinski2017}. Some of these extensions have not been presented before, and hence our work provides an avenue for studying these models, which are  interesting in various applications.

	\subsection{Identifiability\label{sec:identifiability}}
	
	In the COSIE model, note that any orthogonal transformation $\bW\in\real^{d\times d}$ of the parameters keeps the probability matrix of the model unchanged. Indeed, observe that
	\[\bP^{(i)} = \bV\bR^{(i)}\bV^\top  = (\bV\bW)(\bW^\top \bR^{(i)}\bW)(\bV\bW)^\top ,\]
	and therefore the most we can hope is that the parameters of the model are identifiable within the equivalence class
	\[\mathcal{L}(\bV, \{\bR^{(i)}\}_{i=1}^m) = \{\bU,\{\bS^{(i)}\}_{i=1}^m| \ \bU =\bV\bW, \bS^{(i)}=\bW^\top \bR^{(i)}\bW, i\in[m]\text{ for some }\bW\in\mathcal{O}_{d}\}.\]
	This non-identifiability is unavoidable in many latent space models, including the RDPG. However, with multiple graphs the situation is more nuanced, as we do not require the probability matrix of each graph to have the same rank. Note that Definition \ref{definition:model} does not restrict the matrices $\bR^{(1)}, \ldots, \bR^{(m)}$ to be full rank, so the rank of each individual graph can be smaller than the dimension of the joint model.

	The following proposition characterizes the identifiability of the model. Recall that for a given matrix $\bR\in\real^{d\times d}$ with singular value decomposition $\bR=\bU_1\bS\bU_2$, where $\bU_1,\bU_2\in\real^{d\times d}$ are orthogonal matrices and $\bS$ is a non-negative diagonal matrix, the spectral norm is given by  $\|\bR\|:=\max_{i\in[d]} \bS_{ii}$, and the Frobenius norm is defined as $\|\bR\|_F:=\left(\sum_{i=1}^d\sum_{j=1}^d\bR_{ij}^2\right)^{1/2}$.
	The proof of the next result is given on the appendix.
	
	\begin{proposition}[Model identifiability.] \label{prop:identifiability} 
	Let $\bV\in\real^{n\times d}$ be a matrix with orthonormal columns and $\bR^{(1)},\ldots, \bR^{(m)}$ be symmetric matrices such that these are the parameters of the bounded rank $d$ COSIE model.
	\begin{enumerate}[a)]
	    \item For any for any pair of indices $i,j\in[m]$, the pairwise spectral distance $\|\bR^{(i)}-\bR^{(j)}\|$ and the Frobenius distance  $\|\bR^{(i)}-\bR^{(j)}\|_F$ are identifiable. 
	    
        \item Define $\widetilde\bR=\left(\bR^{(1)}, \ldots, \bR^{(m)}\right)\in\real^{d\times dm}$. If the matrix $\widetilde\bR$ has a full rank, then $\bV$ is identifiable up to an orthogonal transformation.
	    \item Given  $\bV$, the matrices $\bR^{(1)},\ldots,\bR^{(m)}$ are identifiable. 
	\end{enumerate}
	\end{proposition}
	
	The previous results present identifiability characterizations at different levels. At the weakest level, the first part of Proposition~\ref{prop:identifiability} ensures that even if the parameters are not identifiable, the Frobenius or spectral distances between the score matrices of the graphs are unique. This property allows the use of distance-based methods for any subsequent inference in multiple graph problems, such as multidimensional scaling, $k$-means or $k$-nearest neighbors; some examples are shown in Section~\ref{sec:simulations} and \ref{sec:data}. 
	
	Proposition~\ref{prop:identifiability} also provides an identifiability condition for the invariant subspace $\bV$. The matrices $\bR^{(i)}$ do not need to have the same rank, but the joint model may require a larger dimension to represent all the graphs, which is given by the rank of $\bR$. To illustrate this scenario, consider the multilayer SBM with 3 communities. Let $\bB^{(1)}, \bB^{(2)}\in\real^{3\times 3}$ be real matrices, $\bZ\in\{0,1\}^{n\times 3}$ a community membership matrix, and $a>b$ some constants so that 
	\[\bB^{(1)} = \left(\begin{array}{ccc}
	     a & b & b\\
	     b & a & a\\
	     b & a & a
	\end{array}\right), \quad\quad\quad \bB^{(2)} = \left(\begin{array}{ccc}
	     a & a & b\\
	     a & a & b\\
	     b & b & a
	\end{array}\right).\]
Although the joint model contains three communities, each matrix $\bB^{(i)}$ is rank 2, so each graph individually  only contains two communities. However, since the concatenated matrix whose columns are given by  $\left(\bB^{(1)},\bB^{(2)}\right)$ has full rank (i.e rank 3),
the three communities of the model can be identified in the joint model, and thus this can be represented as a rank-3 COSIE model.
	
	 The last part of Proposition~\ref{prop:identifiability} shows that the  individual parameters of a graph $\bR^{(i)}$ are only identifiable with respect to a given basis of the eigenspace $\bV$, and hence, all the interpretations that can be derived from depending only on $\bV$. Some instances of the COSIE model, including the multilayer SBM, provide specific characterizations of $\bR^{(i)}$ that can facilitate further interpretation of the parameters.

\section{Fitting COSIE by spectral embedding of multiple adjacency matrices \label{sec:fit}}

This section presents a method to fit the
model provided in Definition~\ref{definition:model} to a sample of $m$ adjacency matrices $(\bA^{(1)},\ldots, \bA^{(m)})\sim\cosie{\bV; \bR^{(1)}, \ldots, \bR^{(m)}}$. Fitting the model requires an estimator $\widehat{\bV}$ for the common subspace $\bV$, for which we present a spectral approach. Given an estimated common subspace, we estimate the individual score matrices $\bR^{(i)}$ for each graph by least squares, which has a simple solution in out setting. In all of our analysis, we assume that the dimension of the model $d$ is known in advance, but we present a method to estimate it in practice.
	
	We start by defining the \emph{adjacency spectral embedding} (ASE) of an adjacency matrix $\bA$, which is a standard tool for estimating the latent positions of a RDPG.
	
	\begin{definition}[Adjacency spectral embedding \citep{Sussman2012,Rubin-Delanchy2017}]
	For an adjacency matrix $\bA$, let $\bA=\widehat{\bV}\widehat{\bD}\widehat{\bV}^\top + \widehat{\bV}_\perp\widehat{\bD}_\perp\widehat{\bV}^\top _\perp$ be the eigendecomposition of $\bA$ such that $(\widehat{\bV}, \widehat{\bV}_\perp)$ is the $n\times n$ orthogonal matrix of eigenvectors, with $\widehat{\bV}\in\real^{n\times d}$, $\widehat{\bV}_\perp\in\real^{n\times(n- d)}$, and $\widehat{\bD}$ is a diagonal matrix containing the $d$ largest eigenvalues in magnitude. The \emph{scaled adjacency spectral embedding}  of $\bA$ is defined as $\widehat{\bX}=\widehat{\bV}|\widehat{\bD}|^{1/2}$. We refer to $\widehat{\bV}$ as the \emph{unscaled adjacency spectral embedding}, or simply as the leading eigenvectors of $\bA$.
	\end{definition}

    The ASE provides a consistent and asymptotically normal estimator of the corresponding latent positions under the RDPG and GRDPG models as the number of vertices $n$ increases \citep{Sussman2014,Athreya2017,Rubin-Delanchy2017}. Therefore, under the COSIE model for which $\bP=\bV\bR\bV^\top $, the matrix $\widehat{\bV}$ obtained by ASE is a consistent estimator of $\bV$, up to an orthogonal transformation, provided that the corresponding $\bR$ has full rank. This suggests that given a sample of adjacency matrices, one can  simply use the unscaled ASE of any single graph to obtain an estimator of $\bV$. However, this method is not leveraging the information about $\bV$ on all the graphs. 
    
	To give an intuition in how to fit the joint model for the data,  we first consider working with the expected probability matrix of each graph $\bP^{(i)}$. For simplicity in all our analysis here, we assume that all the matrices $\bR^{(1)}, \ldots, \bR^{(m)}$ have full rank, so each matrix $\bP^{(i)}$ has exactly $d$ non-zero eigenvalues. Let $\bV^{(i)}$ and $\bX^{(i)}$ be the unscaled and scaled ASE of the matrix $\bP^{(i)}$. Then
	\begin{align*}
	  \bV^{(i)} & = \bV\bW^{(i)},  \\
	  \bX^{(i)} & = \bV\bW^{(i)}|\bD^{(i)}|^{1/2},
	\end{align*}
	for some orthogonal matrix $\bW^{(i)}\in\real^{d\times d}$, and a diagonal matrix $\bD^{(i)}$ containing the eigenvalues of $\bP^{(i)}$. Note that the matrices $\bW^{(i)}$ are possibly different for each graph, which is problematic when we want to leverage the information of all the graphs. Consider the matrix
	\begin{align*}
		\bU & = \left(\bV^{(1)} \ \cdots \  \bV^{(m)}\right),
	\end{align*}
	that is formed by concatenating the unscaled ASEs of each graph, resulting in matrices of size $n\times (dm)$. Alternatively, consider the same matrix but concatenating the scaled ASEs, given by
	\begin{align*}
        \bU' & =  \left(\bX^{(1)} \ \cdots \  \bX^{(m)}\right).
	\end{align*}
	  Note that both matrices $\bU$ and $\bU'$ are rank $d$,  and the $d$ left singular vectors of either of them corresponds to the common subspace $\bV$, up to some orthogonal transformation. The SVD step effectively aligns the multiple ASEs to a common spectral embedding.	  In the scaled ASE, the SVD also eliminates the effect of the eigenvalues of each graph, which are nuisance parameters in the common subspace estimation problem.

	In practice, we do not have access to the expected value of the matrices $\bP^{(i)}$, but we use the procedure described above with the adjacency matrices to obtain an estimator $\widehat{\bV}$ of the common subspace. Given $\widehat{\bV}$, we proceed to estimate the individual parameters of the graphs by minimizing a least squares function, 
	\begin{equation*}
	    \widehat{\bR}^{(i)} = \argmin_{\bR\in\real^{d\times d}}\|\bA^{(i)} - \widehat{\bV}\bR\widehat{\bV}^\top \|_F^2 = \argmin_{\bR\in\real^{d\times d}}\|\bR - \widehat{\bV}^\top \bA^{(i)}\widehat{\bV}\|_F^2,
	\end{equation*}
	and so, the estimator has a closed form given by 
	\[\widehat{\bR}^{(i)}= \widehat{\bV}^\top \bA^{(i)}\widehat{\bV}.\]
	Once these parameters are known, one can also estimate the matrix of edge probabilities as $$\widehat{\bP}^{(i)}=\widehat{\bV}\widehat{\bR}^{(i)}\widehat{\bV}^\top .$$
	We call this procedure the \emph{multiple adjacency spectral embedding} (MASE), and it is summarized in Algorithm~\ref{alg:mase}.

 	\begin{algorithm}
 		\caption{Multiple adjacency spectral embedding (MASE)}
 		\begin{algorithmic} 
 			\Input Sample of graphs $\bA^{(1)},\ldots,\bA^{(m)}$; embedding dimensions $d$ and $\{d_i\}_{i=1}^m$.
 			\begin{enumerate}
 				\item For each $i\in[m]$, obtain the adjacency spectral embedding of $\bA^{(i)}$  on $d_i$ dimensions, and denote it by $\widehat{\bV}^{(i)}\in\mathbb{R}^{n\times d_i}$.
 				\item Let $\widehat{\bU}= \left( \widehat{\bV}^{(1)} \ \cdots \ \widehat{\bV}^{(m)}\right)$ be the $n\times\left(\sum_{i=1}^m d_i\right)$ matrix of concatenated spectral embeddings.
 				\item Define $\widehat{\bV}\in\mathbb{\bR}^{n\times d}$ as the matrix containing the $d$ leading left singular values of $\widehat{\bU}$.
 				\item For each $i\in[m]$, set $\widehat{\bR}^{(i)} = \widehat{\bV}^\top \bA^{(i)}\widehat{\bV}$.
 			\end{enumerate}
 			\Output $\widehat{\bV}, \{\widehat{\bR}^{(i)}\}_{i=1}^m$. 
 		\end{algorithmic}
 		\label{alg:mase}
 	\end{algorithm}
	
	The MASE method presented for estimating the parameters of the COSIE model only relies on singular value decompositions, and hence it is simple and computationally scalable. This method extends the ideas of a single spectral graph embedding to a multiple graph setting.  The first step of Algorithm~\ref{alg:mase}, which is typically the major burden of the method, can be carried in parallel. The estimation of the score matrices only requires to know an estimate for $\bV$, and hence allows to easily compute an out-of-sample-embedding for a new graph $\bA^{(m+1)}$ without having to update all the parameter estimates. 
    
	Figure \ref{fig:alg1} shows a graphical representation of MASE for estimating the common subspace of a set of four graphs. The graphs correspond to the multilayer SBM with two communities, and the connection matrices $\bB^{(i)}$ are selected in a way that the graphs have an heterogeneous structure within the sample. Note that after step 1 of Algorithm~\ref{alg:mase} the latent positions obtained by ASE are slightly rotated relatively to each graph. After estimating the common subspace by SVD, a common set of latent positions is found, which look  tightly clustered within their community, showing that MASE is able to leverage the information of all the graphs in this example.
	
    
	The first step of Algorithm~\ref{alg:mase} corresponds to a separate ASE of each graph. We stated the algorithm using the unscaled version of the ASE, and later in Section~\ref{sec:theory}, we show that this version of the method is consistent in estimating the common subspace $\bV$. In practice, this step can be replaced with other spectral embeddings of $\bA^{(i)}$, including the scaled version described above. Note that the scaled ASE also makes use of the eigenvalues, and thus it puts more weight onto the columns of the matrix $\widehat{\bU}'$ that correspond to the eigenvectors associated with the largest eigenvalues, which can be convenient, especially if the dimension $d$ is overestimated. Other spectral embeddings that also aim to estimate the eigenspace of the adjacency matrices, including the Laplacian spectral embedding or a regularized version of it \citep{Priebe2018,Le2017}, might be preferred in some circumstances, but we do not explore this direction further.

	\subsection{Choice of embedding dimension}
    
    In practice, the dimension of the common subspace $d$ is usually unknown. Moreover, each individual $\bR^{(i)}$ might not be full rank, and so estimating a different embedding dimension $d_i$ for each graph might be necessary. The actual values of $d$ and $\{d_i\}$ correspond to the ranks of $\bU$ and $\{\bP^{(i)}\}$ respectively, and so they can be approximated by estimating the ranks of $\{\widehat{\bA}^{(i)}\}$ and $\widehat{\bU}$. The scree plot method, which consists in looking for an elbow in the plot of ordered singular values of a matrices, provides an estimator of these quantities, and can be automatically performed using the method proposed by \cite{Zhu2006}. We use this method to fit the model to real data in Section \ref{sec:data}.
	
	
	
	\section{Theoretical results\label{sec:theory}}
	
In this section, we study the statistical performance of Algorithm~\ref{alg:mase} in estimating the parameters of the COSIE model when a sample of graphs is given. 
We first study the expected error in estimating the common subspace, and show that this error decreases as a function of the number of graphs. This result demonstrates that our method is able to leverage the information of the multiple graphs to improve the estimation of $\bV$. 
    
	To study the estimation error of our method, we consider a sequence of parameters of the COSIE model $\{\bR^{(1, n)}, \ldots, \bR^{(m, n)}\in\real^{d\times d}, \bV^{(n)}\in\real^{n\times d}\}_{n=n_0}^\infty$, and present non-asymptotic error bounds on the estimation of $\bV^{(n)}$, as well as a result on the asymptotic distribution of the estimators of the score matrices as $n$ goes to infinity.
	The magnitude of the entries of the parameters typically changes with $n$, and in order to obtain consistent estimators we will require that $\|\bR^{(i,n)}\|\rightarrow \infty$. This is a natural assumption considering the fact that $\bR^{(i,n)}$ contains the eigenvalues of the expected adjacency matrix of a graph, and requiring these eigenvalues to grow as $n$ increases is necessary in order to control the error of the ASE \citep{Athreya2017}. In the following, we omit the dependence of the parameters in $n$ to ease the notation. To simplify the analysis, we also assume that all the score matrices $\bR^{(1)}, \ldots, \bR^{(m)}$ have full rank.
	
	Our next theorem introduces a bound on the expected error in estimating the common subspace of the COSIE model, which is the basis of our theoretical results. The proof of this result relies on bounds from \cite{Fan2017} for studying the eigenvector estimation error in distributed PCA. However, our setting presents some substantial differences, including the distribution of the data and the fact that the expected adjacency matrices are not jointly diagonalizable.  Moreover, the scaled version of the algorithm adds the complication of working with eigenvalues, and extending the theory to this case is more challenging. We thus work with the unscaled version of the ASE in Algorithm \ref{alg:mase}. Before stating the theorem, we introduce some notation. For a given square symmetric matrix $\bM\in\real^{r\times r}$, denote by $\lambda_1(\bM)\geq \ldots \geq \lambda_r(\bM)$ to the ordered eigenvalues of $\bM$, and define
	$\lambda_{\min}(\bM)$ and $\lambda_{\max}(\bM)$ as the smallest and largest eigenvalue in magnitude of $\bM$, and $\delta(\bM)=\max_{u\in[r]}\sum_{v=1}^r\bM_{uv}$ the largest sum of the rows of $\bM$. Given a pair of sequences $\{a_n\}_{n=1}^\infty$ and $\{b_n\}_{n=1}^\infty$, denote by $a_n\lesssim b_n$ if there exists some constant $C>0,n_0>0$ such that $a_n\leq Cb_n$ for all $n\geq n_0$.
	

    \begin{theorem}
    \label{thrm:COSIE Expectation Bound} Let $\bR^{(1)},\ldots,\bR^{(m)}$ be a collection of full rank symmetric matrices of size $d\times d$, $\bV\in\real^{n\times d}$ a matrix with orthonormal columns, and $(\bA^{(1)},\ldots, \bA^{(m)})\sim\cosie{\bV, \bR^{(1)}, \ldots, \bR^{(m)}}$ a sample of $m$ random adjacency matrices, and set $\bP^{(i)} = \bV\bR^{(i)}\bV^\top $. Define
		\begin{equation}
		\varepsilon = \sqrt{\frac{1}{m}\sum_{i=1}^m \frac{\delta(\bP^{(i)})}{\lambda_{\min}^2(\bR^{(i)})}}.\label{eq:AVPVerror}
		\end{equation}
		Let $\widehat{\bV}$ be the estimator of $\bV$ obtained by Algorithm~\ref{alg:mase}. Suppose that $\min_{i\in[m]} \delta(\bP^{(i)})=\omega(\log n)$ and $\varepsilon=o(1)$ as $n\rightarrow\infty$.  Then,
		\begin{equation}
			\mathbb{E}\left[\min_{\bW\in\mathcal{O}_d}\|\widehat{\bV}-\bV\bW\|_F\right] \lesssim \sqrt{\frac{d}{m}}\varepsilon + \sqrt{d}\varepsilon^2. \label{eq:theorem-bound}
		\end{equation}
		\label{thm:V-Vhat}
	\end{theorem}
    
    \noindent The above theorem requires two conditions on the expected adjacency matrices, that control the estimation error of each ASE for a single graph. First, the largest expected vertex degree $\delta(\bP^{(i)})$ needs to grow at rate $\omega(\log n )$ to ensure the concentration in spectral norm of the adjacency matrix to its expectation (see Theorem 20 of \cite{Athreya2017}). Second, the condition $\varepsilon=o(1)$ is required to control the average estimation error of each ASE. For this condition to hold, it is enough to have that $\delta(\bP^{(i)})=o\left(\lambda^2_{\min}(\bR^{(i)})\right)$ for each graph $i\in[m]$. In particular, observe that
    $$\delta(\bP^{(i)})= \|\bP^{(i)}\|_1 \leq \sqrt{n}\lambda_{\max}(\bR^{(i)}),$$
    so it is sufficient that $\max_{i\in[m]}\frac{\lambda_{\max}(\bR^{(i)})}{\lambda_{\min}^2(\bR^{(i)})} = o\left(\frac{1}{\sqrt{n}}\right)$. Specific conditions for the multilayer SBM that relate the community sizes and the eigenvalues of the connection probability matrices are presented in Section~\ref{sec:community-detection}.
    
Theorem~\ref{thm:V-Vhat} theorem bounds the expected error in estimating the subspace of $\bV$ by the sum of two terms in Equation~\eqref{eq:theorem-bound}. To understand these terms, note that the error can be partitioned as
    \begin{align}
	\min_{\bW\in\mathcal{O}_d}\|\widehat{\bV}-\bV\bW\|_F & \leq  \|\widehat{\bV}\widehat{\bV}^\top -\bV\bV^\top \|_F \nonumber\\
	& \leq \left\|\widehat{\bV}\widehat{\bV}^\top -\widetilde{\bV}\widetilde{\bV}^\top \right\|_F + \left\|\widetilde{\bV}\widetilde{\bV}^\top -\bV\bV^\top \right\|_F,
	\label{eq:theory-bias-variance}
	\end{align}
	 where $\widetilde{\bV}$ is the matrix containing the $d$ leading eigenvectors of $\frac{1}{m}\sum_{i=1}^m\e\big[\widehat{\bV}^{(i)}(\widehat{\bV}^{(i)})^\top \big]$.
	The first term of the right hand side of Equation~\eqref{eq:theory-bias-variance} measures the variance of the estimated projection matrix, which  corresponds to the first term in Equation~\eqref{eq:theorem-bound}. 
	Thus, this term converges to zero as the number of graphs increases. The estimated projection is not unbiased, and therefore the second term of Equation~\ref{eq:theory-bias-variance} is always positive, but its order is usually smaller than the variance. This bound is a function of $\varepsilon$, which quantifies the average error in estimating $\bV$ by using the leading eigenvectors of $\bA^{(i)}$ on each network. Note that the inclusion of an orthogonal transformation $\bW$ has to do with the fact that the common subspace is only identifiable up to such transformation.
    
    The above result establishes that when the sample size $m$ is small ($m\lesssim 1/\varepsilon$), the variance term dominates the error as observed in the left hand side of Equation \eqref{eq:theorem-bound}, but as the sample size increases this term vanishes. The second term  is the bias of estimating $\bV$ by using the leading eigenvectors of each network, and this term dominates when the sample size is large. These type of results are common in settings when a global estimator is obtained from the average of local estimators \citep{Lee2017,Arroyo2016,Fan2017}.
	
	To illustrate the result of Theorem \ref{thm:V-Vhat}, consider the following example on the Erd\"os-R\'enyi (ER) model \citep{Erdos1959,Gilbert1959}. In the ER random graph model $G(n,p)$, the edges between any pair of the $n$ vertices are formed independently with a constant probability $p$. We consider a sample of adjacency matrices $\bA^{(1)}, \ldots,\bA^{(m)}$, each of them having a different edge probability $p_1,\ldots, p_m$ so that $\bA^{(i)}\sim G(n, p_i)$. By defining $\bV=\frac{1}{\sqrt{n}}\textbf{1}_n$, where $\textbf{1}_n$ is the $n$-dimensional vector with all entries equal to one, we can represent these graphs under the COSIE model by setting $\bP^{(i)} = p_in\bV\bV^\top $. Hence, $\delta(\bP^{(i)}) = p_in$ and
	$$\varepsilon = \sqrt{\frac{1}{nm}\sum_{i=1}^m \frac{1}{p_i}}.$$
	Therefore, if  $p_i=\omega(\log(n)/n)$, Theorem \ref{thm:V-Vhat} implies that
	\begin{align*}
		\mathbb{E}\left[\min_{\bW\in\mathcal{O}_d}\|\widehat{\bV}-\bV\bW\|_F\right]  & \lesssim \sqrt{\frac{1}{nm^2}\sum_{i=1}^m \frac{1}{p_i}} + {\frac{1}{nm}\sum_{i=1}^m \frac{1}{p_i}} 
	\end{align*}
	In the dense regime when $p_i=\omega(1)$, the error of estimating the common subspace is of order $O\left(\frac{1}{\sqrt{nm}}+\frac{1}{n}\right)$. An alternative estimator for the common subspace $\bV$ is given by the ASE of the mean adjacency matrix, $\operatorname{ASE}\left(\frac{1}{m}\sum_{i=1}^m\bA^{(i)}\right)$, for which the expected subspace estimation error is $O\left(\frac{1}{\sqrt{mn}}\right)$  \citep{Tang2018}. When $m\lesssim n$, the error rate of both estimators coincide. However, the estimator of \cite{Tang2018} is only valid when all the graphs have the same expected adjacency matrices, and can perform poorly in  heterogeneous populations of graphs (see Section \ref{sec:simulations}).
	The ER model described above can be regarded as a special case of the multilayer SBM with only one community, which we consider next.

	\subsection{Perfect recovery in community detection \label{sec:community-detection}}
	
	Estimation of the common subspace in the multilayer SBM is of particular interest due to its relation with the community detection problem \citep{Girvan2002,Abbe2017}. In the multilayer SBM with $K$ communities, the matrix of the common subspace basis $\bV$ has $K$ different rows, each of them corresponding to a different community. An accurate estimator $\widehat{\bV}$ will then reflect this structure, and when the community labels of the vertices are not known, clustering the rows of $\widehat{\bV}$ into $K$ groups can reveal these labels. Different clustering procedures can be used for this goal, and in this section we focus on $K$-means clustering defined next. Suppose that $\mathcal{Z}$ is the set of valid membership matrices for $n$ vertices and $K$ communities, that is,
	\[\mathcal{Z} = \left\{\bZ\in\{0,1\}^{n\times K}:\ \sum_{v=1}^K\bZ_{uv}=1,\ u\in[n] \right\}.\]
	The community assignments are obtained by solving the $K$-means problem
	\begin{equation}
	    (\widehat{\bZ}, \widehat{\bC}) = \argmin_{\widetilde{\bZ}\in\mathcal{Z}, \widetilde{\bC}\in\real^{K\times K}}\|\widetilde{\bZ}\widetilde{\bC}-\widehat{\bV}\|_F.
	    \label{eq:k-means}
	\end{equation}
	The matrix $\widehat{\bZ}$ thus contains the estimated community assignments, and $\widehat{\bC}$ their corresponding centroids.
	
	To understand the performance in community detection of the procedure described above, consider a multilayer SBM with parameters $\bB^{(1)}, \ldots, \bB^{(m)}\in[0,1]^{K\times K}$ and $\bZ\in\{0,1\}^{n\times K}$. Denote the community sizes by $n_1,\ldots,n_K$, such that $\sum_{j=1}^k n_j=n$, and define the quantities $\bXi=\operatorname{diag}(n_1,\ldots, n_K)$, $n_{\min}=\min_{j\in[K]}n_j$, $n_{\max}=\max_{j\in[K]}n_j$ and $\bP^{(i)}=\bZ\bB^{(i)}\bZ^\top $. To obtain a consistent estimator of the common invariant subspace, the following assumption is required.
	
	\begin{assumption} \label{assumption:multilayer-SBM-parameters}
	There exist some absolute constants $\kappa\in(0,1], \gamma >0$ such that for all $i\in[m]$,
    \begin{equation*}
        \frac{\sqrt{K} \left(\sum_{j=1}^Kn_j^2\right)^{1/2}}{n_{\min}^{2-\kappa}}\frac{\lambda_1(\bB^{(i)})}{\lambda_{\min}^2(\bB^{(i)})} \leq \gamma.
    \end{equation*}
	\end{assumption}
	
	The previous assumption requires a uniform control on ratios of the community sizes and the eigenvalues of the connectivity matrices. When all communities have comparable sizes (i.e., $n_{\min}\geq cn_{\max}$ for some absolute constant $c>0$), then the smallest eigenvalue of each $\bB^{(i)}$ is allowed to approach zero at a rate
	\[|\lambda_{\min}(\bB^{(i)})| = \Omega\left(\frac{\sqrt{K}\lambda^{1/2}_1(\bB^{(i)})}{n^{1/2-\kappa/2}}\right).\]
	 On the other hand, if the smallest eigenvalue of each  $\bB^{(i)}$ is bounded away from zero by an absolute constant, then the size of the largest community can be at most $n_{\max}= O\left(\frac{n_{\min}^{2-\kappa}}{K}\right)$.
	
	The next Corollary is a consequence of Theorem~\ref{thm:V-Vhat} for the multilayer SBM setting. The proof is given on the Appendix.

    \begin{corollary}
    Consider a sample of $m$ adjacency matrices from the multilayer SBM 
    \[(\bA^{(1)}, \ldots, \bA^{(m)})\sim\sbm{\bZ; \bB^{(1)}, \ldots, \bB^{(m)}},\]
    and let $\widehat{\bV}$ the estimated subspace from Algorithm \ref{alg:mase}. 
    Suppose that $n_{\min}=\omega(1)$, $\delta(\bP^{(i)})=\omega(\log n)$, and that Assumption~\ref{assumption:multilayer-SBM-parameters} holds. Then, for sufficiently large $n$,
    \begin{equation}
			\mathbb{E}\left[\min_{\bW\in\mathcal{O}_d}\|\widehat{\bV}-\bZ\bXi^{-1/2}\bW\|_F\right] \lesssim \sqrt{\frac{\gamma K}{mn_{\min}^\kappa}}+ \frac{\gamma\sqrt{K}}{n_{\min}^\kappa}. 
		\end{equation}\label{cor:sbm-V-Z}
    \end{corollary}
	
	The above corollary presents sufficient conditions for consistent estimation of the common invariant subspace, which contains the information of the community labels encoded in the matrix $\bZ$. The corollary requires the size of the smallest community to grow as $n$ increases, which is something commonly assumed in the literature \citep{Rohe2011,Lyzinski2014}. 
	
	The following result provides a bound on the expected number of misclustered vertices after clustering the rows of $\widehat{\bV}$ according to Equation~\eqref{eq:k-means}. Due to the non-identifiability of the community labels, these are recovered up to a permutation $\bQ\in\mathcal{P}_K$, where $\mathcal{P}_K$ is the set of $K\times K$ permutation matrices.
	
	\begin{theorem} \label{thm:community-detection}
	Suppose that $\widehat{\bZ}$ is the membership matrix obtained by Equation~\eqref{eq:k-means}. Under the conditions of Corollary~\ref{cor:sbm-V-Z}, there exist some absolute constants $\gamma>0$, $\kappa\in(0,1]$ such that the expected number of misclustered vertices satisfies
	\begin{equation}\label{eq:thm-communitydetection}
	    \e\left[\min_{\bQ\in\mathcal{P}_K}\|\widehat{\bZ}-\bZ\bQ\|_F\right] \lesssim
	    \sqrt{\frac{\gamma Kn_{\max}}{mn_{\min}^\kappa}} + \frac{\gamma\sqrt{Kn_{\max}}}{n_{\min}^{\kappa}}.
	\end{equation}
	\end{theorem}
	
Similar to the estimation error of the the common subspace in Theorem \ref{thm:V-Vhat}, the bound on the expected clustering error depends on two terms in Equation~\eqref{eq:thm-communitydetection}. When the number of graphs $m$ increases, the first term vanishes. The second term also decays to zero as long as $Kn_{\max}=o(n_{\min}^{2\kappa})$, and therefore the algorithm recovers the community labels correctly  as $n$ goes to infinity, provided that $m=\Omega(n_{min}^\kappa)$. On the other hand, our theory does not guarantee perfect community recovery for $m=O(1)$ because the first term in Equation~\eqref{eq:thm-communitydetection} has at least a constant order. We remark that this fact is a technical consequence of our analysis of subspace estimation error, and it is analogous to similar results obtained in spectral clustering on a single graph \citep{Rohe2011}. In contrast, previous work by \cite{Lyzinski2014} shows, via a careful analysis of the row-wise error in subspace estimation, that spectral methods can produce perfect recovery in community detection in the single-graph setting. Developing a similar result in our setting requires further investigation that we leave as future work.
	
	
	   
	   \subsection{Asymptotic normality of the estimated score matrices}
	   Now, we study the asymptotic distribution of the individual estimates of the score matrices. 
	      We show not only that the entries of $\widehat{\bR}^{(i)}$ accurately approximate the entries of $\bR^{(i)}$ up to some orthogonal transformation, but that under delocalization assumptions on $\bV$, the entries $\widehat{\bR}_{kl}^{(i)}$, where $1 \leq k \leq d, 1 \leq l \leq d$, exhibit asymptotic normality. In effect, we prove a central limit with a bias term, but this bias term decays as the number of graphs $m$ increases. The central limit theorem depends, in essence, on 
	      the ability to estimate the common subspace $\bV$ to sufficient precision. 

	    To make such the central limit theorem precise, we require the following delocalization and edge variance assumptions.
	    
	    \begin{assumption}[Delocalization of $\bV$]\label{assump:Delocalization} There exist constants $c_1, c_2>0$, and an orthogonal matrix $\bW\in\mathcal{O}_d$ such that each entry of $\bV\bW$ satisfies
	    $$\frac{c_1}{\sqrt{n}} < (\bV\bW)_{kl} < \frac{c_2}{\sqrt{n}}, \quad\quad\forall k\in[n], l\in[d].$$
	    \end{assumption}
	    
\noindent The delocalization assumption requires that the score matrix $\bR^{(i)}$ influence the connectivity of enough edges in the graph. Because the common subspace $\bV$ is invariant to orthogonal transformations, the assumption only requires the existence of a particular orthogonal matrix  $\bW$ that  satisfies the entrywise inequalities.
This delocalization assumption is satisfied by a wide variety of graphs; for example, if we consider $\bV$ as the eigenvectors of an Erd\H{o}s-R\'enyi graph, its entries are all of order $1/\sqrt{n}$, and similarly if $\bV$ are the eigenvectors of the probability matrix of a stochastic blockmodel all of whose block sizes grow linearly with $n$ (see Proposition~\ref{proposition:sbm-delocalization} below). 

	    \begin{assumption}[Edge variance]\label{assump:variance_P} The sum of the variance of the edges satisfies
	    \begin{equation}\label{eq:s_n-definition}
	    s^2(\bP^{(i)}):=\sum_{s=1}^n\sum_{t=1}^n \bP^{(i)}_{st}(1-\bP^{(i)}_{st}) = \omega(1)
	    \end{equation}
	    for all $i\in[m]$. 
	    \end{assumption}

	    Assumptions \ref{assump:Delocalization} and \ref{assump:variance_P} hold for a wide variety of graphs. In particular, when the community sizes of an SBM are balanced and the connection probabilities do not decay very fast, the next proposition shows that these assumptions hold. The proof is included on the Appendix.

	    \begin{proposition}\label{proposition:sbm-delocalization}
	    Let $\bZ\in\{0,1\}^{n\times K}$ be a membership matrix with community sizes $n_1, \ldots, n_K$, and $\bB^{(i)}\in[0,1]^{K\times K}$ a connectivity matrix . 
	    \begin{enumerate}[a)]
        \item If $\min_{k\in[K]}n_k=\Omega(n)$, then Assumption~\ref{assump:Delocalization} holds for the common invariant subspace of a multilayer SBM with parameter $\bZ$.
	        \item If in addition to a), $\min_{i\in[m]}\sum_{v=1}^K\bB^{(i)}_{uv}(1-\bB^{(i)}_{uv})=\omega(1/n^2)$, then Assumption~\ref{assump:variance_P} also holds.
	    \end{enumerate}
	    \end{proposition}

The score matrices are symmetric, so we focus on their upper triangular and  diagonal elements, and for a symmetric matrix $\bR\in\real^{d\times d}$, we define $\operatorname{vec}(\bR)\in\real^r$ as the vector of dimension $r=\frac{d(d+1)}{2}$ that contains those elements of  $\bR$, so that for $k,l\in[d], k\leq l,$
$$[\operatorname{vec}(\bR)]_{\frac{2k+l(l-1)}{2}}:=\bR_{kl}.$$

Given  the matrices $\bV$ and $\bP^{(i)}=\bV\bR^{(i)}\bV^\top $ we  define  $\bSigma^{(i)}\in\real^{r\times r}$ as
\begin{equation}
\bSigma^{(i)}_{\frac{2k + l(l-1)}{2}, \frac{2k'+l'(l'-1)}{2}}  := \sum_{s=1}^{n-1}\sum_{t=s+1}^n\bP^{(i)}_{st}(1-\bP^{(i)}_{st})\left[(\bV_{sk}\bV_{tl} + \bV_{tk}\bV_{sl})(\bV_{sk'}\bV_{tl'} + \bV_{sk'}\bV_{tl'})\right].\label{eq:covarianceR}
\end{equation}
The above  matrix  can change with $n$ since it is a function of $\bV$ and $\bR^{(i)}$. As it will be formalized in Theorem~\ref{thm:COSIE_CLT}, when  $n$ is large enough the covariance of the upper triangular and diagonal entries of $\widehat{\bR}^{(i)}$ (after some proper alignment) is approximated by $\bSigma^{(i)}$. The following assumption provides a sufficient condition to derive the asymptotic joint distribution of such entries of $\widehat{\bR}^{(i)}$.

 \begin{assumption}[Score matrix covariance]\label{assump:score_lambdamin} The magnitude of the smallest eigenvalue of $\bSigma^{(i)}$ satisfies $|\lambda_{\min}(\bSigma^{(i)})|=\omega(n^{-2})$.
\end{assumption}

The next theorem presents the asymptotic distribution of the estimated score matrices as the size of the graphs increases. The proof is given on the Appendix.
\begin{theorem}\label{thm:COSIE_CLT} Suppose $(\bA^{(1)}, \ldots,\bA^{(m)})\sim\cosie{\bV;\bR^{(1)}, \ldots, \bR^{(m)}}$ are a sample of adjacency matrices from the COSIE model such that the parameters satisfy $\min_{i\in[m]}\delta(\bP^{(i)})=\omega(\log n)$, $\varepsilon= O\left(\frac{1}{\max\sqrt{\delta(\bP^{(i)})}}\right)$, as well as the delocalization requirements given in Assumption~\ref{assump:Delocalization}.  
	\begin{itemize}
		\item[a)] Suppose that Assumption~\ref{assump:variance_P} holds. Let $\bW\in\mathcal{O}_d$ be an orthogonal matrix such that $\bW=\argmin_{\bW\in\mathcal{O}_d}\|\widehat{\bV} - \bV\bW\|_F$. Then, 
		the  estimates of $\{\bR^{(i)}\}_{i=1}^m$ provided by MASE (Algorithm~\ref{alg:mase}), $\{\widehat{\bR}^{(i)}\}_{i=1}^m$, satisfy
		\[\left(\bSigma^{(i)}_{\frac{2k + l(l-1)}{2}, \frac{2k+l(l-1)}{2}}\right)^{-1/2}\left(\bW \widehat{\bR}^{(i)}\bW^\top  - \bR^{(i)}  + \bH_{m}^{(i)}\right)_{kl} \overset{d}{\longrightarrow} \mathcal{N}(0,1), \]
		as $n\rightarrow\infty$, where $\mathcal{N}(0,1)$ is a univariate standard normal distribution, $\bH_{m}^{(i)}$ is a random matrix that satisfies $\e[\|\bH_m^{(i)}\|_F]=O(\frac{d}{\sqrt{m}})$, and 
		and $\bSigma^{(i)}_{\frac{2k + l(l-1)}{2}, \frac{2k+l(l-1)}{2}} \sim \frac{s^2(\bP^{(i)})}{n^2}$.
		
		\item[b)] If in addition to the requirements in a), Assumption~\ref{assump:score_lambdamin}  holds, then 
		\[\left(\bSigma^{(i)}\right)^{-1/2}\operatorname{vec}\left(\bW\widehat{\bR}^{(i)}\bW^\top - \bR^{(i)}  + \bH_{m}^{(i)}\right) \overset{d}{\longrightarrow} \mathcal{N}(\mathbf{0}_r,\mathbf{I}_r), \]
		where $\mathcal{N}(\mathbf{0}_r,\mathbf{I}_r)$ is an $r$-dimensional standard normal distribution.
	\end{itemize}
\end{theorem}



    In addition to Assumptions~\ref{assump:Delocalization}, \ref{assump:variance_P} and \ref{assump:score_lambdamin}, Theorem~\ref{thm:COSIE_CLT} imposes slightly stronger assumptions on $\varepsilon$ than Theorem~\ref{thm:V-Vhat} to obtain an asymptotically normal distribution for the entries of $\bR^{(i)}$. 
    
    Due to the non-identifiability of the specific matrix $\bV$, the results in Theorem~\ref{thm:COSIE_CLT} about each estimated score matrix $\widehat{\bR}^{(i)}$ are stated in terms of an  orthogonal matrix $\bW$. The next corollary shows that this unknown matrix $\bW$ can be removed and each individual entry of the matrix $\widehat{\bR}^{(i)}$ also has an asymptotic normal distribution with a unspecified but bounded variance.
    
     \begin{corollary} \label{corollary:variance-VEV}
Suppose that the assumptions of Theorem~\ref{thm:COSIE_CLT} part b) hold. 
Then, there exists a sequence of orthogonal matrices $\bW\in\mathcal{O}_d$ such that
$$\frac{1}{\sigma_{i,k,l}}(\widehat{\bR}^{(i)} - \bW^\top \bR^{(i)} \bW+\bH_m^{(i)})_{kl}\overset{d}{\rightarrow}\mathcal{N}(0,1),$$
where $\e[\|\bH^{(i)}_m\|_f]=O\left(\frac{d}{\sqrt{m}}\right)$ and $\sigma_{i,k,l}^2=O\left( \frac{d^2s^2(\bP^{(i)})}{n^2}\right)=O(1)$.
\end{corollary}

    Using Corollary~\ref{corollary:variance-VEV} in combination with a tail bounds of a normal distribution, for sufficiently large $n$, the mean square error between the estimated score matrix and its expectation satisfies
    \begin{equation}
        \| \widehat{\bR}^{(i)} - \bW^\top \bR^{(i)}\bW  + \bH_m\|_F = O_P\left( \frac{d^2\sqrt{s^2(\bP^{(i)})}}{n} \right).
        \label{eq:R-Rhat-bigO}
    \end{equation}
    Because $s^2(\bP^{(i)})=O(n^2)$, the right hand side of Equation~\eqref{eq:R-Rhat-bigO} is at most of an order $O_P(d^2)$. On the other hand, a lower bound for the the norm of the score matrices can be obtained by observing that
    $$\|\bR^{(i)}\|_F^2 = \|\bP^{(i)}\|_F^2 =  \sum_{k=1}^d\left(\lambda_k(\bP^{(i)})\right)^2.$$
    The eigenvalues of $\bP^{(i)}$ are of order $\delta(\bP^{(i)})$ (see Remark 24 in \cite{Athreya2016}), and hence, under the assumptions of Theorem~\ref{thm:COSIE_CLT}, $\|\bR^{(i)}\|_F=\omega(\log n)$. Therefore, the relative mean square error of the estimated score matrices with respect to its expectation goes to zero as $n$ increases, and if the sample size $m$ is also growing, the estimation bias $\bH_m$ vanishes.

        
        
    To illustrate the result, we consider the problem of estimating the eigenvalues of a random graph under the COSIE model. For that goal, we generate a sample of $m$ adjacency matrices from the COSIE model, and use the eigenvalues of $\widehat{\bR}^{(1)}$ as an estimate of the eigenvalues of $\bP^{(1)}$. According to Theorem \ref{thm:COSIE_CLT}, the entries of $\widehat{\bR}^{(1)}$ are asymptotically normally distributed and centered around a bias term $\bH$ that vanishes as the sample size grows. We conjecture that there is a similar phenomenon in the eigenvalues of $\bR^{(1)}$; namely, that the eigenvalues of $\widehat{\bR}^{(i)}$ are asymptotically normally distributed about the eigenvalues of $\bP$, albeit with some bias. More formally, we conjecture that 
    $$\lambda_k(\widehat{\bR}^{(i)}) - \lambda_k(\bP^{(i)}) - \eta_m \overset{d}{\rightarrow} \mathcal{N}(0, \sigma^2)$$
    where $\eta_m=O_P(1/\sqrt{m})$.
    These results dovetail nicely the theory of \cite{Tang2018b}, which presents a normality result, with bias, for the eigenvalues of stochastic blockmodel graphs. To illustrate the bias-reduction impact of multiple graphs in our case, consider the problem of eigenvalue estimation in both the single and multiple-graph setting. Figure~\ref{fig:eigs-m-hist} shows that as the sample size increases, the distribution of the two leading eigenvalues of a SBM graph estimated with MASE approaches the true eigenvalues. When the number of vertices increases but the sample size remains constant, this is not the case, as observed in Figure~\ref{fig:eigs-n-hist}.

    \begin{figure}
        \centering
        \includegraphics[width=\textwidth]{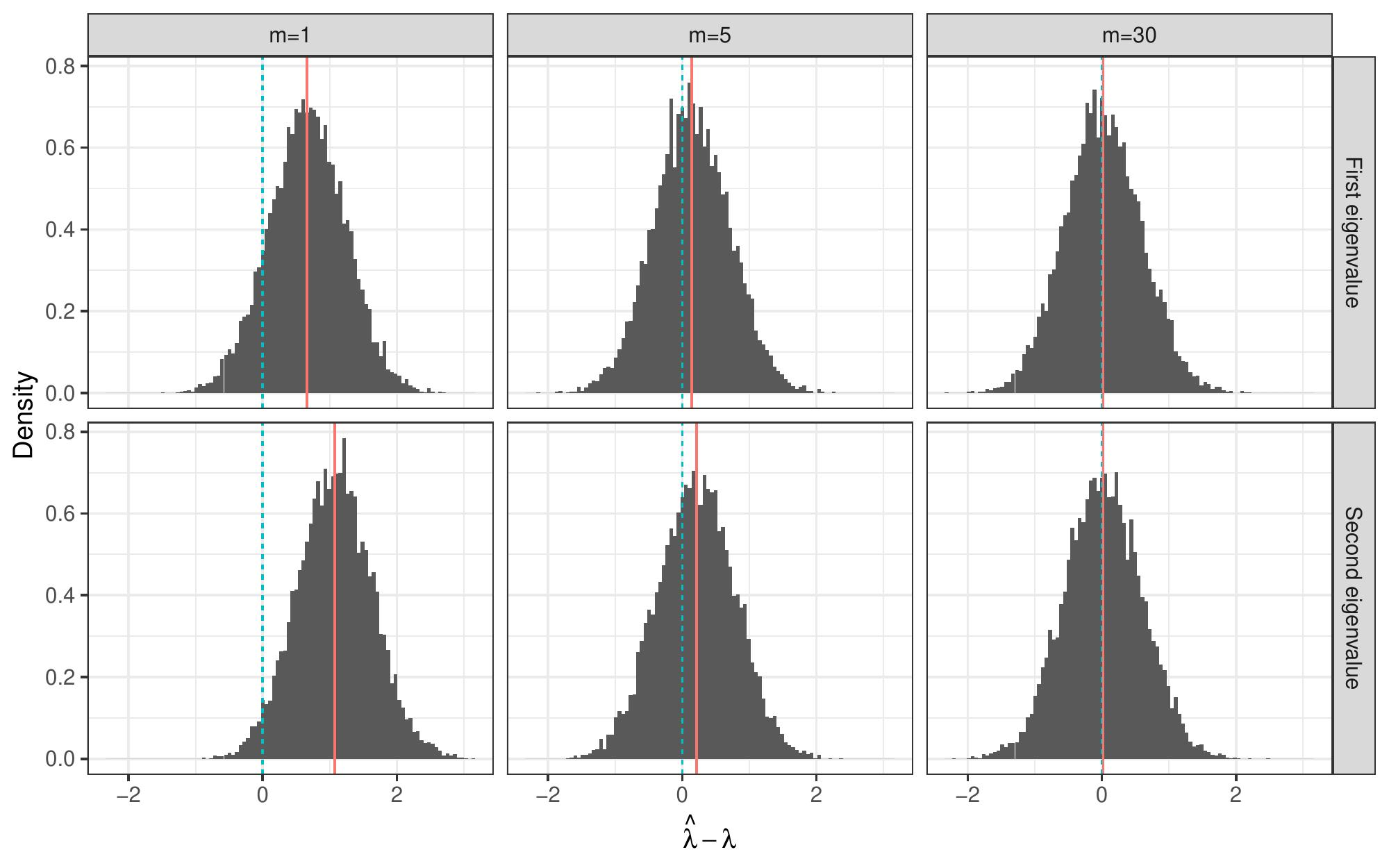}
        \includegraphics[width=\textwidth]{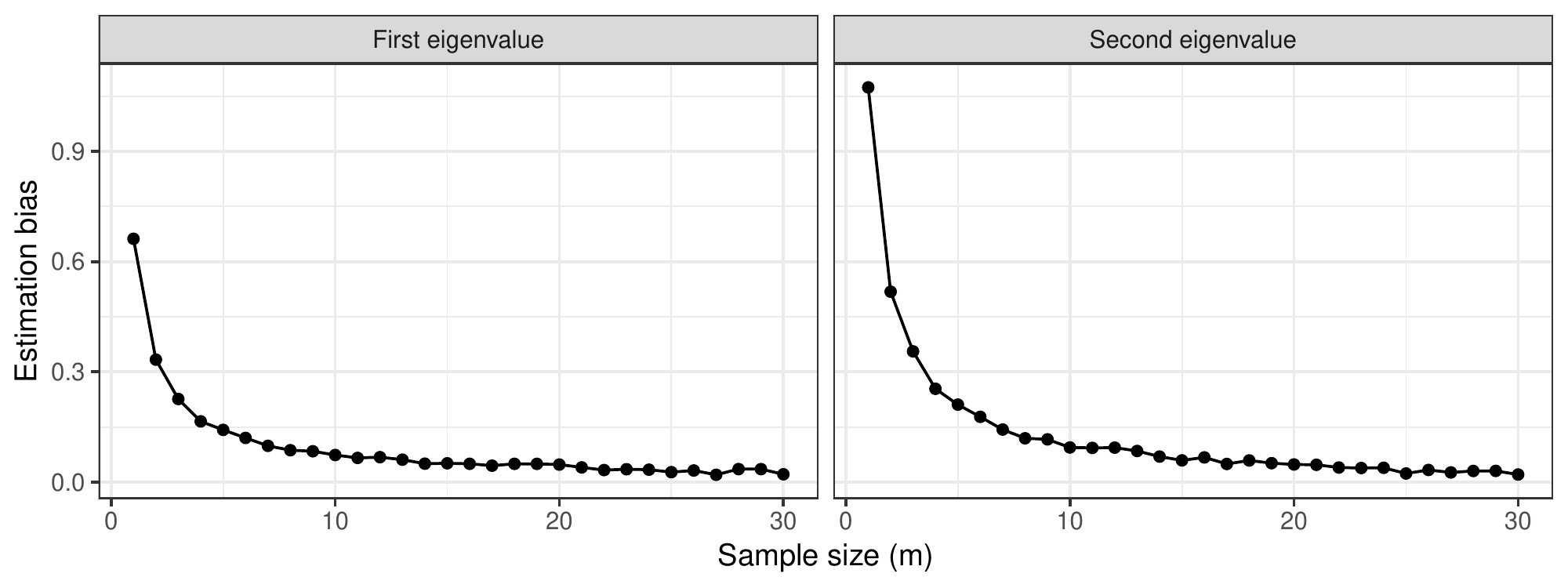}
        \caption{The top panel shows the distribution of the difference between the MASE estimates and true eigenvalues of a single graph. The graphs in the sample are distributed according to a two-block multilayer SBM with $n=300$ vertices, block connection probability matrix $\bB^{(i)}=0.3\bI + 0.1\textbf{1}\textbf{1}^\top $. The distributions appear to be gaussian, and as the sample size $m$ increases, the mean of the distribution (red solid line) approaches to zero (blue dashed line). This phenomenon is also observed in the bottom panel, which shows that the estimation bias decreases with the sample size.}
        \label{fig:eigs-m-hist}
    \end{figure}
    
    \begin{figure}
        \centering
        \includegraphics[width=\textwidth]{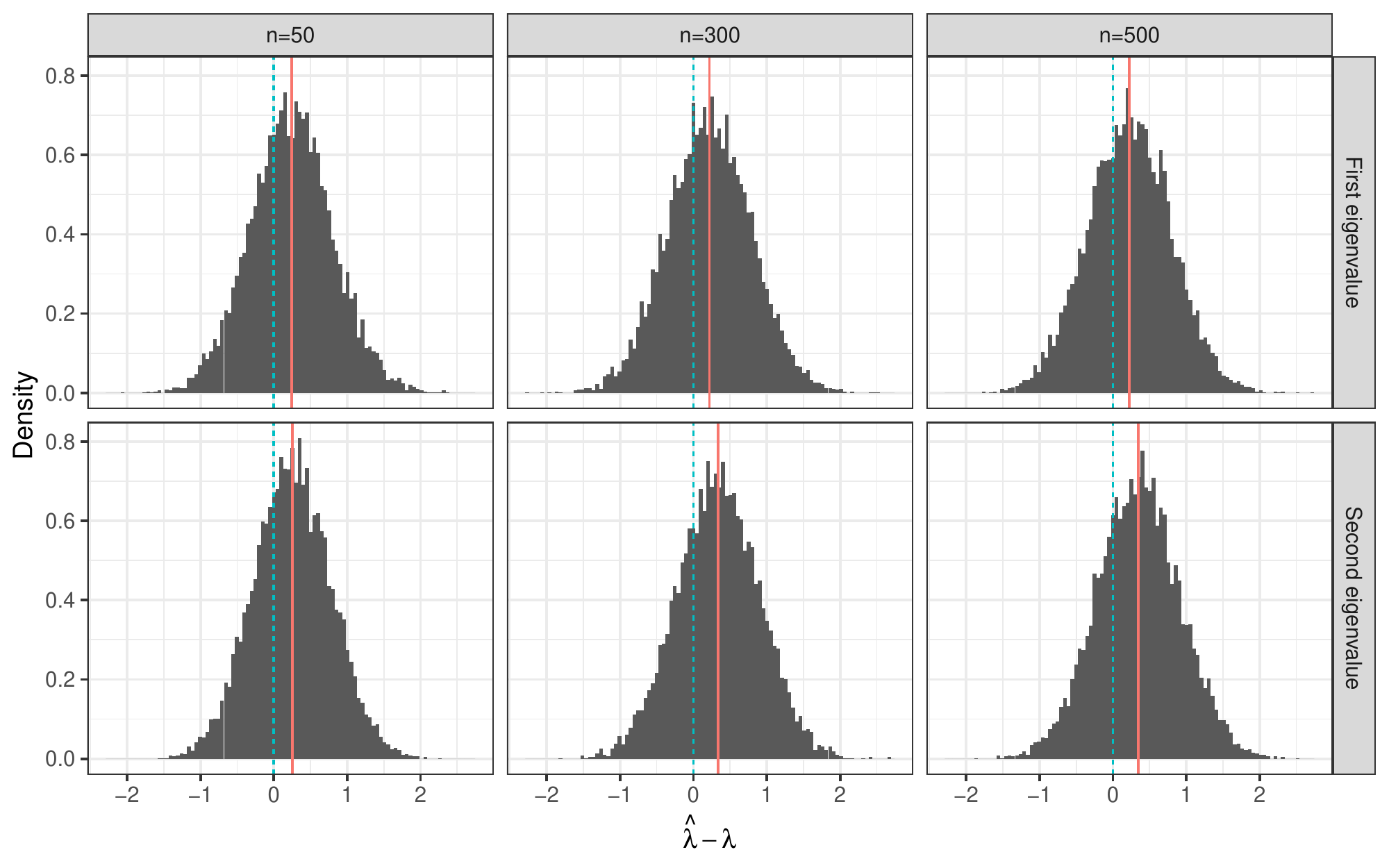}
        \includegraphics[width=\textwidth]{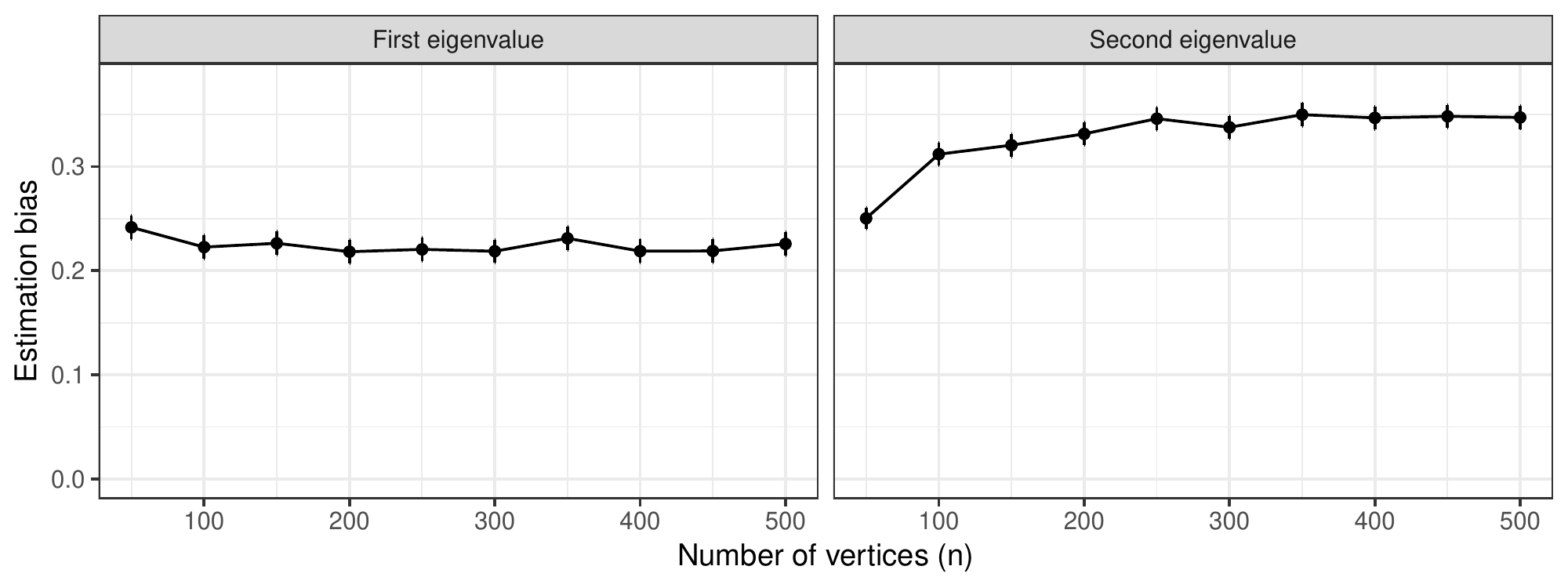}
        \caption{The top panel shows the distribution of the difference between the MASE estimates and true eigenvalues of a single graph. The graphs in the sample are distributed according to a two-block multilayer SBM, block connection probability matrix $\bB^{(i)}=0.3\bI + 0.1\textbf{1}\textbf{1}^\top $, and a fixed sample size $m=3$ but different number of vertices $n$. The estimation bias remains positive even with a large graph size (bottom panel), and for any fixed $n$ the mean of the distributions (red solid line) is away from zero (blue dashed line), but the distributions appear to be gaussian.}
        \label{fig:eigs-n-hist}
    \end{figure}
	


	\section{Simulations\label{sec:simulations}}
	In this section, we study the empirical performance of MASE for estimating the parameters of the COSIE model, as well the performance of this method in subsequent inference tasks, including community detection, graph classification and hypothesis testing. In all cases, we compare with state-of-the-art models and methods for multiple-graph data as baselines. The results show that the COSIE model is able to  effectively handle heterogeneous distributions of graphs and leverage the information across all of them, and demonstrates a good empirical performance in different settings.
	
	 \subsection{Subspace estimation error}
	 We first study the performance in estimating the common subspace $\bV$ by using the estimator $\widehat{\bV}$ obtained by MASE in a setting where the number of graphs $m$ increases. Given a pair of matrices with orthonormal columns $\bV$ and $\widehat{\bV}$, we measure the distance between their invariant subspaces via the spectral norm of the difference between the projections, given by $\|\widehat{\bV}  \widehat{\bV}^\top  - \bV  \bV^\top \|$. This distance is zero only when there exist an orthogonal matrix $W\in\mathcal{O}_d$ such that $\widehat{\bV}=\bV\bW$.
	 
	 Given a sample of graphs $\bA^{(1)},\ldots, \bA^{(m)}$, the ASE of the average of the graphs, given by
	 $$\bar{\bA} = \frac{1}{m}\sum_{i=1}^m\bA^{(i)},$$
	 provides a good estimator of the invariant subspace when all the graphs have the same expectation matrix \citep{Tang2018}.
	 However, when each graph has a different expected value, $\bar{\bA}$ approaches to the population average, which might not contain useful information about the structure of the invariant subspace. MASE, on the other hand, is able to handle heterogeneous structure, and does not require such strict assumptions. We compare this method with the performance of MASE these two settings, i.e.,~equal and different distribution of the graphs. 
	 
	 We simulate graphs from a three-block multilayer stochastic blockmodel with membership matrix $\bZ\in\{0,1\}^{n\times 3}$ and connectivity matrices $\bB^{(1)}, \ldots, \bB^{(m)}$. The number of vertices is fixed to $3^6=729$, and equal-sized communities. We simulate two scenarios. In the first scenario, all graphs have the same connectivity,  given by
	    \begin{align*}
            \bB^{(i)} = \left( \begin{array}{ccc}
                .4 & .1 & .1  \\
                .1 & .4 & .2  \\
                .1 & .2  & .3
            \end{array}\right), \quad\quad \forall i\in[m].
        \end{align*} 
        In the second scenario, each matrix is generated randomly by independently sampling the entries as  $\bB^{(i)}_{uv} \sim U (0,1)$, for each $1\leq u\leq v \leq 3$ and $i\in[m]$.
        
        We then compare the performance of the different methods as the number of graphs increases. We consider both the unscaled and scaled versions of MASE, and ASE on the mean of the adjacency matrices (ASE(mean)). Figure~\ref{fig:sim-subspace-error} shows the average subspace estimation error of several Monte Carlo experiments (25 simulations in the first scenario, and 100 in the second). In the first scenario, the three  methods exhibit a similar performance, decreasing the error as a function of the sample size, and although ASE(mean) has a slightly better performance, the performance of MASE is statistically not worse than ASE(mean). In the second scenario, when all the graphs have different expected values, ASE(mean)  performs poorly, even with a large number of graphs, while MASE mantains a small error and is still able to improve the performance with larger sample size. These results corroborate our theory that MASE is effective in leveraging the information across multiple heterogeneous graphs, and it performs similarly to the ASE on the average of the graphs when the sample has a homogeneous distribution. 
        
        The results of Figure~\ref{fig:sim-subspace-error} also show that in some circumstances the scaled MASE can perform significantly better than its unscaled counterpart in estimating the common invariant subspace. When some eigenvalues of a matrix $\bB^{(i)}$ are close to zero, the estimation of the eigenvectors corresponding to the smallest non-zero eigenvalues is difficult.
        The eigenvalue scaling in ASE can help MASE in discerning the dimensions of $\widehat{\bV}^{(i)}$ that are more accurately estimated, which might explain the improved performance of the scaled MASE in this simulation. Our current theoretical results only address the unscaled MASE, but the simulations encourage extensions of the analysis to the scaled MASE.

	 \begin{figure}
	     \centering
	     \includegraphics[width=\textwidth]{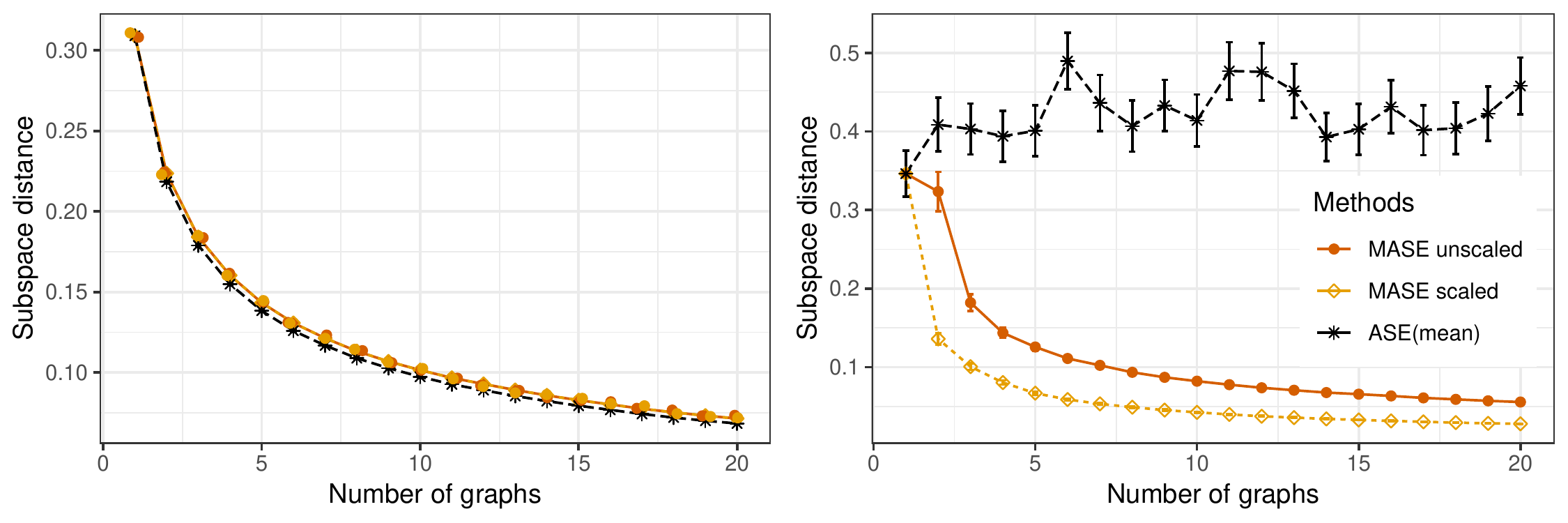}
	     \caption{Distance between the true and estimated invariant subspaces computed with MASE (scaled and unscaled ASE) and ASE of the mean of the graphs for a sample of graphs distributed according to a multilayer SBM with 3 communities. On the left panel, 
	     all graphs have the same expected matrix. The three methods perform almost the same, reducing the error as the sample size increases. On the right panel, the connection probabilities of the SBM are chosen uniformly at random for each graph. ASE on the average adjacency matrix performs poorly, while MASE still improves with a larger sample size.}
	     \label{fig:sim-subspace-error}
	 \end{figure}
   
	 \subsection{Modeling heterogeneous graphs\label{sec:sim-classif}}
    Now, we study the performance of MASE in modeling the structure of graphs with different distributions. First, simulated graphs are distributed according to the  multilayer SBM with two blocks, but now each graph $i$ has an associated class label $y_i\in\{1,2,3,4\}$. The connection matrices $\bB^{(i)}$ have four different classes depending on the label, so if $y_{i}=k$ then
    \begin{equation*}
        \bB^{(i)} = 0.25(\mathbf{1}\mathbf{1}^\top ) + \alpha \bC^{(k)},
    \end{equation*}
    where $\alpha\geq 0$ is a constant that controls the separation between the classes, and the matrices $\bC^{(k)}$ are defined as
        \begin{align*}
            \bC^{(1)} = \left( \begin{array}{cc}
                .1 & 0 \\
                0 & .1
            \end{array}\right)   & \quad   & \bC^{(2)} = -\left( \begin{array}{cc}
                .1 & 0 \\
                0 & .1
            \end{array}\right)\\
            \bC^{(3)} = \left( \begin{array}{cc}
                .1 & 0 \\
                0 & 0
            \end{array}\right)   & \quad   & \bC^{(4)} = \left( \begin{array}{cc}
                0 & 0 \\
                0 & .1
            \end{array}\right).
        \end{align*}
        This choice of graphs allows to simultaneously study the effects on the difference between the classes and the smallest eigenvalue of the score matrices. When $\alpha=0$, all score matrices are the same, but also the smallest eigenvalue is zero, so both subspace estimation and graph classification are hard problems, and as $\alpha$ increases, the signal for these problems improves.
        
        We simulate 40 graphs, with 10 on each class, all with $n=256$ vertices and equal-sized communities. We compare the performance of MASE with respect to other latent space approaches for modeling multiple graphs: joint embedding of graphs (JE) \citep{Wang2017}, multiple RDPG (MRDPG) \citep{Nielsen2018}, and the omnibus embedding (OMNI) \citep{Levin2017}. The first two methods, JE and MRDPG, are based on a model related to COSIE. In particular, the expected value of each adjacency matrix is described as
        \[\bP^{(i)} = \bH\bLambda^{(i)}\bH^\top ,\]
        where $\bH$ is a $n\times d$ matrix, and $\bLambda$ is a $d\times d$ diagonal matrix. In JE, the matrix $\bH$ is only restricted to have columns with norm 1, which makes the model non-identifiable, while MRDPG imposes further restrictions on $\bH$ to be orthogonal and $\bLambda$ to have nonnegative entries. Both models are fitted by optimizing a least squares loss function to obtain estimates $\widehat{\bH}$ and $\{\widehat{\bLambda}^{(i)}\}_{i=1}^m$. The omnibus embedding obtains estimates for individual latent positions for each graph by jointly embedding the graphs into the same space. Given a sample of $m$ adjacency matrices, the omnibus embedding obtains an omnibus matrix $\bM\in\real^{mn\times mn}$, in which each off-diagonal submatrix is the average of each pair of graphs in the sample, so that
        \[\bM_{u + (i-1)m, v+ (j-1)m} = \frac{\left(\bA^{(j)}_{uv} + \bA^{(i)}_{uv}\right)}{2}.\]
        The omnibus embedding is generated by the ASE of $\bM$, and each of the $i$th blocks of $n$ rows of the embedding corresponds to the estimated latent positions of the $i$th graph, such that
        \[\left((\widehat{\bX}^{(1,\text{OMNI})})^\top ,\ldots, (\widehat{\bX}^{(m,\text{OMNI})})^\top   \right)^\top  = \text{ASE}(\bM).\]
        For MASE and OMNI, we use an embedding dimension $d=2$; JRDPG and MRDPG represent the scores of each graph with a diagonal matrix, so in order to capture all the variability in the data we use $d=3$ for the embedding dimension.
        In all comparisons, the two versions of MASE (scaled and unscaled ASE) performed very similarly, and therefore we only report the results of unscaled MASE.
        
        \subsubsection{Semi-supervised graph classification}

        First, we measure the accuracy of these methods in graph classification. We use the simulated graphs as training data; for test data, we generate a second set of graphs with the same characteristics as the first.
        In order to determine classification accuracy, we first use each method to obtain an embedding of all the graphs, including training and test data.  We then use the Frobenius distance between the individual score matrices of each graph to build a 1-nearest neighbor classifier for the test data, and use the labels of the training data to classify the test data. The individual score matrices correspond to $\{\widehat{\bR}^{(i)}\}_{i=1}^m$
        in MASE, $\{\widehat{\bH}^{(i)}\}_{i=1}^m$ in JE and MRDPG, and $\{\widehat{\bX}^{(i,\text{OMNI})}\}_{i=1}^m$ in OMNI. Note that both training and test data are used to generate the estimates of the score matrices; in other words, to produce an {\em unsupervised} joint embedding for all the graphs. Such a joint embedding allows to avoid cumbersome Procrustes alignments that are required by the OMNI and the MRDPG procedures when embedding test and training data separately; neither OMNI and MRDPG have established out-of-sample extensions for their respective embeddings. MASE can generate the embedding without using the test data by constructing the invariant subspace matrix $\widehat{\bV}$ only using the training data, but we use the semi-supervised procedure described before for consistency between the different methods.
        
        Figure~\ref{sim:class-acc} shows the average classification accuracy as a function of $\alpha$ for 50 Monte Carlo simulations. For all methods, the accuracy increases as the separation between each class model increases with $\alpha$,  and MASE performs best among all the methods. This result is not surprising, considering the fact that MASE has better flexibility to model heterogeneous graphs. 
        The performance of OMNI is also excellent, but it is worth noting that OMNI employs a much larger number of parameters ($n \times d$) to represent each graph. JE is based on a model that can also represent each expected adjacency correctly, but the non-convexity and over-parametertization of the objective function can make the optimization problem challenging, 
        and the performance depends on the initial value, which is random, resulting in inferior performance on average. The identifiability constraints of MRDPG limit the type of graphs that this model can represent, and this method is never able to separate two of the classes correctly, resulting in a limited performance in practice.
        
        
        	 \begin{figure}
	     \centering
	     \includegraphics[width=0.7\textwidth]{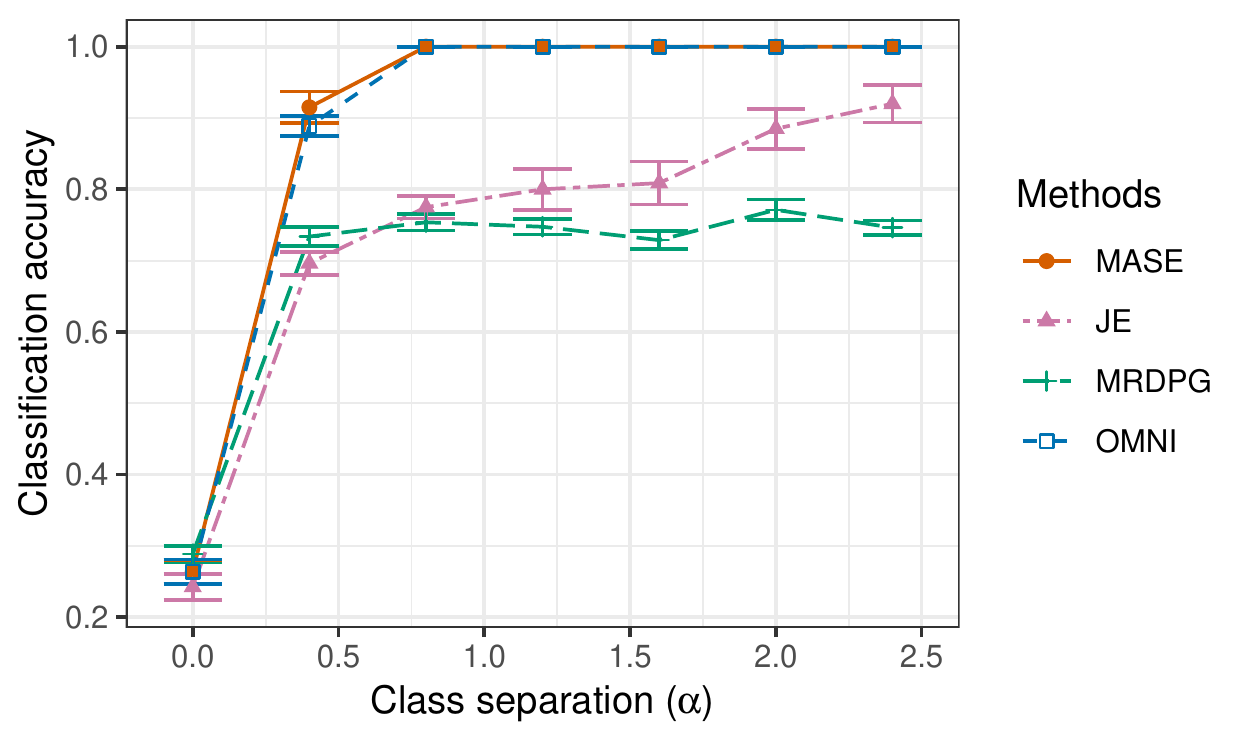}
	     \caption{Out-of-sample classification accuracy as a function of the difference between classes for different graph embedding methods. Graphs in the sample are distributed as a multilayer SBM with $n=256$ vertices, two communities, four different connectivity matrices that correspond to the class labels, and a training and test samples with $m=40$ graphs. As the separation between the classes increase, all methods improve performance, but MASE and OMNI show the most gains.}
	     \label{sim:class-acc}
	 \end{figure}

	 
        \subsubsection{Model estimation}
        We also compare the performance of different methods in terms of approximating the probability matrices of the model. For each method and each graph, we obtain an estimate $\widehat{\bP}^{(i)}$ of the generative probability matrix of the model $\bP^{(i)}$, and measure the model estimation error with the normalized mean squared error as
        \begin{equation*}
            L(\widehat{\bP}^{(i)}, \bP^{(i)}) = \frac{\|\widehat{\bP}^{(i)} - \bP^{(i)}\|_F}{\|\bP^{(i)}\|_F},
        \end{equation*}
        and report the average estimation error over all the graphs.
         Figure~\ref{fig:sim-P-Phat-err} reports the average Monte Carlo estimation error as a function of $\alpha$. Again, we observe that MASE has the best performance among all the methods. This is not surprising, considering that MASE is designed for this model, but it shows the limitations of the other methods. JE also shows an improvement in estimation as the separation between the classes increases, since this is the only other method that can represent correctly the four classes, but as in the classification task, the variance is larger. While OMNI can discriminate the graphs correctly, as previously observed, the error in estimating the generative probability matrix is large. MRDPG is again not able to succeed due to the model limitations.
         
         \begin{figure}
	     \centering
    	     \includegraphics[width=0.7\textwidth]{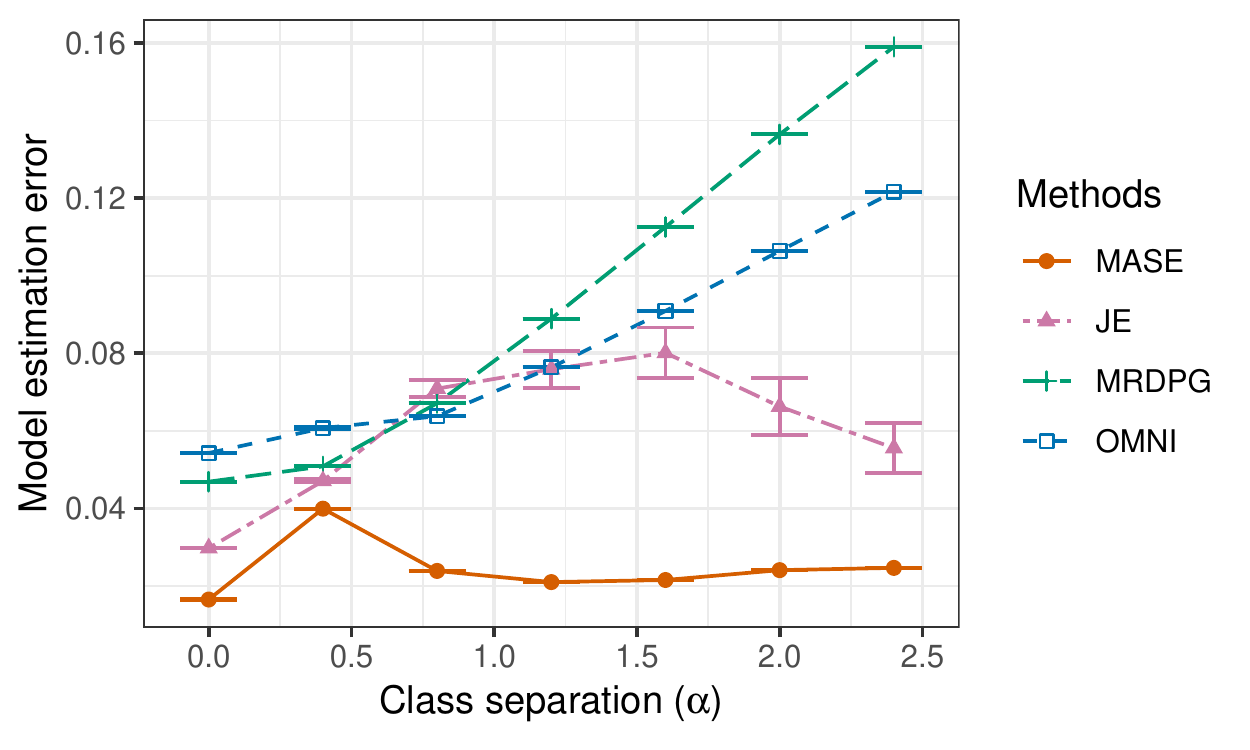}
	     \caption{Average normalized mean squared error of estimating the expected adjacency matrix of a sample of graphs using different embedding methods. Graphs in the sample are distributed as a multilayer SBM with $n=256$ vertices, two communities, four different classes of connectivity matrices, and a training and test samples with $m=40$ graphs. MASE and JE are the only methods that are flexible enough to capture the heterogeneous structure of the graphs, but MASE shows superior performance among these two. As the graphs become more different, the error of MRDPG and OMNI increases.}
	     \label{fig:sim-P-Phat-err}
	 \end{figure}

        \subsubsection{Community detection}
        We use the joint latent positions obtained by each method to perform community detection, by clustering the rows of each embedding using a gaussian mixture model (GMM). For MASE, we use the matrix $\widehat{\bV}$ as the estimated joint latent positions, while for JE and MRDPG we use the matrix $\widehat{\bH}$. For OMNI, we take the average of the individual latent positions of the graphs $\sum_{i=1}^m\widehat{\bX}^{(m,\text{OMNI})}$. In all cases, because the number of communities is known, we fit two clusters using GMM, and compare with the true communities of the multilayer SBM. 
        
        The average accuracy in clustering the communities is reported in Figure~\ref{fig:sim-community}. Spectral clustering methods require to have enough separation on the smallest non-zero eigenvalue. This is the case for MASE and OMNI, which show poor performance with small $\alpha$, but once the signal is strong enough, both methods perform excellent.
        MRDPG shows a good performance, even for small $\alpha$, which could be due to the optimization procedure employed by the method in fitting $\bV$. As in all the other tasks, JE shows a high variance in performance.

	 \begin{figure}
	     \centering
	     \includegraphics[width=0.7\textwidth]{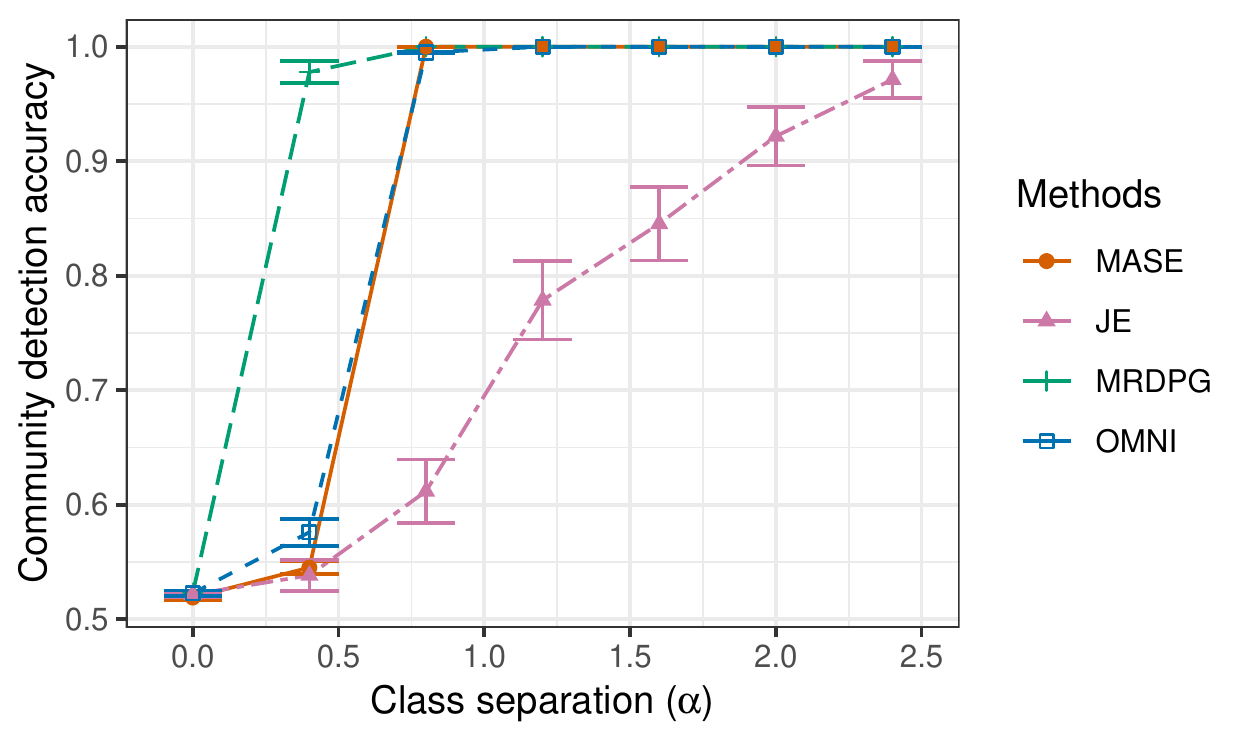}
	     \caption{Performance in community detection of a gaussian mixture clustering  applied to a common set of vertex latent positions estimated with different methods, for a sample of $m=40$ two-block multilayer SBM graphs, with $n=256$ vertices, and four classes of connectivity matrices. The class separation controls the magnitude of the smallest eigenvalue of the graphs, and as this increases, all methods show better accuracy. MRDPG, which is based on a non-convex optimization problem, performance the best for small $\alpha$, but as long as $\alpha$ is large enough, MASE and OMNI also show a great performance.}
	     \label{fig:sim-community}
	 \end{figure}
	 
	 \subsection{Graph hypothesis testing \label{sec:sim-testing}}

Our goal is to test the hypothesis that for a given pair of random adjacency matrices $\bA^{(1)}$ and $\bA^{(2)}$, the underlying probability matrices $\bP^{(1)}$ and $\bP^{(2)}$ are the same. Using the COSIE model, for any pair of matrices $\bP^{(1)}$ and $\bP^{(2)},$ there exists an embedding dimension $d$ and a matrix with orthonormal columns $\bV\in\real^{n\times d}$ such that $\bP^{(i)}=\bV \bR^{(i)} \bV^\top $. Therefore, our framework reduces the problem to testing the hypothesis $H_0: \bR^{(1)} = \bR^{(2)}$. We evaluate the performance of our method and compare it with the omnibus embedding (OMNI) method of \cite{Levin2017}, which is one of the few principled methods for this problem in latent space models.

	 We proceed by generating $\bA^{(k)} \sim \operatorname{Ber}(\bP^{(k)})$ with $\bP^{(k)} = \bZ^{(k)}\bB^{(k)}\bZ^{(k)^\top}$ a mixed membership SBM (MMSBM) with three communities \citep{Airoldi2007}, such that for each $i=1, \ldots, n$, the row that corresponds to vertex $i$ is distributed as $\bZ^{(k)}_{i\cdot} \sim \operatorname{Dir}(0.1, 3)$. The matrix $\bB^{(1)}\in\real^{3\times 3}$ of the first model is fixed as $\bB^{(1)} = 0.3 \bI + 0.1(\mathbf{1}\mathbf{1}^\top)$, where $\mathbf{1}$ is the  3-dimensional vector with all entries equal to one. We simulate two scenarios.
	 \begin{enumerate}
	     \item Same community assignments, but different connectivity matrices: Fix the block memberships in both graphs to be the same, $\bZ^{(1)} = \bZ^{(2)}$, and vary $\|\bB^{(1)}-\bB^{(2)}\|_F$ by increasing $\bB^{(2)}_{11}$ while keeping all the entries of $\bB^{(2)}$ equal to $\bB^{(1)}$.
	     \item Different community assignments for some of the vertices, same connectivity matrices: fix $\bB^{(2)}= \bB^{(1)}$, and sample the memberships of the first $t$ vertices independently in the two graphs, while keeping the other $n-t$ memberships equal: that is, $\{\bZ^{(1)}_{i\cdot}, \bZ^{(2)}_{j\cdot}: i \in[n], j\in[t]\}$ are i.i.d  Dir$((0.1, 3))$, and $\bZ^{(2)}_{k\cdot}=\bZ^{(1)}_{k\cdot}$ for $k>t$.
	 \end{enumerate}
The first scenario can be represented as a COSIE model with dimension $d=d_1=d_2=3$, where $d_i$ is the rank of the matrix $\bP^{(i)}$, but for the second one, to exactly represent these graphs in the COSIE model, we use $d_1=d_2=3$ and $d=5$. 
	 
	 To test the null hypothesis, we use the square Frobenius norm of the difference  between the score matrices $\|\widehat{\bR}^{(1)} - \widehat{\bR}^{(2)}\|_F^2$. A similar test statistic is used for OMNI by calculating the distance between the estimated latent positions $\|\widehat{\bX}^{(1,\text{OMNI})}-\widehat{\bX}^{(2,\text{OMNI})}\|_F^2$. For both MASE and OMNI, we estimate the exact distribution of the test statistic via Monte Carlo simulations using the correct expected value of each graph $\bP^{(1)}$ and $\bP^{(2)}$. For MASE, we also evaluate the performance of empirical p-values calculated using a parametric bootstrap approach, or the asymptotic null distribution of the score matrices, which are described below.
	 
	 \begin{enumerate}
	     \item[a)] A semiparametric bootstrap method introduced in \cite{Tang2014}  uses  estimates $\widetilde{\bP}^{(1)}$ and $\widetilde{\bP}^{(2)}$ (constructed by the ASE of $\bA^{(1)}$ and $\bA^{(2)}$) in place of the true matrices $\bP^{(1)}$ and $\bP^{(2)}$---the latter of which may not, of course, even been known in practice. A total of 1000 independent pairs of adjacency matrices $\widetilde{\bA}^{(1)}$ and $\widetilde{\bA}^{(2)}$ are generated by fixing the distribution of the two graphs to $\widetilde{\bP}^{(1)}$ for the first 500 pairs, and to $\widetilde{\bP}^{(2)}$ for the remaining pairs, and these pairs are used to approximate the empirical null distribution of the test statistic,  and calculate the p-value as the proportion of bootstrapped test statistic values that are smaller than the observed one.
	     
	     \item[b)]
	     When $\bA^{(1)}$ and $\bA^{(2)}$ are both sampled from $\bP^{(1)}$, Theorem~\ref{thm:COSIE_CLT} suggests that the distribution of the difference between the entries of the estimated score matrices  (after a proper orthogonal alignment) is approximated by a multivariate normal, with a mean proportional to the bias term $\bH^{(1)}$ and covariance depending on $\bSigma^{(1)}$.
	     Letting $\widehat{\bP}^{(1)}=\widehat{\bV}\widehat{\bR}^{(1)}\widehat{\bV}^\top$ be the estimate for $\bP^{(1)}$ obtained by MASE, an estimate $\widehat{\bSigma}^{(1)}$ of the asymptotic covariance in Theorem~\ref{thm:COSIE_CLT} can be derived from Equation~\eqref{eq:covarianceR} by plugin in $\widehat{\bV}$ and $\widehat{\bP}^{(1)}$. Neglecting the effect of the bias term $\bH^{(1)}$ (which vanishes as the number of graphs increases), we use a random variable $\by\sim \mathcal{N}(\mathbf{0}_r, {2}\widehat{\bSigma}^{(1)})$ to approximate the  null distribution of  $\operatorname{vec}(\bW(\widehat{\bR}^{(1)} - \widehat{\bR}^{(2)})\bW^\top)$ (with $\bW$ as defined in Theorem~\ref{thm:COSIE_CLT}). Therefore, the null distribution of the test statistic $\|\widehat{\bR}^{(1)} - \widehat{\bR}^{(2)}\|_F^2=\|\bW(\widehat{\bR}^{(1)} - \widehat{\bR}^{(2)})\bW^\top\|_F^2$ is  approximately distributed as  a generalized chi-square distribution (since it is a function of the squared entries of $\by$), or it can be estimated using Monte Carlo simulations of $\by$. The same process is repeated by using $\widehat\bP^{(2)}$, after which the empirical p-value can be estimated using a mixture of the empirical distributions obtained from $\widehat\bP^{(1)}$ and $\widehat\bP^{(2)}$.
	 \end{enumerate}

	 \begin{figure}
	     \centering
	     \includegraphics[width = \textwidth]{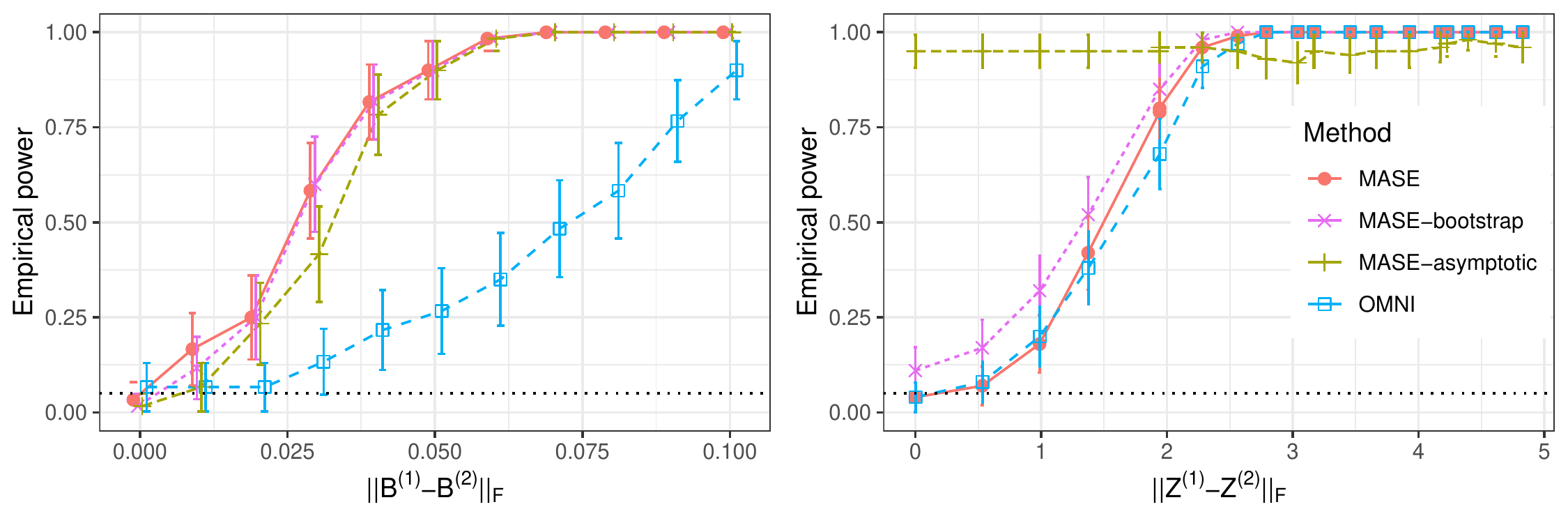}
	     \caption{Empirical power for rejecting the null hypothesis that two graphs have the same distribution as a function of the difference in their parameters at a $0.05$ significance level (black dotted line). Both graphs are sample from a mixed-membership SBM $\bA^{(i)}\sim\operatorname{MMSBM}(\bZ^{(i)}, \bB^{(i)})$. In the left panel, the community assignments $\bZ^{(1)}$ and $\bZ^{(2)}$ are the same, while the connection probability changes on the x-axis. In the right panel, $\bB^{(1)}=\bB^{(2)}$, while the community assignments of a few nodes change. Both methods improve their power as the difference between the graphs get larger, but MASE can effectively use the common subspace structure to outperform OMNI when the difference between the $\bB$ matrices is small. The p-values of MASE are estimated accurately usign a a parametric bootstrap, while the asymptotic distribution is only valid in the first scenario when the common invariant subspace is represented in both graphs.}
	     \label{fig:sim-testing}
	 \end{figure}

	Figure~\ref{fig:sim-testing} shows the result in terms of power for rejecting the null hypothesis that the two graphs are equal (the confidence level was chosen as 0.05). As the separation between null and alternative increases, the result from both methods approach full power when the p-value is calculated using the exact null distribution (MASE and OMNI). The first scenario shows an advantage for MASE, which exploits both the fact that the common subspace is well represented on each graph and that there are a smaller number of parameters to estimate than in the omnibus model. The second scenario is more challenging for MASE since the number of parameters is increased, but it still performs competitively. In both cases, the p-values obtained with the bootstrap method (MASE-bootstrap) perform similarly to the empirical power when using the true $\bP$ matrices.  
	
	The asymptotic distribution of the score matrices provides a computationally efficient way to estimate the p-values, but this approach relies on the validity of the assumptions of Theorem~\ref{thm:COSIE_CLT} and how small the magnitude of the bias term is, which depends on the estimation error of $\bV$. Having only two graphs to estimate the parameters of the model, there is no guarantee that  this bias term will be sufficiently small. Nevertheless, in the first scenario (Figure~\ref{fig:sim-testing}, left panel), the estimated power using the asymptotic distribution (MASE-asymptotic) provides an accurate result; here,  $\bV$ is well represented in both graphs, which helps in reducing its estimation error. In the second scenario (Figure~\ref{fig:sim-testing}, left panel), the dimension of $\bV$ $(d=5)$ is higher than the rank of the expectation of each graph ($d_1=d_2=3$), hence, the score matrices are not full rank, affecting the magnitude of the bias term   and  invalidating some conditions of Theorem~\ref{thm:COSIE_CLT}. As a result,  the estimated p-values with this method appear to be invalid and significantly inflated on this scenario.
	
	\section{Modeling the connectivity of brain networks\label{sec:data}}

	We  evaluate the ability of our method to characterize differences in brain connectivity between subjects using a set of graphs constructed from diffusion magnetic resonance imaging (dMRI). The data corresponds to the HNU1 study \citep{Zuo2014} which consists of dMRI records of 30 different healthy subjects that were scanned ten times each over a period of one month. Based on these recordings, we constructed 300 graphs (one per subject and scan) using the NeuroData's MRI to Graphs (NDMG) pipeline \citep{Kiar2018}. The vertices were registered to the CC200 atlas \citep{Craddock2012}, which identified 200 vertices in the brain. 
	
	We first apply our multiple adjacency spectral embedding method to the HNU1 data. To choose the dimension of the embedding for each individual graph, we use the automatic scree plot selection method of \cite{Zhu2006}, which chooses between 5 and 18 dimensions for each graph (see Figure~\ref{fig:hnu1-hist-d}). The joint model dimension for MASE is chosen based on the scree plot of the singular values of the matrix of concatenated adjacency spectral embeddings. We perform the method using both the scaled and unscaled versions of ASE. In both cases, we observe an elbow on the scree plot at $d=15$ (see Figure~\ref{fig:hnu1-screeplot}), and thus we selected this value for our model.
	
	\begin{figure}
	    \centering
	    \includegraphics[width=0.5\textwidth]{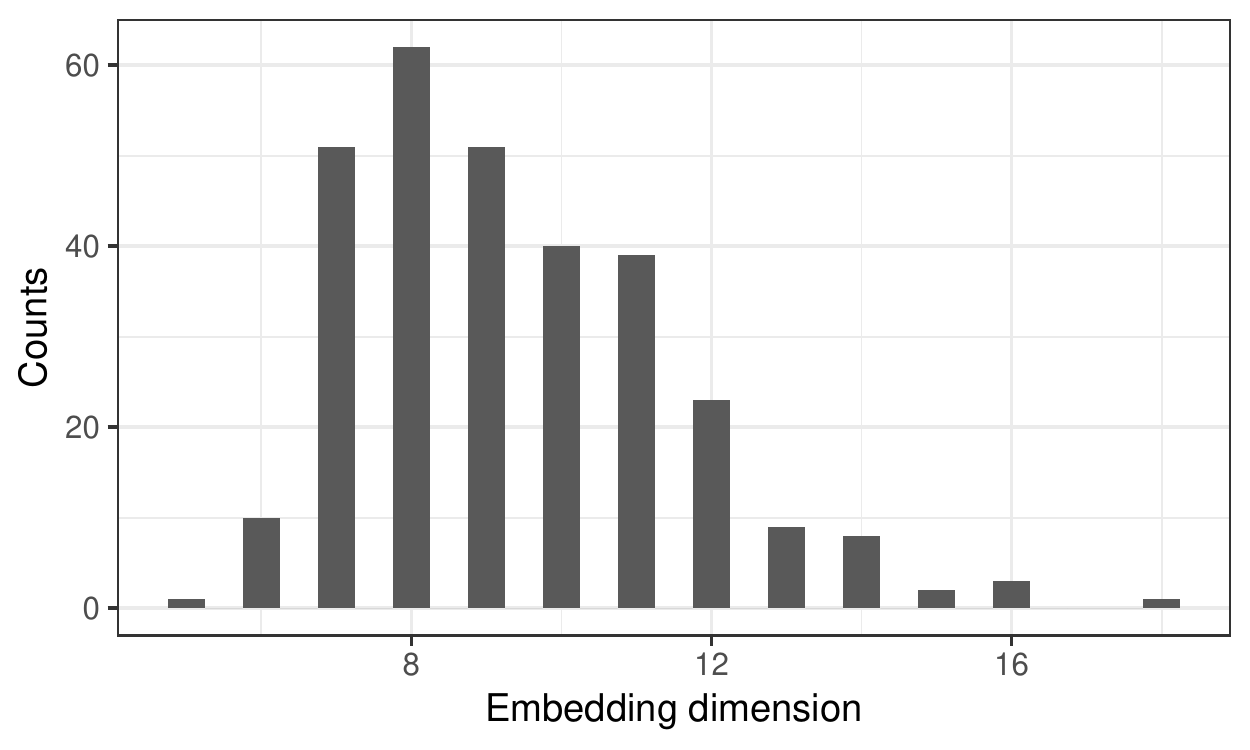}
	    \caption{Histogram of the embedding dimension selected by the method of \cite{Zhu2006} for each of the 300 graphs in the HNU1 data. The graphs are composed by 200 vertices, and the values of the embedding dimensions ranges between 5 and 18.}
	    \label{fig:hnu1-hist-d}
	\end{figure}
	
	\begin{figure}
	    \centering
	    \includegraphics[width=0.48\textwidth]{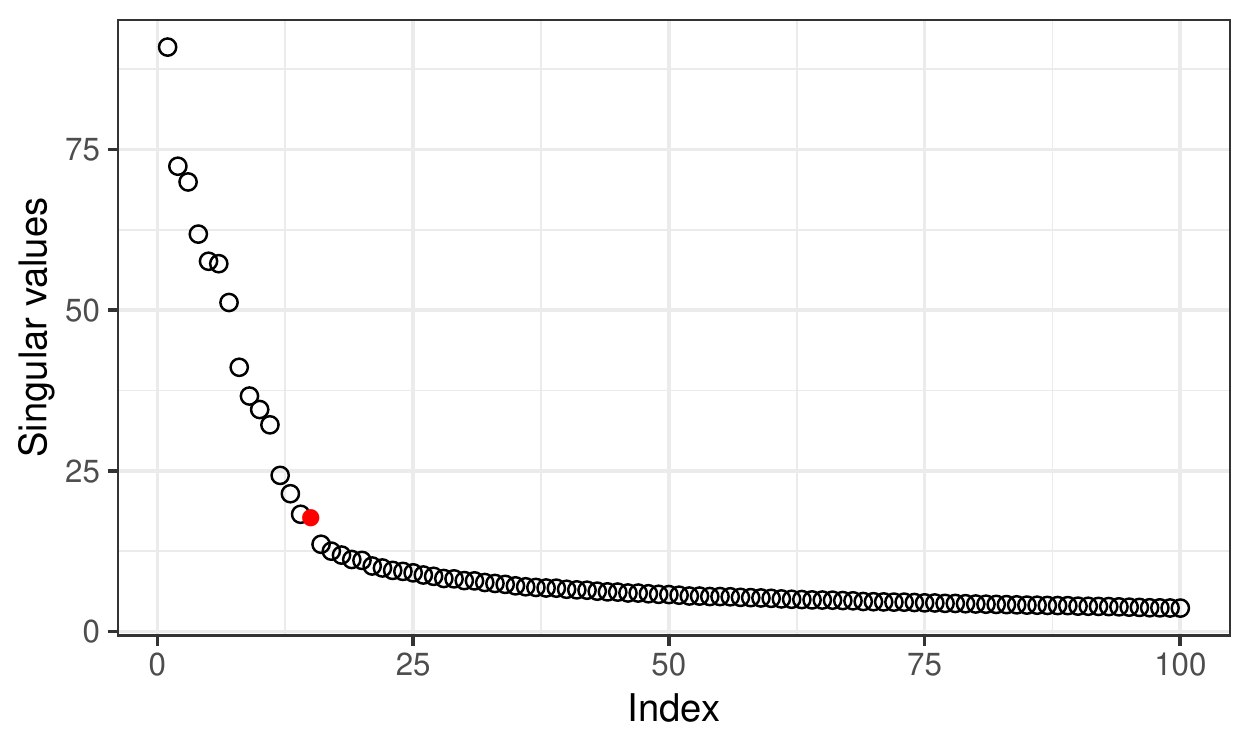}
	    \includegraphics[width=0.48\textwidth]{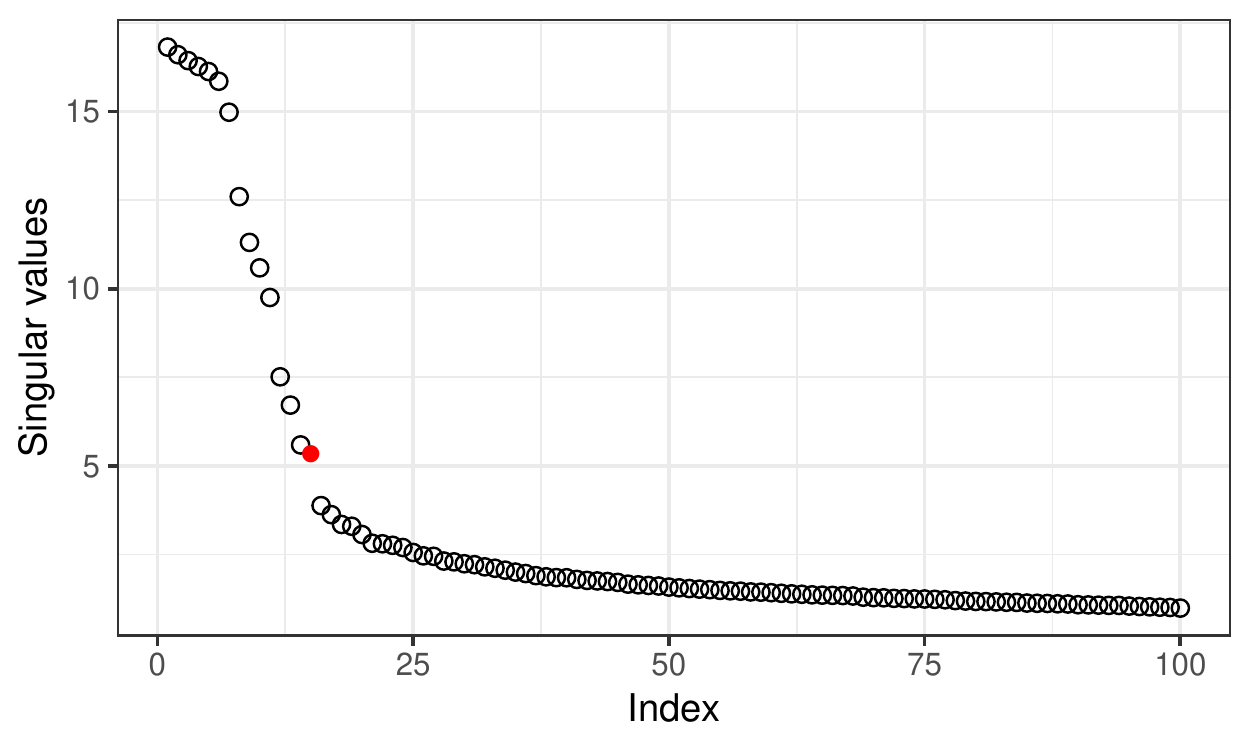}
	    \caption{Scree plot of the top singular values of the concatenated scaled (left) and unscaled (right) adjacency spectral embeddings of the HNU1 data graphs. The plots show the 100 largest  singular values ordered by magnitude. In both scaled and unscaled methods, we identified an elbow at $d=15$.}
	    \label{fig:hnu1-screeplot}
	\end{figure}
	
	When we apply the MASE method to the HNU1 data, we obtain a matrix of joint latent positions $\widehat{\bV}\in\mathbb{\bR}^{200\times 15}$ and a set of symmetric matrices $\{\widehat{\bR}^{(i)}\}_{i=1}^{300}$  of size $15\times 15$ for each individual graph, yielding 120 parameters to represent each graph. Figure~\ref{fig:hnu1-latentpos} shows a three-dimensional plot of the latent positions of the vertices for $\widehat{\bV}$. Most of the vertices of the CC200 atlas are labeled according to their spatial location as left or right hemisphere, and this structure is reflected in Figure~\ref{fig:hnu1-latentpos}---in fact, the plane $v_3=0$ appears to be a useful discriminant boundary between left and right hemisphere. This is additional fortifying evidence that the embedding obtained by MASE is anatomically meaningful. 
	
	\begin{figure}[tbh]
	    \centering
	    \includegraphics[width=0.6\textwidth]{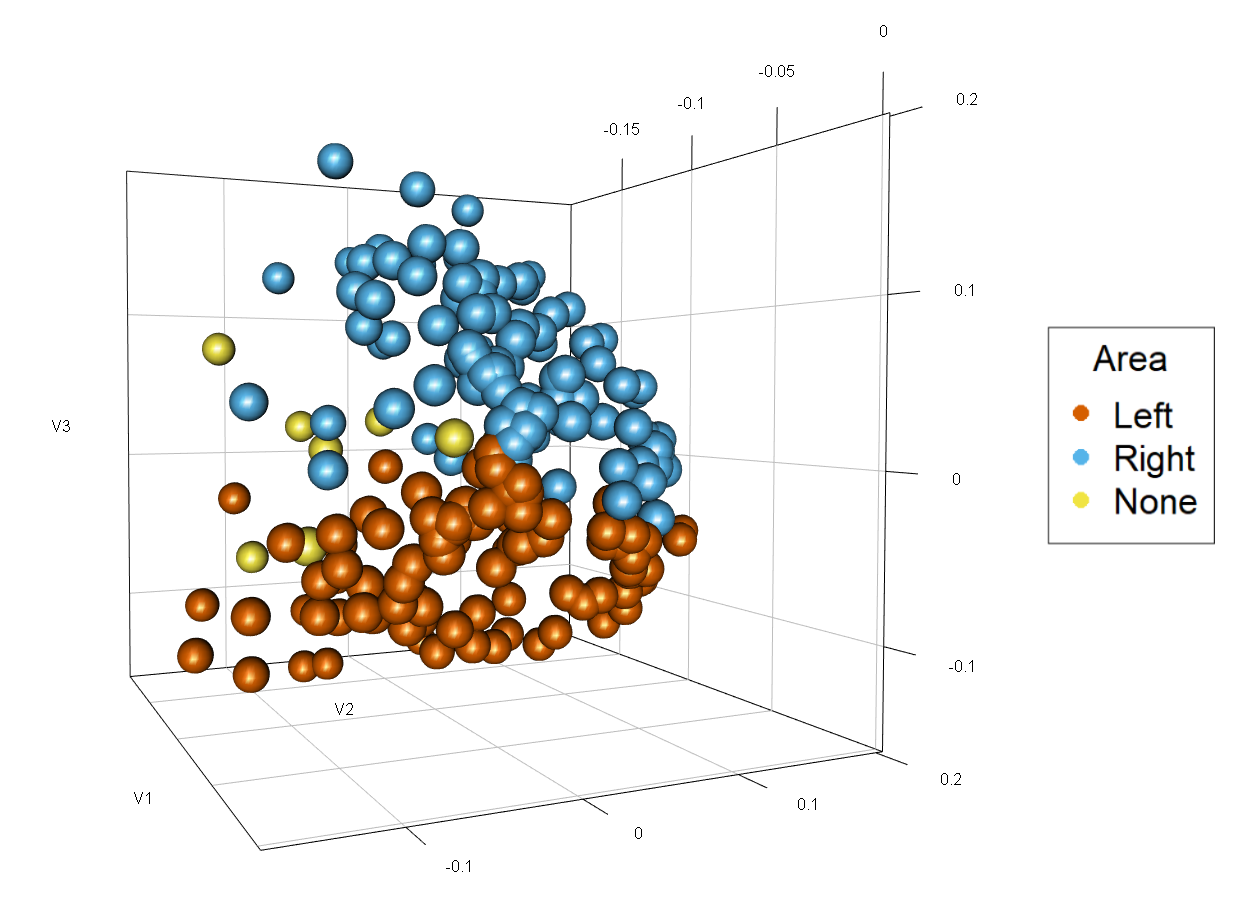}
	    \caption{Estimated latent positions obtained by MASE on the HNU1 data graphs with a 3-dimensional embedding.}
	    \label{fig:hnu1-latentpos}
	\end{figure}
	
	The individual graph parameters $\{\widehat{\bR}^{(i)}\}_{i=1}^{300}$  of dimension $15\times 15$ are difficult to interpret because their values are identifiable only with respect to $\widehat{\bV}$, but also because of the large dimensionality. Nevertheless, the matrix of Frobenius distances $\bD\in\mathbb{\bR}^{300\times 300}$ such that $\bD_{ij}=\|\widehat{\bR}^{(i)} - \widehat{\bR}^{(j)}\|_F$, is invariant to any rotations on $\widehat{\bV}$ according to Proposition~\ref{prop:identifiability}. This matrix is shown in Figure~\ref{fig:hnu1-D}, and reflects that graphs coming from the same individual tend to be more similar to each other than graphs of different individuals. We perform classical multidimensional scaling (CMDS) \citep{Borg2003} on the matrix $\bD$ with 5 dimensions for the scaling. Figure~\ref{fig:hnu1-mds} shows scatter plots of the positions discovered by MDS for a subset of the graphs corresponding to 5 subjects; these subjects were chosen because the distances between them represent a useful snapshot of the variability in the data. We observe that the points representing graphs of the same subject usually cluster together. The location of the points corresponding to subjects 12 and 14 suggests some similarity between their connectomes, while subject 13 shows a larger spread, and hence more variability in the brain network representations.

	\begin{figure}
	    \centering
	    \includegraphics[width=0.5\textwidth]{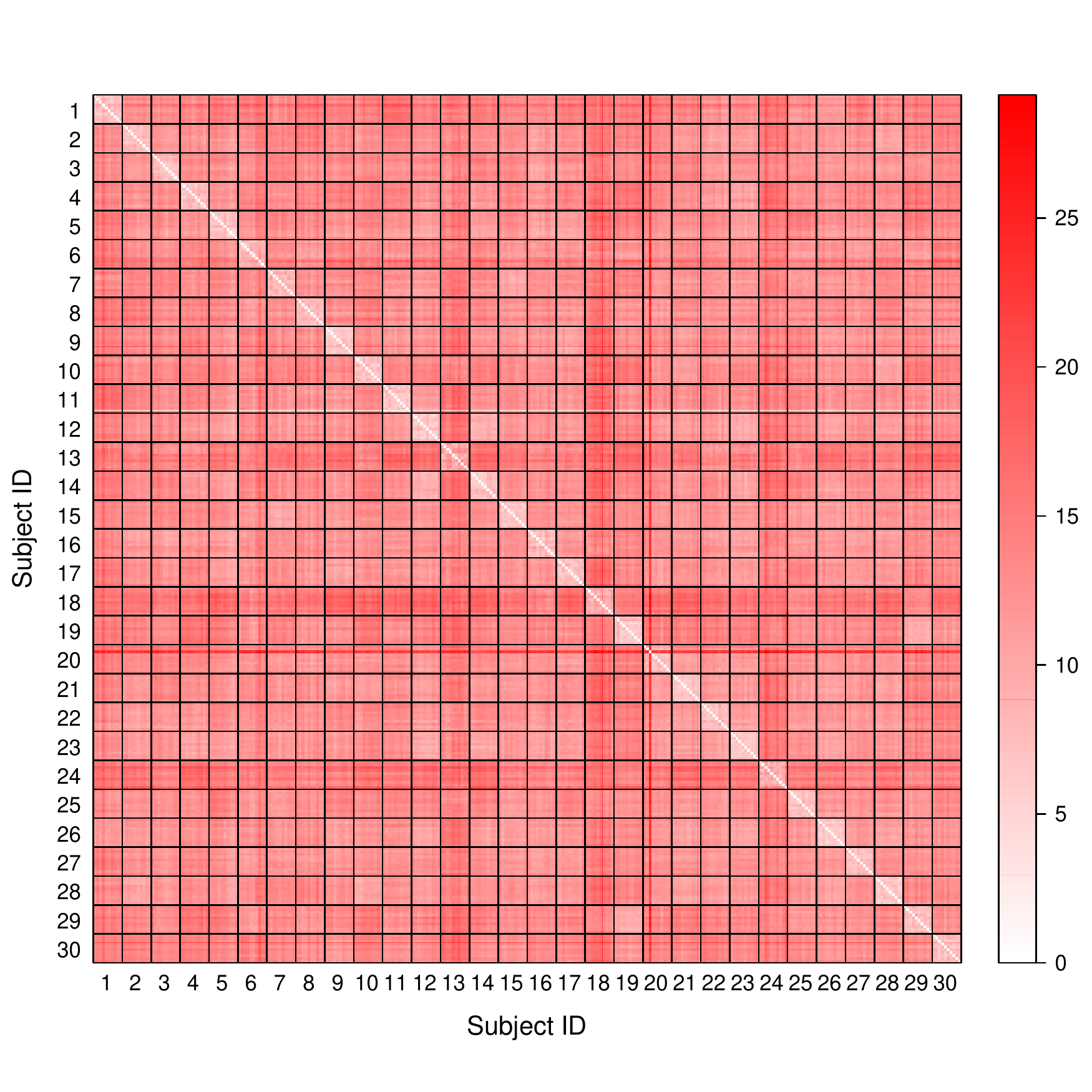}
	    \caption{Matrix of distances between the estimated score parameters.   Each entry of the matrix corresponds to the distance between the MASE estimated score matrices of a pair of graphs from the HNU1 data, composed by 30 subjects with 10 graphs each. Distances between graphs of the same subject are usually smaller than distances within subjects.}
	    \label{fig:hnu1-D}
	\end{figure}
	
	\begin{figure}[tbh]
	    \centering
	    \includegraphics[width=0.6\textwidth]{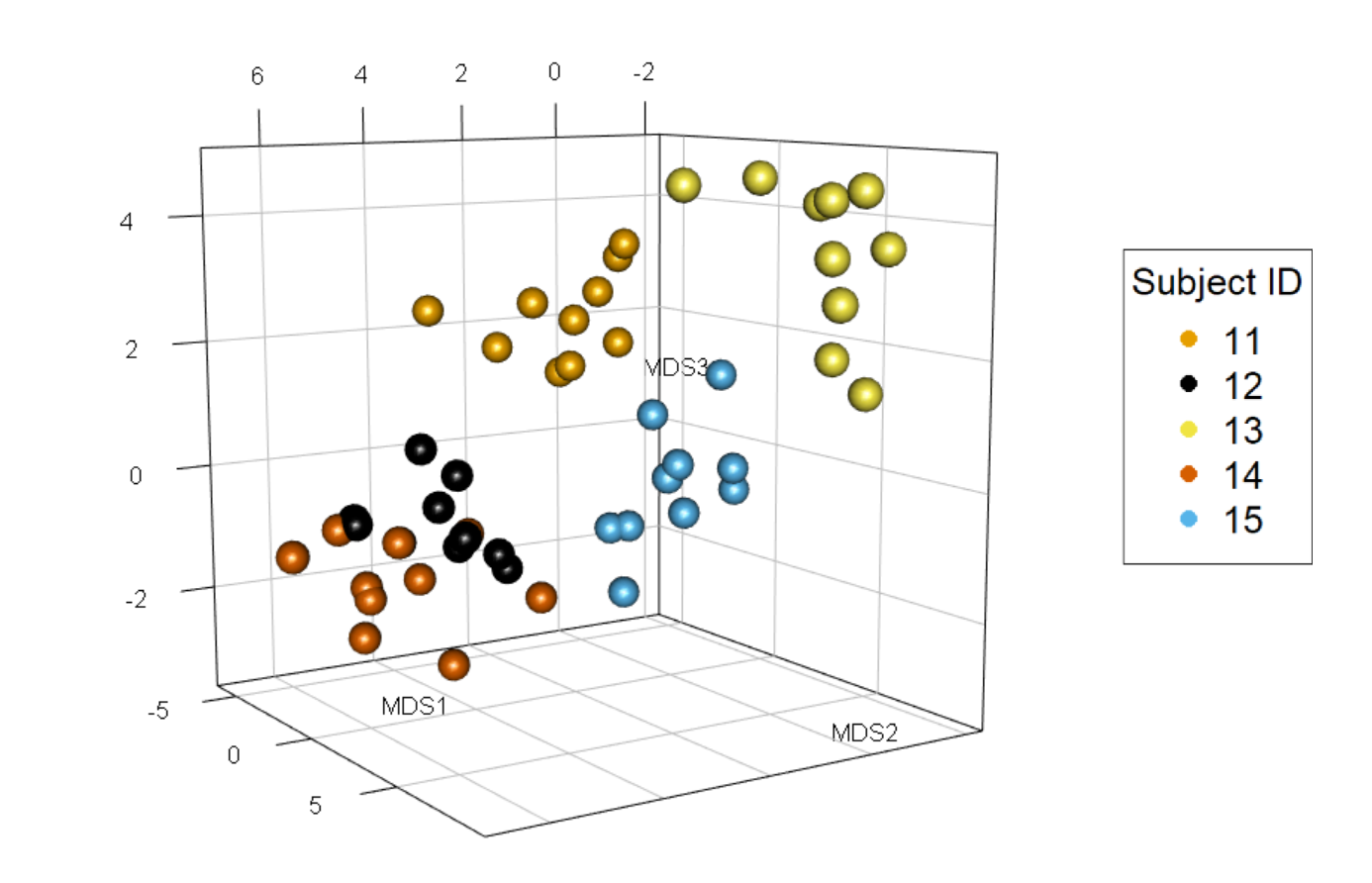}
	    \caption{Multidimensional scaling of the individual graph representations obtained by MASE for the HNU1 data. The plot shows the positions discovered by MDS for the graphs corresponding to 5 subjects of the data, showing that graphs corresponding to the same subject are more similar between each other.}
	    \label{fig:hnu1-mds}
	\end{figure}
	
	
	The ability of our method to distinguish between graphs of different subjects is further evaluated via classification and hypothesis testing.  In the classification analysis, we performed a comparison to other methods that are based on low-rank embeddings following a similar procedure to Section \ref{sec:sim-classif}, namely, the joint embedding of graphs of \cite{Wang2017} (JE), the multiple random dot product graphs of \cite{Nielsen2018} (MRDPG), and the omnibus embedding of \cite{Levin2017} (OMNI). Additionally, we compare with the common and individual structure explained (CISE) algorithm and its variant 2 (CISEv2), introduced in \cite{wang2019common}, which fit a graph embedding based on a logistic regression for the edges. 
	In a real data setting, low-rank models are only an approximation to the true generation mechanism, and a more parsimonious model that requires fewer parameters to approximate the data accurately is preferable. Hence, for model selection, we measure the accuracy as a function of the number of embedding dimensions, which controls the description length, defined as the total number of parameters used by the model. Figure~\ref{fig:hnu1-classification} shows the 10-fold cross-validated error of a 1-nearest neighbor classifier of the individual graph parameters obtained by the methods. Note that both MASE and OMNI are able to obtain an almost perfect classification, with only one graph being misclassified, in both cases corresponding to subject 20. This graph can be observed to be different to all the other graphs in Figure~\ref{fig:hnu1-D}. Although OMNI  can achieve good accuracy with only one embedding dimensions, the description length is much larger because OMNI requires $dn$ parameters for each graph, while MASE only uses $d\times d$ symmetric matrices. MASE achieves an optimal classification error with only $d=9$ dimensions; this classification analysis also offers a supervised way to choose the number of embedding dimensions for MASE suggesting that $9$ might be enough, but we keep the more conservative choice of $d=15$. Neither JE or MRDPG achieve an optimal classification error even with $d=20$ dimensions, and they are outperformed by MASE even when fixing the description length in all methods. CISE and CISEv2 are able to classify the graphs accurately, but require a larger number of parameters, and are computationally more demanding since are based on a non-convex optimization problem.
	This result suggests that MASE is not only accurate, but more parsimonious than other similar low-rank embeddings. 

	\begin{figure}
	    \centering
	    \includegraphics[width=\textwidth]{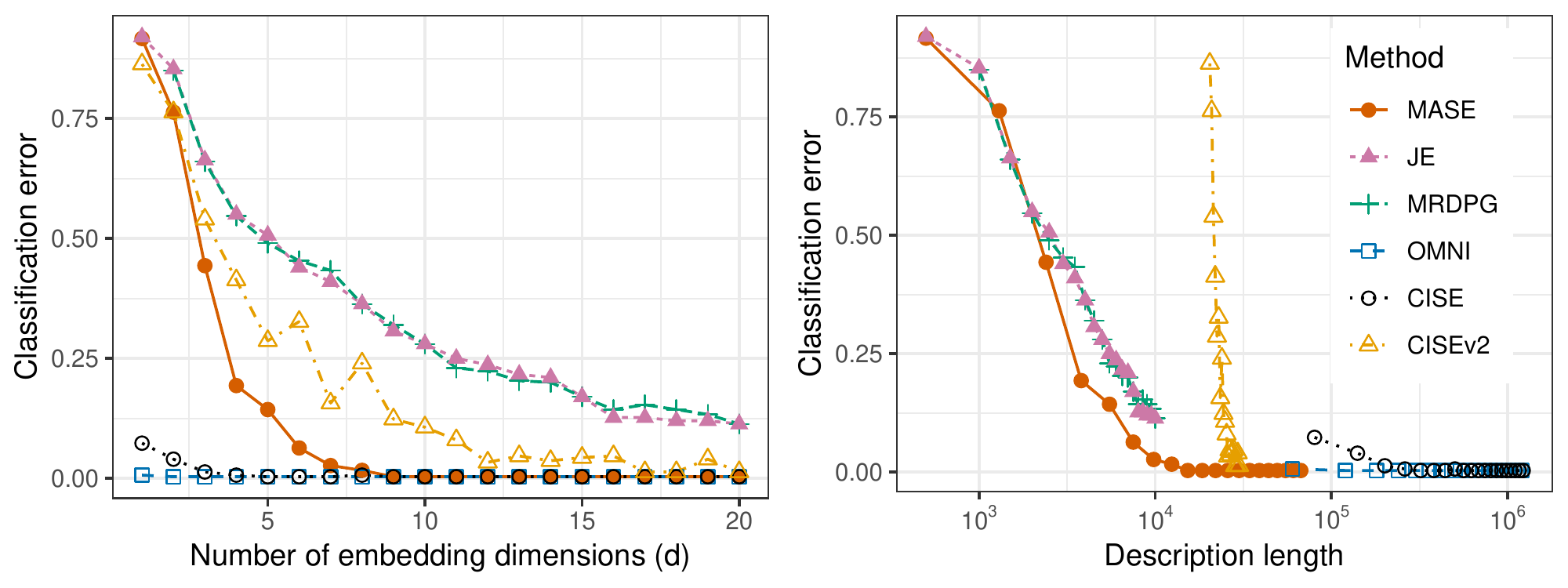}
	    \caption{Cross-validation error in classifying the subject labels of the HNU1 data. For each embedding estimated with a different method, a 1-NN classifier was trained using the distance between the individual graph embeddings, varying the number of embedding dimensions (left plot) which control the description length of the different embeddings (right plot). MASE, OMNI and CISE are the only methods that achieve perfect accuracy, but MASE requires much fewer parameters.}
	    \label{fig:hnu1-classification}
	\end{figure}

    The matrix of distance between subject parameters obtained by MASE can be used in combination with other distance-based methods to find differences and similarities between subjects, as the previous analyses show. However, a more principled approach to the identification of graph similarity is to measure, for instance, the extent to which differences between graphs can be attributed to random noise. In this spirit, we use the COSIE model to test the hypothesis that the expected edge connectivity of each pair of adjacency matrices is the same. That is, for each pair of graphs $\bA^{(i)}$ and $\bA^{(j)}$ we test the hypothesis $H_0: \bR^{(i)} = \bR^{(j)}$. We follow the same procedure described in Section~\ref{sec:sim-testing} 
    and used the parametric bootstrap method of \cite{Tang2014} and the asymptotic distribution of the score matrices to estimate the null distribution of the test statistic. Figure~\ref{fig:hnu1-pvalues} shows the $p$-values of the test performed for every pair of graphs. The diagonal blocks of the matrix correspond to graphs representing the same subject, and in most cases these p-values are large, especially when using the parametric bootstrap to estimate the null distribution. This suggests that in most cases we cannot reject the hypothesis that these graphs have the same distribution, which is reasonable given that these paired graphs represent brain scans of the same individual.
    Off-diagonal elements of the matrices in Figure~\ref{fig:hnu1-pvalues} represent the result of the equality test for a pair of graphs corresponding to different subjects, and exhibit small $p$-values in general. For some pairs of subjects, such as 12 and 14, the $p$-values are also high, suggesting some possible similarities in the brain connectivity of these two subjects, which is also observed on Figure~\ref{fig:hnu1-mds}. 
    The left panel in Figure~\ref{fig:hnu1-power} shows the percentage of rejections of the test for the $p$-values in the diagonal of the matrix, and the black dotted line represents the expected number of rejections under the null hypothesis. Note that both OMNI and MASE have slightly more rejections than what would be expected by chance under their corresponding models, especially for the p-values calculated using the asymptotic distribution (MASE-A). The right panel in Figure~\ref{fig:hnu1-power} shows the percentage of rejected tests for pairs of graphs corresponding to different subjects. Both MASE and OMNI reject a large portion of those tests.

	\begin{figure}
	    \centering
	    \begin{subfigure}{.48\textwidth}
	        \includegraphics[width=\textwidth]{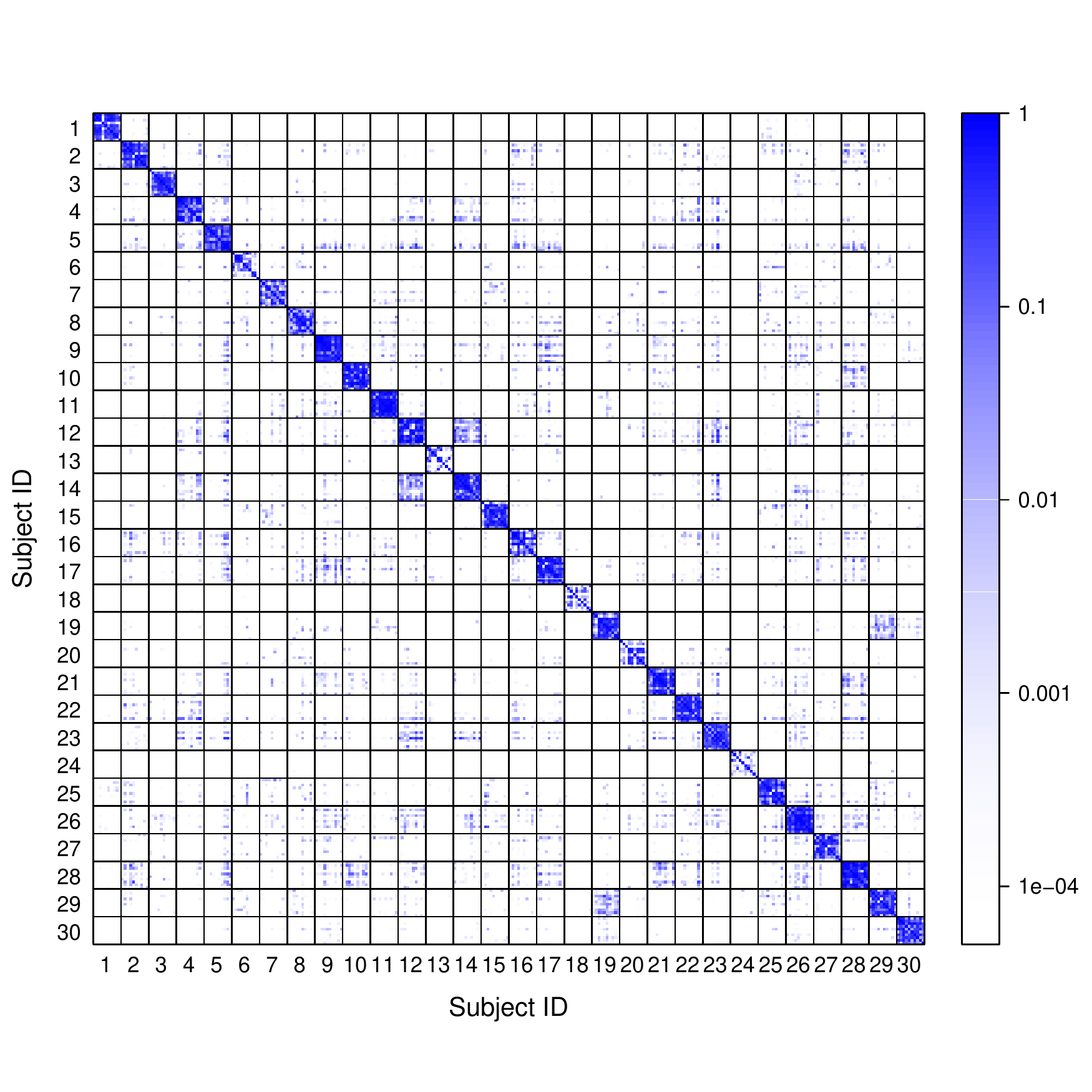}
	        \caption{Parametric bootstrap}
	    \end{subfigure}
	    \begin{subfigure}{.48\textwidth}
	        \includegraphics[width=\textwidth]{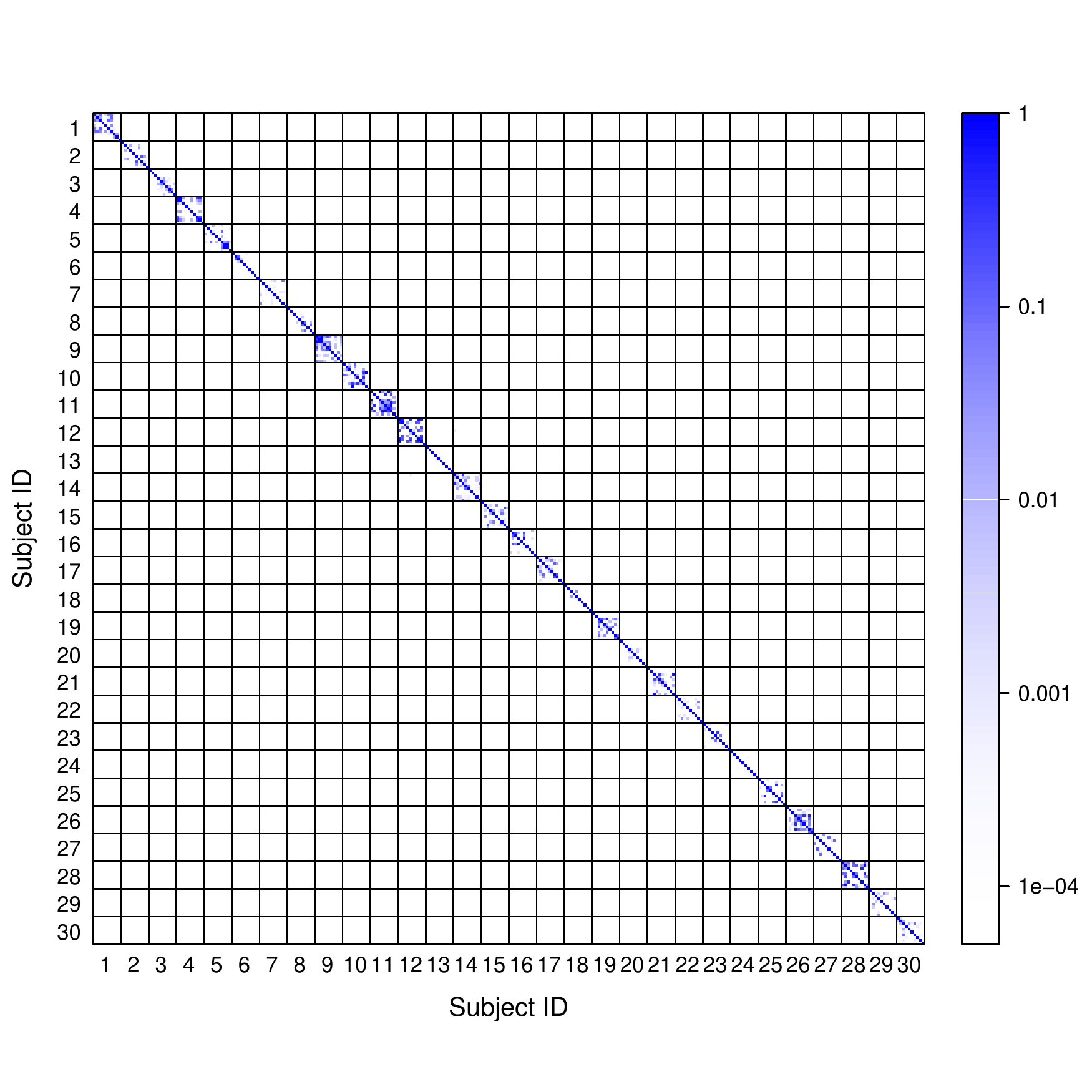}
	        \caption{Asymptotic distribution}
	    \end{subfigure}
	    \caption{Matrices of p-values for the hypothesis test that a pair of graphs has the same score matrix in the COSIE model. Each entry corresponds to the aforementioned test for a pair of the graphs in the HNU1 data. The test generally assigns small $p$-values to pairs of graphs corresponding to different subjects, and large $p$-values to same-subject pairs. }
	    \label{fig:hnu1-pvalues}
	\end{figure}

	\begin{figure}
	    \centering
	    \begin{subfigure}{0.48\textwidth}
            \includegraphics[width=\textwidth]{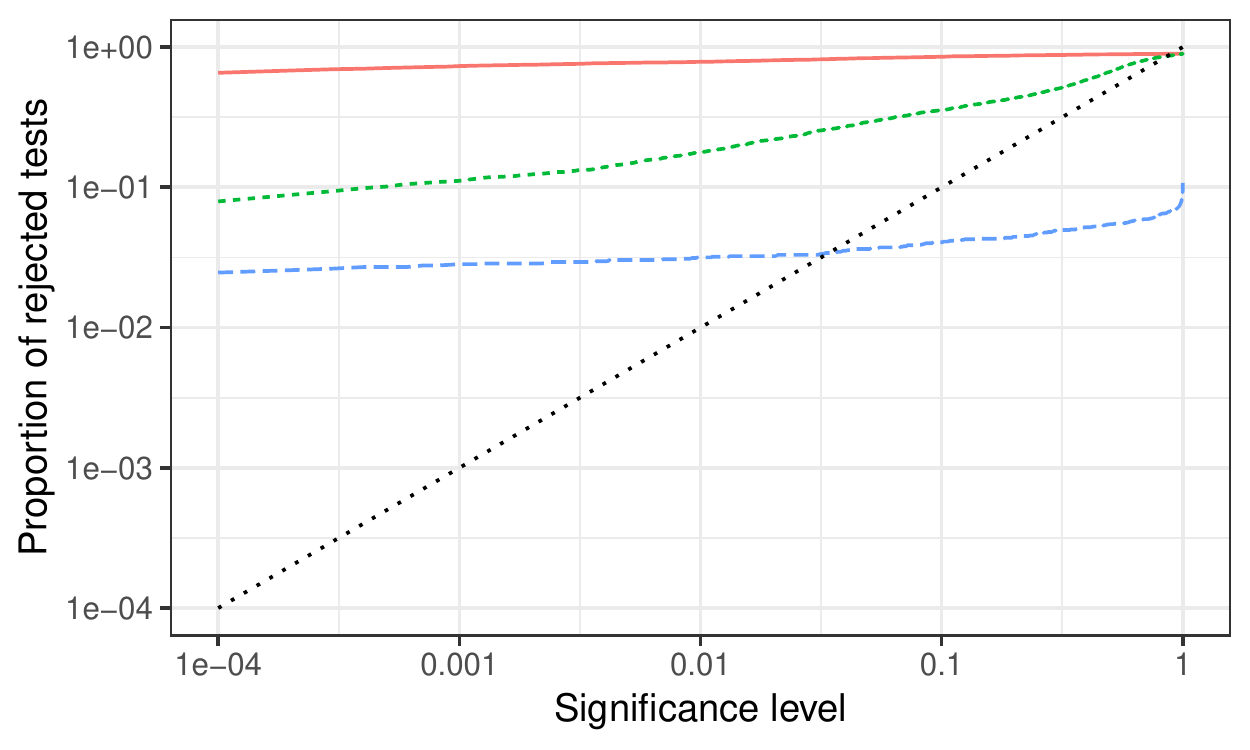}
            \caption{Same subject graph pairs}
	    \end{subfigure}
	    \begin{subfigure}{0.48\textwidth}
            \includegraphics[width=\textwidth]{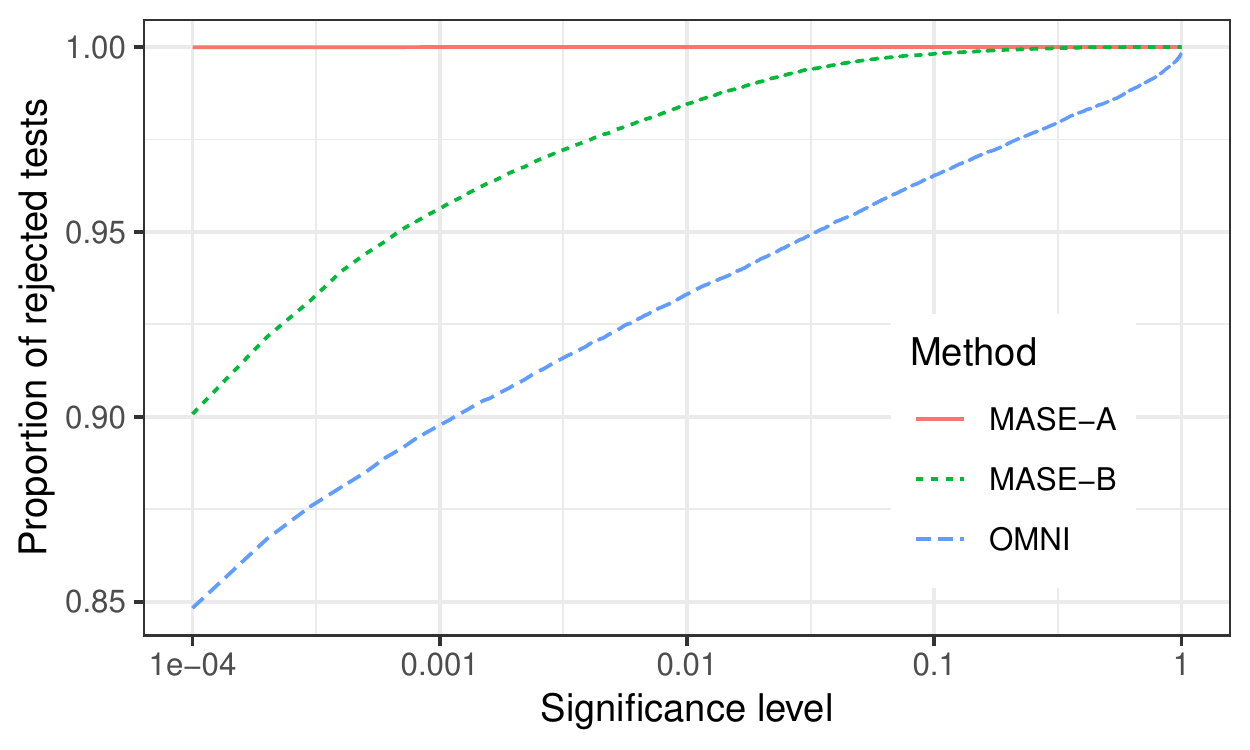}
            \caption{Different subject graph pairs} 
	    \end{subfigure}
	    \caption{Proportion of rejected test as a function of the significance level for pairwise equal distribution tests between pairs of graphs representing scans of the same subject (left) and different subjects (right). The black dotted line on the left figure corresponds to the identity, indicating the expected proportion of rejections assuming that the distribution of same-subject graphs is the same. The test constructed with MASE estimates either using the asymptotic distribution (MASE-A) or parametric bootstrap (MASE-B) to approximate the null is rejected more often than the OMNI test.}
	    \label{fig:hnu1-power}
	\end{figure}

	\section{Discussion\label{sec:concl}}	
The COSIE model presents a flexible, adaptable model for multiple graphs, one that encompasses a rich collection of independent-edge random graphs and yet retains identifiability and model parsimony. 
As described in the text, several single-graph models extend to the multiple-graph setting \citep{Airoldi2007,Zhang2014,Lyzinski2017}
within the COSIE framework. The presence of a common invariant subspace and the allowance for different score matrices guarantee that COSIE can be used to approximate real-world graph heterogeneity. The MASE procedure builds on the long history of spectral graph inference; it is intuitive, scalable, and consistent for the recovery of the common subspace in COSIE. The accurate inference of the common subspace, in turn, enables  to further determine the score parameters, which can then be deployed to discern distinctions between graphs and even subgraphs. Moreover, the MASE procedure enables graph eigenvalue estimation, community detection, and testing whether graphs arise from a common distribution.

Much work remains and includes the following: allowing the score matrices to grow in rank or capture important population-level properties such as variation in time or multilayer structure; to consider estimation and limit results with more relaxed sparsity and delocalization assumptions; and to investigate other spectral embeddings, such as the normalized or regularized Laplacian, which has been to shown to uncover different network structure \citep{Priebe2018} and provide consistent estimators in sparse networks \citep{Le2017}. Indeed, a rich array of open problems can begin to be addressed within the scope of the COSIE model.
Open \textsf{R} source code for MASE is available at \url{https://github.com/jesusdaniel/mase}, and in the Python \texttt{GraSPy} package at \url{https://neurodata.io/graspy}~\citep{Chung2019-tn}.

\section*{Acknowledgements}	
This research has been supported by the Lifelong Learning Machines (L2M) program of the Defence Advanced Research Projects Agency (DARPA) via contract number HR0011-18-2-0025. This work is also supported in part by the D3M program of DARPA,  NSF DMS award number 1902755 and funding from Microsoft Research. We would like to  thank Keith Levin and Elizaveta Levina for helpful discussions.

	\bibliographystyle{apalike}
	\bibliography{biblio2}
	\section*{Appendix: Proofs}
		\subsection{Stochastic blockmodels are COSIE}
		\begin{proof}[Proof of Proposition~\ref{prop:sbm-is-cosie}] Define $\bV=\bZ(\bZ^\top\bZ)^{-1/2}$. Because $\bZ^\top\bZ$ is diagonal and full rank, it is easy to observe that the columns of $\bV$ are orthogonal. By writing $\bR^{(i)}=(\bZ^\top\bZ)^{1/2}\bB^{(i)}(\bZ^\top\bZ)^{1/2}$ we obtain the result.
\end{proof}

	\subsection{Identifiability}
	\begin{proof}[Proof of Proposition~\ref{prop:identifiability}]
	    Suppose that there exists a matrix of orthonormal columns $\bU\in\real^{n\times d}$ and symmetric matrices $\bS^{(1)},\ldots,\bS^{(m)}\in\real^{d\times d}$ such that
	\begin{equation}
	\label{eq:VRV-USU}
	    \bV\bR^{(i)}\bV^\top=\bU\bS^{(i)}\bU^\top
	\end{equation}
	for all $i\in[m]$.
	\begin{enumerate}[a)] 
	    \item Observe that
	        \begin{align*}
	            \|\bR^{(i)} - \bR^{(j)}\| & = \|\bV(\bR^{(i)} - \bR^{(j)})\bV^\top\|\\
	            & =  \|\bU(\bS^{(i)} - \bS^{(j)})\bU^\top\|\\
	            & =  \|\bS^{(i)} - \bS^{(j)}\|.
	        \end{align*}
	        The claim follows from the fact that the same equalities hold with Frobenius norm in place of the spectral norm.
	        \item  Squaring the matrices in Equation~\eqref{eq:VRV-USU} yields
	        \begin{align*}
	            \bU\left(\sum_{i=1}^m\left(\bS^{(i)}\right)^2\right)\bU^\top & = \bV\left(\sum_{i=1}^m\left(\bR^{(i)}\right)^2\right)\bV^\top\\
	            & = \bV (\widetilde{\bR}\widetilde{\bR}^\top)\bV^\top.
	        \end{align*}
	        Note that the full-rank condition of $\widetilde{\bR}$ implies that $\widetilde{\bR}\widetilde{\bR}^\top$ is also full-rank, so the inverse of $\widetilde{\bR}\widetilde{\bR}^\top$ exists. Set \[\bW=\left(\sum_{i=1}^m\left(\bS^{(i)}\right)^2\right)\bU^\top\bV(\widetilde{\bR}\widetilde{\bR}^\top)^{-1},\]
	        so that $\bU\bW=\bV$. To check that $\bW$ is orthogonal, note that $\bU^\top\bU\bW=\bW=\bU^\top\bV$ and $\bV^\top\bU\bW=\bI \in \mathbb{R}^{d \times d}$, which imply $\bW^\top\bW=\bI$. Since $\bW$ is square ($d \times d)$, this implies that $\bW$ is of full rank. The uniqueness of the matrix inverse guarantees that $\bW^{-1}=\bW^\top$, and, in turn, that $\bW\bW^\top=\bI$; this implies that $\bW$ is orthonormal.
	        \item Suppose that $\bU=\bV$. Then, multiplying Equation~\eqref{eq:VRV-USU} by $\bV^\top$ and $\bV$ on the left and right establishes that $\bS^{(i)}=\bR^{(i)}$.
	        
	    \end{enumerate}
	\end{proof}
	\subsection{Estimation of the common invariant subspace}
	
	

	Before presenting the proof of Theorem~\ref{thm:V-Vhat}, we first establish the following lemma.
	\begin{lemma} \label{lemma:subgaussian-A-P}
	Let $(\bA_n)_{n=1}^{\infty}$ be a sequence of the adjacency matrices of  independent-edge Bernoulli random graphs such that $\bA_n\in\{0,1\}^{n\times n}$ and $\e[\bA_n|\bP_n]=\bP_n$ for some $\bP_n\in\real^{n\times n}$. Suppose that $\delta(\bP_n)=\omega(\log n)$. Then, there exists some $n_0\in\mathbb{N}$ such that for all $n\geq n_0,$ the random variable $\|\bA_n-\bP_n\|$ is sub-gaussian and
	\[\e\left[\|\bA_n-\bP_n\|^p\right] \leq 2^{2p-1}p^{p/2+1}\left(\delta(\bP_n)\right)^{p/2}.\]
	\end{lemma}
	
	\begin{proof}
	By Theorem 20 in \cite{Athreya2017},  for any $c>0$ there exists $C>0$  such that if  $\delta(\bP_n)>C\log n$ then for any $n^{-c}<\eta<1/2$ we have
		\begin{equation*}			\p\left(\|\bA_n - \bP_n\| \leq 4 \sqrt{\delta(\bP_n)\log (n/\eta)}\right) \geq 1-\eta. 
		\end{equation*}
		Taking $c=1.1$, $\eta=1/n$ and $n_0$ such that $\delta(\bP_n)>C\log n$ holds for the appropriate $C$ and all $n\geq n_0$, we have that
		\begin{equation*}
			\p\left(\|\bA_n - \bP_n\| > t\right) \leq \exp\left(-\frac{t^2}{32\delta(\bP_n)}\right). 
		\end{equation*}
		Therefore, by Lemma 5.5 in 
		\cite{Vershynin2010}, $\|\bA_n-\bP_n\|$ is sub-Gaussian for all $n\geq n_0$, and 
		\begin{align*}
		    \e\left[\|\bA_n-\bP_n\|^p\right] & = \int_0^{\infty}\p(\|\bA_n-\bP_n\|>t)p t^{p-1}dt\\
		    & \leq \left(\frac{p}{2}\right)^{1+p/2}\left(32\delta(\bP_n)\right)^{p/2}.
		\end{align*}
	\end{proof}
	
	\begin{proof}[Proof of Theorem \ref{thm:V-Vhat}]
	 We introduce some notation. Define the following matrices as 
	\begin{align*}
	\widehat{\bPi}^{(i)} & =\widehat{\bV}^{(i)}(\widehat{\bV}^{(i)})^\top, & \widehat{\bPi} = \frac{1}{m}\sum_{i=1}^m\widehat{\bPi}^{(i)},\\
	\widetilde{\bPi}^{(i)} & =\e\left[\widehat{\bV}^{(i)}(\widehat{\bV}^{(i)})^\top\right] &  \widetilde{\bPi} = \frac{1}{m}\sum_{i=1}^m\widetilde{\bPi}^{(i)},\\
	\bPi & = \bV\bV^\top.
	\end{align*}
	Note that $\widehat{\bU}\widehat{\bU}^\top=\widehat{\bPi}$, and therefore $\widehat{\bV}$ corresponds to the $d$ leading eigenvectors of $\widehat{\bPi}$. Recall that $\widetilde{\bV}$ is the matrix containing the $d$ leading eigenvectors of $\widetilde{\bPi}=\frac{1}{m}\sum_{i=1}^m\e\big[\widehat{\bV}^{(i)}(\widehat{\bV}^{(i)})^\top\big]$.
	
	To prove Theorem~\ref{thm:V-Vhat}, we follow the main arguments of \cite{Fan2017}. We analyze the terms of the right hand side of Equation~\eqref{eq:theory-bias-variance} separately.
	For the first term, corresponding to  $\|\widehat{\bV}\widehat{\bV}^\top- \widetilde{\bV}\widetilde{\bV}^\top\|_F$, note that by the Davis-Kahan theorem (see \cite{Yu2015}),
		\begin{align}
		\mathbb{E}\left[\|\widehat{\bV}\widehat{\bV}^\top- \widetilde{\bV}\widetilde{\bV}^\top\|_F\right] & \leq \frac{2^{3/2} \mathbb{E}\|\widehat{\bPi} - \widetilde{\bPi}\|_F}{\lambda_d(\widetilde{\bPi}) - \lambda_{d+1}(\widetilde{\bPi})}\nonumber\\
		& =  \frac{2^{3/2}
		\mathbb{E}\left\|\frac{1}{m}\sum_{i=1}^m(\widehat{\bPi}^{(i)} - \widetilde{\bPi}^{(i)} )\right\|_F}{\lambda_d(\widetilde{\bPi}) - \lambda_{d+1}(\widetilde{\bPi})}.\label{eq:varianceproof}
		\end{align}
		Using again the the Davis-Kahan Theorem, the spectral norm of each term inside the norm of the numerator can be bounded as 
		\begin{align}
		\|\widehat{\bPi}^{(i)} - \bPi\|_F & \leq \frac{2^{3/2}\sqrt{d}\|\bA^{(i)}-\bP^{(i)}\|}{|\lambda_{\min}(\bR^{(i)})|} \label{eq:daviskahan-Vi-V}
		\end{align}

		Now, to control the error of each term in the numerator of Equation \eqref{eq:varianceproof}, we use the triangle and Jensen inequalities, and the previous bound in Equation~\eqref{eq:daviskahan-Vi-V} as follows
		\begin{align*}
		\mathbb{E}\left[\|\widehat{\bPi}^{(i)}- \widetilde{\bPi}^{(i)}\|_F^p\right]^{1/p}  & \leq 
		\mathbb{E}\left[\|\widehat{\bPi}^{(i)}- \bPi \|_F^p\right]^{1/p} + \mathbb{E}\left[\|\mathbb{E}[\bPi -\widehat{\bPi}^{(i)}] \|_F^p\right]^{1/p}\\
		& \leq 2\mathbb{E}\left[\|\widehat{\bPi}^{(i)}- \bPi \|_F^p\right]^{1/p}\\
		& \leq \frac{2^{5/2}\sqrt{d}\mathbb{E}\left[\|\bA^{(i)}-\bP^{(i)}\|^p\right]^{1/p}}{|\lambda_{\min}(\bR^{(i)})|}\\
		&  \leq \frac{2^{9/2}p\left(d\  \delta(\bP^{(i)})\right)^{1/2}}{|\lambda_{\min}(\bR^{(i)})|} 
		\end{align*}
		where the last inequality comes from Lemma~\ref{lemma:subgaussian-A-P}, assuming that $n$ is sufficiently large.
		Using Lemma 4 of \cite{Fan2017},
		\begin{equation*}
		 \mathbb{E}\left[\left\|\frac{1}{m}\sum_{i=1}^m(\widehat{\bPi}^{(i)} - \widetilde{\bPi}^{(i)} )\right\|_F \right]
		 \lesssim \frac{1}{\sqrt{m}}\sqrt{\frac{d}{m}\sum_{i=1}^m \frac{\delta(\bP^{(i)})}{\lambda_{\min}^2(\bR^{(i)})}}.
		\end{equation*}
		Combining the previous equation and Equation~\eqref{eq:varianceproof}, 
		\begin{equation}
			\mathbb{E}\left[\|\widehat{\bV}\widehat{\bV}^\top - \widetilde{\bV}\widetilde{\bV}^\top\|_F\right]
			 \lesssim  \frac{1}{\sqrt{m}(\lambda_d(\widetilde{\bPi}) - \lambda_{d+1}(\widetilde{\bPi}))}\sqrt{\frac{d}{m}\sum_{i=1}^m \frac{\delta(\bP^{(i)})}{\lambda^2_{\min}(\bR^{(i)})}}.\label{eq:proof-variance-bound}
		\end{equation}

		To control the error of the second term in Equation~\eqref{eq:theory-bias-variance}, we use the triangle inequality and Lemma 2 of \cite{Fan2017}.  
		For each $i\in[m]$, suppose that $U_1^{(i)},\ldots,U_n^{(i)}\in\real^{n}$ is an orthonormal basis of $\mathbb{R}^n$ such that $\operatorname{span}\{U_1^{(i)}, \ldots, U^{(i)}_d\}= \operatorname{span}(\bV)$ and $\bP^{(i)}U^{(i)}_j=\lambda_j(\bP^{(i)})U_j^{(i)}$ for each $i\in[d]$. Thus, the vectors $U_1^{(i)}, \ldots, U^{(i)}_d$ correspond to the $d$ leading eigenvectors of $\bP^{(i)}$ when $\lambda_j(\bP^{(i)})\neq 0$, and they span the same invariant subspace as $\bV$. Define $\bU^{(i)}\in\real^{n\times d}$ as
		\[\bU^{(i)} = \left( U_1^{(i)}\ \cdots \ U_d^{(i)}\right).\]
		Note that $\bU^{(i)}\bU^{(i)^\top} = \bV\bV^\top$. Also, define $\bLambda^{(i)}\in\mathbb{R}^{d\times d}$ be a diagonal matrix containing the $d$ leading eigenvalues in magnitude of $\bP^{(i)}$, such that $\bLambda^{(i)}_{jj} =\lambda_j(\bP^{(i)})$. 
		Set $\bG^{(i)}=\sum_{k=d+1}^nU_k^{(i)}(U_k^{(i)})^\top$. Define $\bH^{(i)}$ as
		\[\bH^{(i)} = \bG^{(i)}(\bA^{(i)}-\bP^{(i)})\bU^{(i)}(\bLambda^{(i)})^{-1}\bU^{(i)^\top}.\]
		Note that $\e[\bH^{(i)}]=0$, and
		\begin{align}
		    \| \bH^{(i)} + \bH^{(i)^\top}\|_F & \leq 2\|\bG^{(i)}(\bA^{(i)}-\bP^{(i)})\bU^{(i)}(\bLambda^{(i)})^{-1}\bU^{(i)^\top}\|_F\nonumber\\
		    & \leq 2\|(\bLambda^{(i)})^{-1}\|\|\bG^{(i)}\|\|\bA^{(i)}-\bP^{(i)}\|_F\nonumber\\
		    & \leq 2\sqrt{d} \frac{\|\bA^{(i)}-\bP^{(i)}\|}{|\lambda_{\min}(\bR^{(i)})|}.\label{eq:boundHH}
		\end{align}
		By Lemma 2 of \cite{Fan2017}, if $\|\bA^{(i)} - \bP^{(i)}\|\leq \lambda_{\min}(\bR^{(i)})/10$ then
		\begin{equation*}
		 \|\widehat{\bPi}^{(i)} - \bPi - \bH^{(i)} - \bH^{(i)^\top}\|_F \leq 24 \sqrt{d} \frac{\|\bA^{(i)}-\bP^{(i)}\|^2}{\lambda_{\min}(\bR^{(i)})^2}.\label{eq:lemma2-fan}.   
		\end{equation*}
		On the other hand, if $\|\bA^{(i)} - \bP^{(i)}\|> \lambda_{\min}(\bR^{(i)})/10$, by using the Davis-Kahan theorem and Equation~\eqref{eq:boundHH},
		\begin{align*}
		    \|\widehat{\bPi}^{(i)} - \bPi - \bH^{(i)} - \bH^{(i)^\top}\|_F  & \leq \|\widehat{\bPi}^{(i)} - \bPi\|_F + \| \bH^{(i)} + \bH^{(i)^\top}\|_F\\
		    & \leq 4\sqrt{d} \frac{\|\bA^{(i)}-\bP^{(i)}\|}{|\lambda_{\min}(\bR^{(i)})|}.
		\end{align*}
		 Define $\mathcal{D}$ as the event in which $\|\bA^{(i)}- \bP^{(i)}\| > \lambda_{\min}(\bR^{(i)})/10$. Observe that
		\begin{align*}
		    \|\widetilde\bPi^{(i)} - \bPi\|_F  & = \|\mathbb{E}[\widehat\bPi^{(i)} - \bPi -\bH^{(i)} -\bH^{(i)^\top}]\|_F \\
		    & \leq \mathbb{E}[\|\widehat\bPi^{(i)} - \bPi -\bH^{(i)} -\bH^{(i)^\top}\|_F(\mathbbm{1}_\mathcal{D}+ \mathbbm{1}_\mathcal{D^C})] \\
		    & \leq 4\sqrt{d} \frac{\e[\mathbbm{1}_\mathcal{D}\|\bA^{(i)}-\bP^{(i)}\|]}{|\lambda_{\min}(\bR^{(i)})|} +
		    24 \sqrt{d} \frac{\e[\|\bA^{(i)}-\bP^{(i)}\|^2]}{\lambda_{\min}(\bR^{(i)})^2}\\
		    & \leq 40\sqrt{d} \frac{\e[\|\bA^{(i)}-\bP^{(i)}\|^2]}{\lambda_{\min}(\bR^{(i)})^2} +
		    24 \sqrt{d} \frac{\e[\|\bA^{(i)}-\bP^{(i)}\|^2]}{\lambda_{\min}(\bR^{(i)})^2}\\
		    & \leq 
		     64\sqrt{d}\frac{\mathbb{E}[\|\bA^{(i)} - \bP^{(i)}\|^2]}{\lambda_{\min}(\bR^{(i)})^2} \\
		     & \lesssim \sqrt{d}\frac{\delta(\bP^{(i)})}{\lambda_{\min}(\bR^{(i)})^2}
		\end{align*}
    By the triangle inequality,
    \begin{align} \label{eq:proof-bias-bound}
		\|\widetilde{\bPi} - \bPi\|_F & \lesssim \frac{\sqrt{d}}{m}\sum_{i=1}^m \frac{\delta(\bP^{(i)})}{\lambda^2_{\min}(\bR^{(i)})}.
	\end{align}
	Combining Equations~\eqref{eq:proof-variance-bound} and \eqref{eq:proof-bias-bound}, we obtain that
	\[\e[\|\widehat{\bV}\widehat{\bV}^\top-\bV\bV^\top\|_F] \lesssim  
	\frac{1}{\sqrt{m}(\lambda_d(\widetilde{\bPi}) - \lambda_{d+1}(\widetilde{\bPi}))}\sqrt{\frac{d}{m}}\varepsilon
	 + 
	\sqrt{d}\varepsilon^2.\]
     Finally, the denominator of the first term can be bounded using Weyl's inequality, by observing that
		\begin{align*}
		    \lambda_d(\widetilde{\bPi}) & \geq 1 - \|\widetilde{\bPi} - \bPi\|,\\
		    \lambda_{d+1}(\widetilde{\bPi}) &\leq \|\widetilde{\bPi} - \bPi\|\\
		    \lambda_d(\widetilde{\bPi}) - \lambda_{d+1}(\widetilde{\bPi}) & \geq 1 - 2\|\widetilde{\bPi} - \bPi\|_F.
		\end{align*}
		Finally, Equation~\eqref{eq:proof-bias-bound} and the fact that $\varepsilon=o(1)$ ensures that this denominator is bounded away from zero, which completes the proof.
	\end{proof}

	\subsection{Community detection results}
		
	\begin{proof}[Proof of Corollary~\ref{cor:sbm-V-Z}] Define the parameters of the COSIE model
	$$ \bV = \bZ\bXi^{-1/2},$$
	$$\bR^{(i)}=\bXi^{1/2} \bB^{(i)}\bXi^{1/2},$$
	so that $\bV\bR^{(i)}\bV^\top=\bZ\bB^{(i)}\bV^\top$.
	To obtain a bound on $\varepsilon$ defined according to Theorem~\ref{thrm:COSIE Expectation Bound}, observe that the largest expected degree can be bounded above as
	\begin{align*}
	    \delta(\bP^{(i)}) & = \max_{j\in[K]} \sum_{k=1}^K\bB^{(i)}_{jk}n_k\\
	    & \leq \left(\sum_{j=1}^Kn_j^2\right)^{1/2}\|\bB^{(i)}\|_1\\
	    & \leq \sqrt{K}\left(\sum_{j=1}^Kn_j^2\right)^{1/2}\lambda_1(\bB^{(i)}).\\
	\end{align*}
	Also, the smallest eigenvalue of the corresponding score matrix $\bR^{(i)}$ is bounded below as
	\begin{align*}
	    |\lambda_{\min}(\bR^{(i)})| & = |\lambda_{\min}(\bXi^{1/2} \bB^{(i)}\bXi^{1/2})|\\
	    & \geq \lambda_{\min}(\bXi)|\lambda_{\min}(\bB^{(i)})|  = n_{\min}|\lambda_{\min}(\bB^{(i)})|.
	\end{align*}
	Therefore, using the above inequalities and Assumption~\ref{assumption:multilayer-SBM-parameters}, 
	\begin{align*}
	    \varepsilon^2=\frac{1}{m}\sum_{i=1}^m\frac{\delta(\bP^{(i)})}{\lambda^2_{\min}(\bR^{(i)})} & \leq \frac{1}{m} \sum_{i=1}^m \frac{\sqrt{K}\left(\sum_{j=1}^Kn_j^2\right)^{1/2} \lambda_1(\bB^{(i)})}{n_{\min}^2 \lambda^2_{\min}(\bB^{(i)})} \\
	    &\leq \frac{\gamma}{n_{\min}^\kappa}.
	\end{align*}
	Based on the above calculations, Theorem \ref{thm:V-Vhat} implies the result.
	\end{proof}
	
	\begin{proof}[Proof of Theorem \ref{thm:community-detection}] 
	Consider the permutation matrix $\widehat{\bQ}$ and the orthogonal matrix $\widehat{\bW}$  such that
	\[\widehat{\bQ} =\argmin_{\bQ\in\mathcal{P}_K}\|\widehat{\bZ}\bQ-\bZ\|_F.\]
	\[\widehat{\bW} =\argmin_{\bW\in\mathcal{O}_K}\|\widehat{\bZ}-\bZ\bW\|_F.\]
	Using Lemma 3.2 of \cite{Rohe2011}, observe that if $\|\widehat{Z}_u\widehat\bC\widehat{\bQ}-V_u\widehat\bW\|<1/\sqrt{2n_{\max}}$ then $\widehat{Z}_u\bQ=Z_u$. Hence,
	    \begin{align*}
	    \|\widehat{\bZ}\widehat{\bQ} -\bZ\|_F & \leq \sum_{u=1}^n\mathbbm{1}\{\widehat{Z}_u\widehat{\bQ}\neq Z_u\}\\
	     & \leq \sqrt{2n_{\max}}
	        \|\widehat{\bZ}\widehat{\bC}\widehat{\bQ} - \bV\widehat\bW\|_F\\
	        & \leq \sqrt{2n_{\max}}\left(\|\widehat{\bZ}\widehat{\bC}\widehat{\bQ} - \widehat{\bV}\|_F + \|\widehat{\bV} - \bV\widehat\bW\|_F\right)\\
	        & \leq 2^{3/2}\sqrt{n_{\max}}\|\widehat{\bV} - \bV\widehat{\bW}\|_F.
	    \end{align*}
	    Corollary \ref{cor:sbm-V-Z} immediately implies the result.
	\end{proof}

\subsection{Asymptotic normality results}

\begin{proof}[Proof of Proposition~\ref{proposition:sbm-delocalization}]
	    \begin{enumerate}[a)]
	        \item Define $\bXi=\operatorname{diag}(n_1,\ldots,n_K)$ as before. Note that $\bZ\bXi^{-1/2}$ is the common subspace of a multilayer SBM with community membership $\bZ$. Therefore, it is enough to show that there exist some $\bW\in\mathcal{O}_K$ such that Assumption~\ref{assump:Delocalization} holds for $\bZ\bXi^{-1/2}\bW$. Define an orthogonal matrix $\bW$ such that
	        $$\frac{b_1}{\sqrt{K}} \leq |\bW_{kl}| \leq \frac{b_2}{\sqrt{K}}$$
	        for all $k,l\in[K]$. Such a matrix exists for any $K\geq 1$ (see for example \cite{Jaming2015}). If the node $u$ belongs to community $k$, then $$\frac{b_1}{\sqrt{Kn_{\max}}}\leq \left|(\bZ\bW)_{uv}\right| = \frac{|\bW_{uv}|}{\sqrt{n_k}}\leq \frac{b_2}{\sqrt{Kn_{\min}}},$$
	        which completes the proof.
	        
	        \item Under the SBM, $\bP^{(i)}=\bZ\bB^{(i)}\bZ^\top$. Therefore
	        $$ \sum_{u=1}^n\sum_{v=1}^n\bP^{(i)}_{uv}(1-\bP^{(i)}_{uv}) = \sum_{s=1}^K\sum_{t=1}^Kn_sn_t\bB^{(i)}_{st}(1-\bB^{(i)}_{st})=\omega(1).$$
	    \end{enumerate}
	    \end{proof}

	    We now turn to our proof of the asymptotic distribution of $\widehat{\bR}^{(i)}$. Recall that the score matrices $\bR^{(i)}=\bR^{(i,n)}$ and the basis of the invariant subspace $\bV=\bV^{(n)}$ depend on $n$, and so does the distribution of the random matrices $\bA^{(1)}, \ldots, \bA^{(m)}$.
In this section, we write
$$\bE^{(i)} := \bA^{(i)} - \bV\bR^{(i)}\bV^\top  .$$
We emphasize that $\bE^{(i)}$ also depends on $n$. The asymptotic normality ultimately depends on sums of the entries for a sequence of matrices $\bE^{(i)}$ that increase their dimension and change the distribution as $n$ increases.

We recall an important elementary result relating the Frobenius norm and the $\sin(\Theta)$ distance (see \cite{Bhatia1997}) for any pair of matrices: if two
matrices $\widehat{\bU}, \bU \in \mathbb{\bR}^{n \times d}$ each have orthonormal columns, then
    \begin{equation}\label{eq:Frob_norm_sin_theta_1}
        \|\sin\Theta(\widehat{\bU},\bU)\|_{F}
        \le
        \operatorname{inf}_{\bW \in \mathcal{O}_{d}}\|\widehat{\bU}-\bU \bW\|_{F}
        \leq \sqrt{2}\|\sin\Theta(\widehat{\bU},\bU)\|_{F}.
    \end{equation}
    (Observe that the expectation bound stated in Theorem~\ref{thrm:COSIE Expectation Bound} immediately applies to $\sin\Theta$ distances up to an additional absolute constant factor.)
    Now for $\bW = \operatorname{\textnormal{arg inf}}_{\bW \in \mathcal{O}_d}\|\widehat{\bU}-\bU\bW\|_{F}$, then
    \begin{equation}\label{eq:Frob_norm_sin_theta_2}
        \|\bU^{\top}\widehat{\bU}-\bW\|_{F}
        \le
        \|\sin\Theta(\widehat{\bU},\bU)\|_{F}^{2}.
    \end{equation}
    To observe that the previous inequality holds, write $\widetilde{\bQ}_1\widetilde{\bD}\widetilde{\bQ}_2^\top=\bU\bU^\top $ as the singular value decomposition of $\bU\bU^\top $, with $\widetilde{\bQ}_1, \widetilde{\bQ}_2\in\real^{d\times d}$ orthogonal matrices, and $\widetilde{\bD}\in\real^{d\times d}$ a diagonal matrix with  $0\leq \widetilde{\bD}_{kk}\leq 1$. Then,
    $$ \|\sin\Theta(\widehat{\bU},\bU)\|_{F}^{2} = \sum_{k=1}^d(1-\widetilde{\bD}_{kk}^2).$$
    Now, by taking $\bW =  \widetilde{\bQ}_1 \widetilde{\bQ}_2^\top$ (which is an orthogonal matrix), we have that 
    $\bU^{\top}\widehat{\bU}- \widetilde{\bQ}_1 \widetilde{\bQ}_2^\top=  \widetilde{\bQ}_1(\widetilde{\bD} - \bI) \widetilde{\bQ}_2^\top$, and hence
    $$\|\bU^{\top}\widehat{\bU}- \widetilde{\bQ}_1 \widetilde{\bQ}_2^\top\|_F = \left(\sum_{k=1}^d(1-\widetilde{\bD}_{kk})^2\right)^{1/2} \leq \sum_{k=1}^d(1-\widetilde{\bD}_{kk}^2).$$
    Finally, notice that $\bW = \operatorname{\textnormal{arg inf}}_{\bW \in \mathcal{O}_d}\|\left[\bU, \bU_\perp\right]^\top \left(\widehat{\bU}-\bU\bW\right)\|_{F}=\operatorname{\textnormal{arg inf}}_{\bW \in \mathcal{O}_d}\|\bU^\top\widehat{\bU}-\bW\|_{F}$.


We now proceed to further quantify the estimation of $\bR^{(i)}$ via $\widehat{\bR}^{(i)} := \widehat{\bV}^{\top}\bA^{(i)}\widehat{\bV}$. In particular, we will prove the following lemma.

\begin{lemma}\label{lem:COSIE_R_decomp} Under the general assumptions of Theorem~\ref{thm:COSIE_CLT}, there exist sequences of matrices $\bB^{(i)}, \bN^{(i)},\bW \in \mathbb{R}^{d \times d}$  depending on $d$, $m$, and $n$, such that
\begin{equation}
    \bW\widehat{\bR}^{(i)}\bW^\top - \bR^{(i)}  - \bW \bB^{(i)} \bW^{\top} - \bW \bN^{(i)} \bW^{\top}=
    \bV^{\top}\bE^{(i)}\bV ,\label{eq:R-decomposition}
\end{equation}
with $\bH_{m}^{(i)} = -(\bW \bB^{(i)} \bW^{\top} + \bW \bN^{(i)} \bW^{\top})$ satisfying $\e[\|\bH_m^{(i)}\|_F] =O_P(d/\sqrt{m})$, and  $\bW = \operatorname{\textnormal{arg inf}}_{\bW \in \mathcal{O}_{d}}\|\widehat{\bV}-\bV\bW\|_{F}$ is an orthogonal matrix
\end{lemma}
\begin{proof}
Letting $\bE^{(i)}:=\bA^{(i)}-\bP^{(i)}$, we obtain the expansion
\begin{align}
    \widehat{\bR}^{(i)}
        &= \widehat{\bV}^{\top}\bA^{(i)}\widehat{\bV}\nonumber\\
        &= \widehat{\bV}^{\top}\bP^{(i)}\widehat{\bV} + \widehat{\bV}^{\top}(\bA^{(i)}-\bP^{(i)})\widehat{\bV}\nonumber\\
        &= (\widehat{\bV}^{\top}\bV) \bR^{(i)} (\bV^{\top} \widehat{\bV}) + \widehat{\bV}^{\top}\bE^{(i)}\widehat{\bV}.\label{eq:decomp-R}
\end{align}
For any orthogonal matrix $\bW \in \mathbb{R}^{d \times d}$, the first term can be further expanded as
\begin{align*}
    (\widehat{\bV}^{\top}\bV) \bR^{(i)} (\bV^{\top} \widehat{\bV})
        &= \bW^{\top} \bR^{(i)} \bW + (\widehat{\bV}^{\top}\bV - \bW^{\top})\bR^{(i)} (\bV^{\top} \widehat{\bV}) + \bW^{\top} \bR^{(i)} (\bV^{\top} \widehat{\bV} - \bW).
\end{align*}
Furthermore, let $\bW = \operatorname{\textnormal{arg inf}}_{\bW \in \mathcal{O}_{d}}\|\widehat{\bV}-\bV\bW\|_{F}$. The transpose invariance of the Frobenius and spectral norms and Equations \eqref{eq:Frob_norm_sin_theta_1} and \eqref{eq:Frob_norm_sin_theta_2} guarantee that terms two and three on the right hand side above can be bounded above as
\begin{align*}
    \|\bW^{\top} \bR^{(i)}(\bV^{\top}\widehat{\bV}-\bW)\|_{F}
        &\le \|\bW^{\top}\|\|\bR^{(i)}\|\|\bV^{\top}\widehat{\bV}-\bW\|_{F}
        &\le 2\|\bR^{(i)}\|\|\widehat{\bV}- \bV \bW\|_{F}^{2},\\
    \|(\widehat{\bV}^{\top}\bV-\bW^{\top})\bR^{(i)}(\bV^{\top}\widehat{\bV})\|_{F}
        &\le \|\widehat{\bV}^{\top}\bV-\bW^{\top}\|_{F} \|\bR^{(i)}\|\|\bV^{\top}\widehat{\bV}\|
        &\le 2\|\bR^{(i)}\|\|\widehat{\bV}- \bV \bW\|_{F}^{2}.
\end{align*}
On the other hand, the second term in Equation~\eqref{eq:decomp-R} can be further expanded as the sum of four matrices, namely
\begin{align*}
    \widehat{\bV}^{\top}\bE^{(i)}\widehat{\bV}
        &= \bW^{\top}\bV^{\top}\bE^{(i)} \bV \bW\\
        &\hspace{2em}+\bW^{\top}\bV^{\top}\bE^{(i)}(\widehat{\bV}- \bV \bW)\\
        &\hspace{2em}+ (\bW^{\top}\bV^{\top}-\widehat{\bV}^{\top})\bE^{(i)} \bV \bW\\
        &\hspace{2em}+
        (\bW^{\top}\bV^{\top}-\widehat{\bV}^{\top})\bE^{(i)}(\widehat{\bV}- \bV \bW).
\end{align*}
The final term in the above expansion can be bounded via
\begin{equation*}
    \|(\bW^{\top}\bV^{\top}-\widehat{\bV}^{\top})\bE^{(i)}(\widehat{\bV}- \bV \bW)\|_{F}
    \le \sqrt{d}\|\bE^{(i)}\|\|\widehat{\bV}- \bV \bW\|_{F}^{2}.
\end{equation*}
Observe that terms two and three above are simply the transpose of one another, hence it suffices to analyze the former, which satisfies the bound
\begin{equation*}
    \|\bW^{\top}\bV^{\top}\bE^{(i)}(\widehat{\bV}- \bV \bW)\|_{F}
    \le \sqrt{d}\|\bE^{(i)}\|\|\widehat{\bV}- \bV \bW\|_{F}.
\end{equation*}

Recall that we have a high probability bound on $\bE^{(i)}$ in spectral norm (see Theorem 20 in \cite{Athreya2017}), namely
\begin{equation*}
    \mathbb{P}\left(\|\bE^{(i)}\| >
    8\sqrt{\delta(\bP^{(i)})\log (n)}\right) \leq n^{-2}.
\end{equation*}
Theorem~\ref{thrm:COSIE Expectation Bound} can be leveraged together with Markov's inequality to yield that there exists an absolute constant $C>0$ such that for any $t>0$,
\begin{equation*}
    \mathbb{P}\left( \|\widehat{\bV}-\bV\bW\|_{F} > t\right)
    \le \frac{C}{t}\left(\sqrt{\frac{d}{m}}\varepsilon + \sqrt{d}\varepsilon^2\right).
\end{equation*}
Thus, there exist (sequences of) matrices $\bB, \bN \in \mathbb{R}^{d \times d}$ depending on $d$, $m$, and $n$, such that
\begin{equation}
    \widehat{\bR}^{(i)} - \bW\bR^{(i)}\bW^{\top}
    =
    \bV^{\top}\bE^{(i)}\bV + \bW \bB \bW^{\top} + \bW \bN \bW^{\top},
\end{equation}
for which
\begin{align*}
    \|\bB\|_{F}
    &\le 2\left(\sqrt{d}\|\bE^{(i)}\|\right)\|\widehat{\bV}-\bV \bW\|_{F},\\
    \|\bN\|_{F}
    &\le 4\left(\sqrt{d}\|\bE^{(i)}\| + \|\bR^{(i)}\|\right)\|\widehat{\bV}-\bV \bW\|_{F}^{2}.
\end{align*}
Put $\bH_{m}=\bW \bB \bW^{\top} + \bW \bN \bW^{\top}$; we call $\bH_{m}$ a bias term.
Recall that $$\delta(\bP^{(i)}) = \max_{v=1,\ldots,n} \sum_{u=1}^n(\bP^{(i)})_{uv},$$ and as such
\begin{equation*}
    \|\bR^{(i)}\| = \|\bV\bR^{(i)}\bV^{\top}\| = \|\bP^{(i)}\| \le \delta(\bP^{(i)}).
\end{equation*}

Setting $t$ such that $\frac{1}{t}\sqrt{\delta(\bP^{(i)})\log(n)}\left(\sqrt{\frac{d}{m}}\varepsilon + \sqrt{d}\varepsilon^2\right) \rightarrow 0$ as $m, n \rightarrow \infty$ guarantees that the Frobenius norm of the bias term $\bH_m$ vanishes in probability as $n, m \rightarrow \infty$. In particular, note that under the assumption that $(\delta(\bP^{(i)}))^{1/2}\varepsilon = O(1)$,
\begin{align*}
    \e[\|\bH_m\|_F]  & \lesssim \sqrt{d}\e[\|\bE^{(i)}\| ]\e[\|\widehat{\bV}-\bV\bW\|_F] + (\sqrt{d}\e[\|\bE^{(i)}\| + \delta(\bP^{(i)})) ]\e[\|\widehat{\bV}-\bV\bW\|_F^2]\\
    & \lesssim \sqrt{d\delta(\bP^{(i)})}\left(\sqrt{\frac{d}{m}}\varepsilon + \sqrt{d}\varepsilon^2\right)\left(1 + \sqrt{\delta(\bP^{(i)})/d}\left(\sqrt{\frac{d}{m}}\varepsilon + \sqrt{d}\varepsilon^2\right)\right)\\
    & \lesssim \frac{d}{\sqrt{m}}.
\end{align*}
\end{proof}
The central limit theorem will be complete once we show that the entries of $\bV \bE^{(i)} \bV^\top$ converge to a normal distribution as $n \rightarrow \infty$.  We address this in our remaining lemma.

\begin{lemma}\label{lem:VEV_normality}
	Let $\bSigma^{(i)}\in\real^{r\times r}$ be the matrix defined on Equation~\eqref{eq:covarianceR}.
	\begin{itemize}
		\item[a)] Under the assumptions of Theorem~\ref{thm:COSIE_CLT} part a), for any $k,l\in [d]$, $k\leq l$,	 it holds that
			\[\left(\bSigma^{(i)}_{\frac{2k + l(l-1)}{2}, \frac{2k+l(l-1)}{2}}\right)^{-1/2}(\bV^\top\bE^{(i)}\bV)_{kl} \overset{d}{\rightarrow} \mathcal{N}\left(0,1\right). \]
			where 
				\begin{equation}
			\label{eq:sigma_variance_bound}
			\left(  \frac{c_1^4}{n^2} + \frac{c_1^8}{n^4}\right)s^2(\bP^{(i)}) \leq \bSigma^{(i)}_{\frac{2k + l(l-1)}{2}, \frac{2k+l(l-1)}{2}} \leq \left(  \frac{c_2^4}{n^2} + \frac{c_2^8}{n^4}\right)s^2(\bP^{(i)})
			\end{equation}
			
		\item[b)] Under the assumptions of Theorem~\ref{thm:COSIE_CLT} part b), it holds that
		\[(\bSigma^{(i)})^{-1/2} \operatorname{vec}(\bV^\top\bE^{(i)}\bV) \overset{d}{\rightarrow} \mathcal{N}\left(\mathbf{0}_r, \bI_r\right). \]
	\end{itemize}
\end{lemma}

\begin{proof}
To obtain the asymptotic distribution of the entries of $\bV^\top\bE^{(i)}\bV$, observe that, first, the entire matrix $\bV^\top\bE^{(i)}\bV$ depends on $n$. Second, note that
\begin{align}
    (\bV^\top\bE^{(i)}\bV)_{kl} 
    & = \sum_{s=1}^n\sum_{t=1}^n\bE^{(i)}_{st}\bV_{sk}\bV_{tl} \label{eq:VEV-1}\\
    & = \sum_{s=1}^{n-1}\sum_{t=s+1}^n \bE^{(i)}_{st}\left( \bV_{sk}\bV_{tl} + \bV_{tk}\bV_{sl}\right),\label{eq:VEV-2}
\end{align}
so each of the entries of the matrix $\bV^\top\bE^{(i)}\bV$ is a sum of the independent random variables. The expected value of the above sum is
\begin{align*}
    \e[(\bV^\top\bE^{(i)}\bV)_{kl}] & = \left(\bV^\top\e[\bE^{(i)}]\bV\right)_{kl}=0.
\end{align*}
To ease notation, we write $\bSigma = \bSigma^{(i)}$. For each pair $(k,l)$ and $(k',l')$ with $k\leq l$ and $k'\leq l',$ the covariance between those corresponding entries of $\bV^\top\bE^{(i)}\bV$ is given by
\begin{align}
    \text{Cov}\left( (\bV^\top\bE^{(i)}\bV)_{kl},  (\bV^\top\bE^{(i)}\bV)_{k'l'}\right) &= \sum_{s=1}^{n-1}\sum_{t=s+1}^n\sum_{s'=1}^{n-1}\sum_{t'=s'+1}^n \Big\{\text{Cov}\left(\bE^{(i)}_{st}, \bE^{(i)}_{s't'}\right)\times \nonumber\\
    & \quad\quad\quad\quad\quad\quad (\bV_{sk}\bV_{tl} + \bV_{tk}\bV_{sl})(\bV_{s'k'}\bV_{t'l'} + \bV_{s'l'}\bV_{t'k'})\Big\}\nonumber\\
    & = \sum_{s=1}^{n-1}\sum_{t=s+1}^n\text{Cov}\left(\bE^{(i)}_{st}, \bE^{(i)}_{st}\right) (\bV_{sk}\bV_{tl} + \bV_{tk}\bV_{sl})(\bV_{sk'}\bV_{tl'} + \bV_{sl'}\bV_{tk'})\nonumber\\
    & = \sum_{s=1}^{n-1}\sum_{t=s+1}^n\bP_{st}^{(i)}(1-\bP^{(i)}_{st})\left[(\bV_{sk}\bV_{tl} + \bV_{tk}\bV_{sl})(\bV_{sk'}\bV_{tl'} + \bV_{sl'}\bV_{tk'})\right] \nonumber\\
    & = \bSigma_{\frac{2k + l(l-1)}{2}, \frac{2k'+l'(l'-1)}{2}}. \label{eq:proof-Cov-vev}
\end{align}

Recall that $s^2(\bP^{(i)})$ is defined as
$$s^2(\bP^{(i)}) = \sum_{s=1}^{n}\sum_{t=1}^n \bP^{(i)}_{st}(1- \bP^{(i)}_{st}).$$
Using Assumption~\ref{assump:Delocalization} and Equation~\eqref{eq:VEV-1}, the variance of $(\bV^\top\bE^{(i)}\bV)_{kl}$ can be bounded from below by
\begin{align}
\bSigma_{\frac{2k + l(l-1)}{2}, \frac{2k+l(l-1)}{2}} & = \operatorname{Var}\left((\bV^\top\bE^{(i)}\bV)_{kl}\right)\nonumber \\
& =\sum_{s=1}^{n}\sum_{t=1}^n \operatorname{Var}\left(\bE^{(i)}_{st}\right)\left(\bV_{sk}\bV_{tl}\right)^2 + 2\sum_{s=1}^{n-1}\sum_{t=s+1}^n\operatorname{Cov}\left(\bE^{(i)}_{st}, \bE^{(i)}_{ts}\right)\left(\bV_{sk}\bV_{tl}\right)^2\left(\bV_{sl}\bV_{tk}\right)^2\nonumber\\
& = 
\sum_{s=1}^{n}\sum_{t=1}^n \bP^{(i)}_{st}(1- \bP^{(i)}_{st})\left( \bV_{sK}\bV_{tl}\right)^2 + 2\sum_{s=1}^{n-1}\sum_{t=s+1}^n\bP^{(i)}_{st}(1- \bP^{(i)}_{st})\left(\bV_{sk}\bV_{tl}\right)^2\left(\bV_{sl}\bV_{tk}\right)^2\nonumber\\
& \geq \left(  \frac{c_1^4}{n^2} + \frac{c_1^8}{n^4}\right)s^2(\bP^{(i)})\nonumber\\
& = \omega\left(\frac{1}{n^2}\right),\label{eq:rate-variance-P}
\end{align}
where the last equality follows from Assumption~\ref{assump:variance_P}.  The same calculations allow to obtain a similar upper bound:
\begin{equation*}
    \bSigma_{\frac{2k + l(l-1)}{2}, \frac{2k+l(l-1)}{2}} \leq \left(\frac{c_2^4}{n^2} + \frac{c_2^8}{n^4}\right)s^2(\bP^{(i)}).
\end{equation*}
This  shows Equation~\eqref{eq:sigma_variance_bound}.

To prove part a), define $\bF^{(i,n)}$ as
$$ \bF_{st}^{(i,n)}:= \bE^{(i)}_{st}\left( \bV_{sk}\bV_{tl} + \bV_{tk}\bV_{sl}\right),$$
which is  bounded above, because of Assumption~\ref{assump:Delocalization}, by
\begin{equation}
    \left|\bF_{st}^{(i,n)}\right| \leq \frac{2c_2^2}{n}. \label{eq:bound-E_st}
\end{equation}
Equations~\eqref{eq:bound-E_st} and \eqref{eq:rate-variance-P} imply that for any $\epsilon>0$ and sufficiently large $n$, 
$$\left|\bF_{st}^{(i,n)}\right| < \epsilon \left[\operatorname{Var}\left((\bV^\top\bE^{(i)}\bV)_{kl}\right)\right]^{1/2},$$
for all $s,t$ with $1 \leq s \leq n, 1 \leq t \leq n$. Hence, for all $n$ sufficiently large, 
\begin{equation*}
    \frac{1}{\operatorname{Var}\left((\bV^\top\bE^{(i)}\bV)_{kl}\right)}\sum_{s=1}^{n-1}\sum_{t=s+1}^n\e\left[|\bF_{st}^{(i,n)}|^2\mathbbm{1}\left\{|\bF_{st}^{(i,n)}| > \epsilon \left[\operatorname{Var}\left((\bV^\top\bE^{(i)}\bV)_{kl}\right)\right]^{1/2}\right\}\right]=0.
\end{equation*} 
This shows that the Lindeberg's condition is satisfied, and by the Lindeberg-Feller Central Limit Theorem applied to Equation~\eqref{eq:VEV-2}, we have that
\[\left(\bSigma_{\frac{2k + l(l-1)}{2}, \frac{2k+l(l-1)}{2}}\right)^{-1/2}(\bV^\top\bE^{(i)}\bV)_{kl} \overset{d}{\rightarrow} \mathcal{N}\left(0,1\right). \]

To prove part b), let $\bs\in\real^r$ be an arbitrary vector. We show that for any $\bs$,
\begin{equation}
\bs^\top \bSigma^{-1/2} \text{vec}(\bV^\top\bE^{(i)}\bV) \overset{d}{\rightarrow} \bs^\top\mathcal{N}(0, \bI_r)=\mathcal{N}(0, \bs^\top\bs).   
\label{eq:asymp-normality-bS}
\end{equation}
The above quantity can be expressed as a sum of independent random variables
\begin{align*}
    \bs^\top \bSigma^{-1/2} \text{vec}(\bV^\top\bE^{(i)}\bV)    & = \sum_{k=1}^d\sum_{l=k}^d\bs_{\frac{2k+l(l-1)}{2}}\left\{ \sum_{k'=1}^d\sum_{l'=k'}^d\bSigma^{-1/2}_{\frac{2k+l(l-1)}{2}, \frac{2k'+l'(l'-1)}{2}}\left[(\bV^\top\bE^{(i)}\bV)_{k'l'}\right] \right\}\\
    & = \sum_{s=1}^{n-1}\sum_{t=s+1}^n\bE^{(i)}_{st}\sum_{k=1}^d\sum_{l=k}^d\bs_{\frac{2k+l(l-1)}{2}} \sum_{k'=1}^d\sum_{l'=k'}^d\bSigma^{-1/2}_{\frac{2k+l(l-1)}{2}, \frac{2k'+l'(l'-1)}{2}}\left(\bV_{sk'}\bV_{tl'} + \bV_{tk'}\bV_{sl'}\right).
\end{align*}
Define
$$ \bG_{st}^{(i,n)}:=\bE^{(i)}_{st}\sum_{k=1}^d\sum_{l=k}^d\bs_{\frac{2k+l(l-1)}{2}} \sum_{k'=1}^d\sum_{l'=k'}^d\bSigma^{-\frac{1}{2}}_{\frac{2k+l(l-1)}{2}, \frac{2k'+l'(l'-1)}{2}}\left(\bV_{sk'}\bV_{tl'} + \bV_{tk'}\bV_{sl'}\right).$$
The magnitude of each of these summands can be bounded above as
\begin{align}
    |\bG_{st}^{(i,n)}|   &\leq   |\bE^{(i)}_{st}|  \sum_{k=1}^d\sum_{l=k}^d\left|\bs_{\frac{2k+l(l-1)}{2}}\right| \sum_{k'=1}^d\sum_{l'=k'}^d \left|\bSigma^{-\frac{1}{2}}_{\frac{2k+l(l-1)}{2}, \frac{2k'+l'(l'-1)}{2}}\right| \left|\bV_{sk'}\bV_{tl'} + \bV_{tk'}\bV_{sl'}\right|\nonumber\\
    & \leq \frac{2c_2^2}{n} \sum_{k=1}^d\sum_{l=k}^d\left|\bs_{\frac{2k+l(l-1)}{2}}\right| \sum_{k'=1}^d\sum_{l'=k'}^d \left|\bSigma^{-\frac{1}{2}}_{\frac{2k+l(l-1)}{2}, \frac{2k'+l'(l'-1)}{2}}\right| \nonumber\\
    & \leq \frac{2c_2^2}{n} \left\|\bs\right\|_2
    \left\{\sum_{k=1}^d\sum_{l=k}^d \left(\sum_{k'=1}^d\sum_{l'=k'}^d \left|\bSigma^{-\frac{1}{2}}_{\frac{2k+l(l-1)}{2}, \frac{2k'+l'(l'-1)}{2}}\right|\right)^2\right\}^{\frac{1}{2}} \nonumber\\
    & \leq \frac{2c_2^2}{n} \left\|\bs\right\|_2
    \left\{\sum_{k=1}^d\sum_{l=k}^d \sum_{k'=1}^d\sum_{l'=k'}^d r\left|\bSigma^{-\frac{1}{2}}_{\frac{2k+l(l-1)}{2}, \frac{2k'+l'(l'-1)}{2}}\right|^2\right\}^{\frac{1}{2}} \nonumber\\
    & = \frac{2\sqrt{r}c_2^2}{n} \left\|\bs\right\|_2\left\|\bSigma^{-1/2}\right\|_F \nonumber\\
     & \leq \frac{2\sqrt{r}c_2^2}{n} \left\|\bs\right\|_2 (\sqrt{r}|\lambda_{\max}(\bSigma^{-1/2})|\nonumber\\
    & \leq  \frac{2rc_2^2\left\|\bs\right\|_2 }{n|\lambda_{\min}(\bSigma)|^{1/2}} = o( \left\|\bs\right\|_2),
    \label{eq:proof-Fbound}
\end{align}
where the last equation follows by Assumption~\ref{assump:score_lambdamin}. The bilinearity of the covariance and Equation~\eqref{eq:proof-Cov-vev} give that
\begin{equation}
    \operatorname{Var}\left(\bs^\top  \bSigma^{-1/2} \text{vec}(\bV^\top \bE^{(i)}\bV)\right) = \bs^\top \bs.
    \label{eq:proof-variance_sumS}
\end{equation}
Equations~\eqref{eq:proof-variance_sumS} and \eqref{eq:proof-Fbound} imply that for any $\epsilon>0$ and sufficiently large $n$, 
$$\left|\bG_{st}^{(i,n)}\right| < \epsilon \operatorname{Var}\left(\bs^\top  \bSigma^{1/2} \text{vec}(\bV^\top \bE^{(i)}\bV)\right)^{1/2}  = \epsilon\|\bs\|_2,$$
uniformly for all $s,t$ with $1 \leq s \leq n, 1 \leq t \leq n$. Hence, for all $n$ sufficiently large, 
\begin{equation*}
    \frac{1}{\|\bs\|_2^2}\sum_{s=1}^{n-1}\sum_{t=s+1}^n\e\left[|\bG_{st}^{(i,n)}|^2\mathbbm{1}\left\{|\bG_{st}^{(i,n)}| > \epsilon \|\bs\|_2\right\}\right]=0.
\end{equation*} 
This shows that Lindeberg's condition is satisfied for all $n$ sufficiently large, and by the Lindeberg-Feller Central Limit Theorem, we have that the convergence in distribution given by Equation~\eqref{eq:asymp-normality-bS} holds. Since this convergence is valid for any $\bs\in\real^r$, the  Cram\'er-Wold theorem implies that
\[\bSigma^{-1/2}\operatorname{vec}(\bV^\top \bE^{(i)}\bV) \overset{d}{\rightarrow} \mathcal{N}\left(\mathbf{0}_r,\bI_r\right). \]

\end{proof}

\begin{proof}[Proof of Theorem~\ref{thm:COSIE_CLT}]
The results follow by combining the decomposition given in the left hand side of Equation~\ref{eq:R-decomposition} in Lemma~\ref{lem:COSIE_R_decomp} and the asymptotic distribution for the right hand side obtained in Lemma~\ref{lem:VEV_normality}.
\end{proof}

\begin{proof}[Proof of Corollary~\ref{corollary:variance-VEV}]
By Lemma~\ref{lem:COSIE_R_decomp}, there exist a sequence of orthogonal matrices $\bW$ such that
$$\widehat{\bR}^{(i)} - \bW^\top \bR^{(i)} \bW+\bH_m^{(i)} = \bW\bV^\top \bE^{(i)}\bV\bW.$$
The term of the right hand side converges to a multivariate normal distribution by Lemma~\ref{lem:VEV_normality}, and using the same arguments as in Equation~\eqref{eq:rate-variance-P}, the variance of its entry $k,l\in[d], k\leq l,$ can be bounded from above as
\begin{align*}
    \operatorname{Var}( [(\bV\bW)^\top \bE^{(i)}(\bV\bW)]_{kl}) = &\sum_{s=1}^{n}\sum_{t=1}^n \operatorname\bP^{(i)}_{st}(1- \bP^{(i)}_{st})\left((\bV\bW)_{sk}(\bV\bW)_{tl}\right)^2 +\\
    & 2\sum_{s=1}^{n-1}\sum_{t=s+1}^n\bP^{(i)}_{st}(1- \bP^{(i)}_{st})\left((\bV\bW)_{sk}(\bV\bW)_{tl}\right)^2\left((\bV\bW)_{sl}(\bV\bW)_{tk}\right)^2\nonumber\\
    & \leq \max_{s,k}\left((\bV\bW)_{sk}^4 + 2(\bV\bW)_{sk}^8\right)s^2(\bP^{(i)}).
\end{align*}
Finally, observe that $|(\bV\bW)_{sk}| \leq \|\bV_{s\cdot}\|_2\|\bW_{\cdot k}\|_2 \leq \sqrt{\frac{{d}}{n}}$, which completes the proof.
\end{proof}

\end{document}